\documentclass{article}

\usepackage{amsmath}
\usepackage{amssymb}
\usepackage{amsthm}
\usepackage{amscd}
\usepackage{bbm}
\usepackage{enumitem}
\usepackage{float}
\usepackage{authblk}

\usepackage{tikz-cd}
\usetikzlibrary{decorations.markings}
\tikzset{->-/.style={decoration={
  markings,
  mark=at position #1 with {\arrow{>}}},postaction={decorate}}}
\tikzset{->-/.default=0.5}

\usepackage{pgfplots}

\pgfplotsset{compat=1.10}
\usepgfplotslibrary{fillbetween}
\usetikzlibrary{patterns}

\newtheorem{thm}{Theorem}[section]
\newtheorem{prop}[thm]{Proposition}
\newtheorem{lem}[thm]{Lemma}
\newtheorem{cor}[thm]{Corollary}

\theoremstyle{definition}
\newtheorem{dfn}[thm]{Definition}

\theoremstyle{remark}
\newtheorem{rem}[thm]{Remark}
\newtheorem{example}[thm]{Example}

\newcommand{\C}{\mathbb{C}}
\newcommand{\R}{\mathbb{R}}
\newcommand{\T}{\mathbb{T}}
\newcommand{\Z}{\mathbb{Z}}

\newcommand{\CC}{\mathbb{C}}
\newcommand{\HH}{\mathbb{H}}
\newcommand{\NN}{\mathbb{N}}

\newcommand{\RR}{\mathbb{R}}
\newcommand{\TT}{\mathbb{T}}
\newcommand{\ZZ}{\mathbb{Z}}

\newcommand{\G}{\mathcal{G}}
\newcommand{\U}{\mathfrak{U}}
\newcommand{\V}{\mathfrak{V}}
\newcommand{\W}{\mathfrak{W}}

\newcommand{\pt}{\mathrm{pt}}

\newcommand{\vect}[1]{\boldsymbol{#1}}

\title{`Real' gerbes and Dirac cones of topological insulators}

\author[1]{Kiyonori Gomi\thanks{kgomi@math.titech.ac.jp}}
\author[2]{Guo Chuan Thiang\thanks{guochuanthiang@bicmr.pku.edu.cn, ORCID 0000-0003-0268-0065}}
\affil[1]{Department of Mathematics, Tokyo Institute of Technology}
\affil[2]{Beijing International Center for Mathematical Research, Peking University}

\date{\today}

\begin{document}

\maketitle

\begin{abstract}
A time-reversal invariant topological insulator occupying a Euclidean half-space determines a `Quaternionic' self-adjoint Fredholm family. We show that the discrete spectrum data for such a family is geometrically encoded in a non-trivial `Real' gerbe. The gerbe invariant, rather than a na\"{i}ve counting of Dirac points, precisely captures how edge states completely fill up the bulk spectral gap in a topologically protected manner.
\end{abstract}

\setcounter{tocdepth}{2}
\tableofcontents

%%%%%%%%%%%%%%%%%%%%%%%%%%%%%%%%%%%%%%%%%%%%%%%%
%%%%%%%%%%%%%%%%%%%%%%%%%%%%%%%%%%%%%%%%%%%%%%%%

\section{Introduction}
Traditionally, (bundle) gerbes are ``higher'' geometric objects which realise the third integral cohomology of a space $X$ \cite{Hi,Murray,MS}, much like complex line bundles realise the second cohomology via the Chern class. They have physical applications in quantum field theory and string theory \cite{CJM,CMM,Gaw-R}, and more recently, in topological condensed matter physics \cite{C-T}. In the latter, a so-called ``Fermi gerbe'' is constructed from the self-adjoint Fredholm family describing a semimetal or insulator occupying a Euclidean half-space. The discrete spectra of the family correspond to edge states, and are encoded in the Fermi gerbe. Remarkably, the gerbe invariant (Dixmier--Douady class) protects the essential spectral gap-filling property of the edge states, leading to stability of the Fermi surface, generalising the role played by spectral flows in \cite{Thiang-sf}. Intuitively, discrete spectra may connect the positive and negative essential spectra of a self-adjoint Fredholm family in a robust way which goes beyond the ordinary spectral flow along loops in the family --- the Fermi gerbe encodes such a ``higher'' spectral flow. However, non-trivial examples only occur in rather high dimensions, which may be difficult to simulate in actual experiments. 

In the same way, a `Real' version of gerbes can realise the equivariantly twisted third cohomology $H_\pm^3(X)$ (see Appendix \ref{appendix:Borel.Cech} for the meaning of this cohomology). Various mathematical definitions are available, typically motivated by orientifold constructions in string theory \cite{GSW,SSW,HMSV1,HMSV2}. In contrast to ordinary gerbes, non-trivial examples of `Real' gerbes exist in low dimensions (see \S\ref{sec:Real.gerbe.examples}).

In a completely independent development in physics, 2D and 3D time-reversal invariant/symmetric (TRS) insulators (also called ``Class AII'' insulators) were theoretically predicted to be classified in the bulk by a subtle $\ZZ_2$ topological invariant \cite{KM,FKM}. Mathematically, this invariant can be thought of as a `Quaternionic' $K$-theory class \cite{FM} or characteristic class \cite{D-G-AII,D-G-cohom}, as reviewed in \S\ref{sec:FKMM.invariant}. The topological magneto-electric effect was proposed as an observable feature of 3D TRS topological insulators \cite{QHZ}, but this has yet to be observed directly. Instead, experimental successes have been found by probing the surface physics. For instance, the ``smoking gun'' mod-2 ``Dirac cone edge states'' which fill up the insulating bulk's spectral gap (see Fig.\ \ref{fig:phase.function3} for a sketch) were observed in \cite{Hsieh}. For 3D TRS topological insulators which have a single Dirac cone of surface states, Landau quantization of these unusual 2D surface states was even observed when a TRS-breaking magnetic field was applied \cite{Cheng, Hanaguri}. We refer the reader to \cite{LP} for a detailed and up-to-date discussion on experimental investigations of TRS topological insulators, as well as a $K$-theoretic perspective of the surface physics. Later on, topological semimetals were discovered, and as they are typically time-reversal invariant, the formulation of their edge ``Fermi arc'' invariants also requires the `Quaternionic' characteristic class theory \cite{TSG}. From both theoretical and experimental viewpoints, it is therefore important to understand carefully the topological invariants of TRS physical systems \emph{in the presence of a boundary}.  

The main aim of this paper is to bridge two seemingly unrelated developments --- `Real' gerbes and Dirac cone edge states of topological insulators --- with concrete examples, computations, and applications. We show that time-reversal invariant topological insulators occupying a Euclidean half-space are described by a $\ZZ_2$-equivariant (or `Quaternionic') self-adjoint Fredholm family (Prop.\ \ref{prop:H.equivariance}), whose discrete spectrum (or ``edge states'') data assembles into a `Real' Fermi gerbe (\S\ref{sec:Real.Fermi}). Our notion of `Real' gerbe, detailed in \S\ref{sec:real.gerbes} (see Remark \ref{rem:Gao.Hori} for some relations to existing definitions in the literature), is particularly well-adapted for this application. The `Real' gerbe invariant is computed to be non-trivial, implying spectral gaplessness of the topological insulator on a half-space. This constitutes an explicit statement and proof of the bulk-boundary correspondence in the time-reversal invariant setting (Theorem \ref{thm:BBC}); see Remark \ref{rem:BBC.literature} for a discussion of existing literature on this and some gaps therein. In physical terms, the edge states of the half-space topological insulator inherit the bulk topological insulator's abstract non-trivial $\ZZ_2$-invariants. Consequently, the half-space topological insulator enjoys ``topological protection'' of its spectral gaplessness.

While such a ``bulk-edge correspondence'' is often taken for granted in physical practice, our mathematical formulation in terms of gerbes is new and of independent interest. Specifically, we provide a $K$-theory boundary invariant for topological insulators (Definition \ref{defn:boundary.K.class}), which maps via Eq.\ \eqref{eqn:KR.to.DD} to the `Real' Fermi gerbe invariant (an element of $H^3_\pm(X)/H^3_\pm(\pt)$). The gerbe invariant provides, for the first time, a precise topological invariant which properly ``counts'' how the edge states of time-reversal invariant topological insulators fill up the bulk spectral gap. 
Here, \emph{we stress that a na\"{i}ve mod-2 counting of Dirac cone vertices (or ``Dirac points'') generally fails}, as illustrated in Fig.\ \ref{fig:phase.function2} and Fig.\ \ref{fig:phase.function}, and discussed in \S\ref{sec:counting.Dirac}. 
Furthermore, the `Real' Fermi gerbe construction is a mathematically general one, whose applicability is not restricted to the Class AII topological insulator setup (in which the self-adjoint Fredholm family is parametrized by a torus). Next, for the crucial transfer maps from bulk to boundary topological invariants, we will provide some explicit Fredholm formulations of \emph{real} $K$-theory connecting maps, which are not readily available in the literature. Upon doing so, we learn the sense in which the `Real' Fermi gerbe invariant measures ``higher'', or ``secondary'', spectral flow for an equivariant self-adjoint Fredholm family. 

\medskip

\noindent
{\bf Outline.} 
In \S\ref{sec:bulk.invariant} and \S\ref{sec:gap.filling}, we provide an operator $K$-theory (and, equivalently, topological $KR$) proof of the time-reversal invariant bulk-boundary correspondence. To understand the $K$-theoretic boundary invariant, and therefore the bulk-boundary correspondence, in a more descriptive way, we need a notion of `Real' gerbes, which are defined carefully in \S\ref{sec:real.gerbes}, and illustrated with examples in \S\ref{sec:Real.gerbe.examples}. The passage from $KR$-theory to `Real' gerbes and the relation to ``higher'' spectral flow is explained in \S\ref{sec:Real.gerbe.spectral.flow}, and explicit computations for 2D and 3D topological insulators are provided in \S\ref{sec:non.triviality}.

\section{Time-reversal invariant topological insulator: $K$-theory}\label{sec:bulk.invariant}
Let us first make precise the (bulk) ``$\ZZ_2$-topology'' of time-reversal invariant topological insulators, first found by ad-hoc means by physicists Fu--Kane--Mele \cite{KM,FKM}. We provide the operator algebraic approach in some detail, because the time-reversal invariant \emph{bulk-edge correspondence} requires the use of index morphisms at the level of operator $K$-theory of real $C^*$-algebras. ``Topological'' approaches, which may be more familiar to the reader, are recalled later, in \S\ref{sec:KR.formulation} and \S\ref{sec:FKMM.invariant}.

\medskip

\subsection{Tight-binding Hamiltonians and time-reversal}\label{sec:tight.binding}
To begin with, the \emph{tight-binding} Hilbert space is $V\otimes\ell^2(\ZZ^d)$,
where $V$ is a finite-dimensional \emph{quaternionic} Hilbert space. That is, $V$ is an even-dimensional complex Hilbert space equipped with an antiunitary operator $\Theta$ such that $\Theta^2=-1$. Such an operator $\Theta$ is called a \emph{quaternionic structure} on $V$. Note that multiplication by $i,\Theta$ and $i\Theta$ are mutually anticommuting real-linear operations squaring to $-1$, so that $V$ has a right quaternionic ``scalar multiplication''.

A \emph{quaternionic basis} for $V$ is a set of vectors $\{\mathsf{e}_i\}_{i=1,\ldots,n}$ such that the set $\{\mathsf{e}_1,\Theta\mathsf{e}_1,\ldots,\mathsf{e}_n,\Theta\mathsf{e}_n\}$ is an orthonormal basis (over $\CC$) for $V$. In such a basis, $\Theta$ has the standard form,
\begin{equation}
\Theta=
\left(
\begin{array}{rr|c|rr}
0 & -1 & & & \\
1 & 0 & & & \\
\hline
 & & \ddots & & \\
\hline
 & & & 0 & -1 \\
 & & & 1 & 0
\end{array}
\right)
\circ \kappa=:J\circ\kappa,\label{eqn:standard.quaternionic}
\end{equation}
where $\kappa$ denotes complex conjugation in that basis.
Sometimes, the rearranged basis $\{\mathsf{e}_1,\ldots,\mathsf{e}_n,\Theta\mathsf{e}_1,\ldots,\Theta\mathsf{e}_n\}$ is used instead, and $\Theta$ has the form
\begin{equation*}
\Theta=\begin{pmatrix} 0 & -\mathbbm{1}_n \\
\mathbbm{1}_n & 0 \end{pmatrix}\circ \kappa.
\end{equation*}
We will use the same symbol $\Theta$ for the extension $\Theta\otimes\kappa$ acting on the tight-binding Hilbert space $V\otimes\ell^2(\ZZ^d)$. 

The physical meaning of $\Theta$ is \emph{fermionic time-reversal}. Also, $\ell^2(\ZZ^d)$ is the regular representation space for an abelian lattice $\ZZ^d$ of (Euclidean space) translations. The generating unitary translation operators on $\ell^2(\ZZ^d)$ are denoted $T_i, i=1,\ldots,d$.

\emph{We are interested in Hamiltonians $H=H^*\in \mathcal{B}(V\otimes\ell^2(\ZZ^d))$ which commute with the translations $T_i$, as well as the time-reversal operator $\Theta$.} The first condition means that $H$ belongs to the commutant of the $*$-algebra generated by $\mathbbm{1}_{2n}\otimes T_i$, which is $M_{2n}(\CC)\otimes W^*(\ZZ^d)$, where $W^*(\ZZ^d)$ denotes the group von Neumann algebra. It is physically sensible to restrict to finite-range Hamiltonians, or at least those that can be norm-approximated by finite-range operators --- these might be called \emph{local} Hamiltonians, and \emph{we will assume this throughout}. 

The reduced group $C^*$-algebra $C^*_r(\ZZ^d)$ is the $C^*$-algebra generated by the translations $T_i$ acting in the regular representation on $\ell^2(\ZZ^d)$. The desired local Hamiltonians are approximated by (matrices  of) of polynomials in the $T_i, T_j^*$, so they belong to the $C^*$-algebra $M_{2n}(\CC)\otimes C^*_r(\ZZ^d)$. The second requirement, that $H$ commutes with $\Theta$, means that $H$ belongs to the distinguished real $C^*$-subalgebra, $M_n(\HH)\otimes_\RR C^*_{r,\RR}(\ZZ^d)\subset M_{2n}(\CC)\otimes C^*_r(\ZZ^d)$ which commutes with $\Theta$, as explained in Example \ref{ex:translation.quaternion.algebra} below.

\subsection{$C^*$-algebras of quaternionic-linear operators}
\begin{dfn}\label{dfn:real.structure}
A \emph{real structure} $\flat$ on a complex $C^*$-algebra $B$ is an \emph{antilinear} automorphism (preserves multiplication and $*$) which squares to the identity.
\end{dfn}
The $\flat$-fixed subset of $B$,
\begin{equation*}
B^\flat:=\{a\in B\,:\, a^\flat=a\},
\end{equation*}
is a \emph{real} $C^*$-algebra. 

In reverse, the complexification of a real $C^*$-algebra $A$ is a complex $C^*$-algebra $A^\CC:=\CC\otimes_\RR A$ with canonical real structure given by complex conjugation $\kappa$ on the $\CC$ factor. Here, the adjoint on $A^\CC$ is that on $A$ extended complex-antilinearly. Note that complexifying $B^\flat$ recovers $B$.

\medskip

\begin{rem}
A real structure $\flat$ on a complex $C^*$-algebra gives rise to a \emph{linear} \emph{antiautomorphism} ($*$-preserving but multiplication reversing) $\iota=\flat\circ*$ which squares to the identity. In terms of $\iota$, the real subalgebra is specified by the condition $a^\iota=a^*$
(this convention is adopted in \cite{BL}). Given $\iota$, one obtains a real structure in the sense of Definition \ref{dfn:real.structure} by taking  
$\flat=\iota\circ *$. 
\end{rem}

There are equivalences of categories between (i) real $C^*$-algebras, (ii) complex $C^*$-algebras with real structure $\flat$, and (iii) complex $C^*$-algebras with linear antiautomorphism $\iota$. The reality condition can be written as
\begin{equation*}
a=(a^\iota)^*=a^\flat.
\end{equation*}

\medskip

\begin{example}[Commutative algebras]
For $X$ a locally compact Hausdorff space, let $C_0(X)$ be the commutative complex $C^*$-algebra of complex-valued functions on $X$ vanishing at infinity. The $*$-operation is just pointwise complex conjugation of functions, $\kappa$. An \emph{involution} on $X$ is an order-2 homeomorphism $\iota:X\rightarrow X$. Via pullback, an involution induces a linear (anti)automorphism on $C(X)$, also denoted $\iota$,
\begin{equation*}
f^\iota:=\iota(f)(x)=f(\iota(x)).
\end{equation*}
The real structure on $C(X)$ associated to $\iota$ is $\flat=\iota\circ\kappa$, i.e.,
\begin{equation*}
f^\flat(x)=\overline{f(\iota(x))}.
\end{equation*}
We write $C(X,\iota)$ for the commutative complex $C^*$-algebra $C(X)$ equipped with the real structure $\flat=\iota\circ\kappa$ associated to the involution $\iota$.

Generally speaking, there is an antiequivalence between the category of (non)unital commutative real $C^*$-algebras, and the category of (locally) compact, Hausdorff spaces with involution. This is the real version of the Gelfand--Naimark theorem, due to Arens--Kaplansky (\cite{AK}, Theorem 9.1).
\end{example}

\medskip

\begin{example}[Abelian group algebras]\label{ex:group.algebra.Z}
The complex (reduced) group $C^*$-algebra of $\ZZ$, denoted $C^*_r(\ZZ)$, is isomorphic via Fourier transform to $C(\TT)$.  Let $\iota_{\rm flip}:z\mapsto \overline{z}$ be the ``flip'' involution on $\TT$ regarded as the unit complex numbers. Then $C(\TT,\iota_{\rm flip})$ is a complex $C^*$-algebra with real structure $\iota_{\rm flip}\circ\kappa$. The real subalgebra $C(\TT)^{\iota_{\rm flip}\circ\kappa}$ is identified via Fourier transform with $C^*_{r,\RR}(\ZZ)$, the \emph{real} (reduced) group $C^*$-algebra of $\ZZ$. That is, Fourier transform converts pointwise complex-conjugation $\kappa$ on functions $\ZZ\rightarrow\CC$, into $\iota_{\rm flip}\circ\kappa$ on functions $\TT\rightarrow \CC$. 

Notice that $\TT$ is the (complex) character space, or Pontryagin dual, of $\ZZ$ (for $z\in \TT$, take the character $n\mapsto z^n$), and that $\iota_{\rm flip}$ is the operation of complex-conjugating characters. The two $\iota_{\rm flip}$-fixed characters, $\pm 1$, are precisely the real characters of $\ZZ$. Similarly, the character space for $\ZZ^d$ is $\TT^d$ (called the \emph{Brillouin torus} in physics), and the induced involution is just $\iota_{\rm flip}:(z_1,\ldots z_d)\mapsto (\overline{z_1},\ldots,\overline{z_d})$ (reusing the notation), which has $2^d$ fixed points. We have $C^*_{r,\RR}(\ZZ^d)\otimes_\RR \CC\cong C(\TT^d,\iota_{\rm flip})$ as complex $C^*$-algebras with real structure.
\end{example}

\medskip

\begin{example}
For a non-commutative example, take $M_2(\CC)$ with $*$-operation being the Hermitian adjoint. Consider the linear antiautomorphism $\mathrm{t}$ given by the transpose map. The associated real structure is $\mathrm{t}\circ *=\kappa$, i.e.\ entry-wise complex conjugation. The real $C^*$-algebra fixed under $\kappa$ is of course just $M_2(\RR)$.
\end{example}

\medskip

\begin{example}[Quaternions]\label{ex:quaternion}
Let $\Theta=\begin{pmatrix} 0 & -1 \\ 1 & 0 \end{pmatrix}\circ\kappa$ be the standard quaternionic structure on $\CC^2$, which defines a real structure on $M_2(\CC)$ by conjugation,
\begin{equation*}
\flat={\rm Ad}_\Theta
:\begin{pmatrix}a & b \\ c & d\end{pmatrix}\mapsto \begin{pmatrix}\bar{d} & -\bar{c} \\ -\bar{b} & \bar{a}\end{pmatrix}.
\end{equation*}
Thus the corresponding real subalgebra of $M_2(\CC)$ comprises the ${\rm Ad}_\Theta$-invariant operators, i.e.\ the \emph{quaternionic-linear} operators. Explicitly, they are the matrices of the form
\begin{equation*}
q=\begin{pmatrix}
p_0+ip_1 & -p_2-ip_3 \\ p_2-ip_3 & p_0-ip_1
\end{pmatrix},\qquad p_0, p_1, p_2, p_3\in\RR,
\end{equation*}
which we recognise as the $2\times 2$ complex matrix representation of the quaternion algebra $\HH$. 
Complexifying the real $C^*$-algebra $\HH$ recovers $\HH\otimes_\RR\CC\cong M_2(\CC)$.
Note that the standard $*$-operation on $M_2(\CC)$ (conjugate-transpose) restricts on $\HH$ to the usual quaternion conjugation $q\mapsto \bar{q}\leftrightarrow (p_0,-p_1,-p_2,-p_3)$. 
For the associated linear antiautomorphism, we adopt the ``sharp-notation'' (taken from \cite{BL}),
\begin{equation}
\sharp=\flat\circ*:\begin{pmatrix}a & b \\ c & d\end{pmatrix}\mapsto
\begin{pmatrix}d & -b \\ -c & a\end{pmatrix}.\label{eqn:sharp}
\end{equation}

\end{example}

\medskip

\begin{example}[Quaternion matrices]\label{ex:quaternion.higher}
The matrix algebra version of Example \ref{ex:quaternion} proceeds as follows.
Let $\Theta$ be the standard quaternionic structure on $\CC^{2n}$, Eq.\ \eqref{eqn:standard.quaternionic}. 
Conjugation by $\Theta$ provides a real structure on $M_{2n}(\CC)$, with real subalgebra comprising the quaternionic-linear operators, i.e.\ $M_n(\HH)$. It is convenient to write $M_{2n}(\CC)\cong M_n(\CC)\otimes M_2(\CC)$, and
\begin{equation*}
\Theta=\left(\mathbbm{1}_n\otimes\begin{pmatrix} 0 & -1 \\ 1 & 0 \end{pmatrix}\right)\circ\kappa.
\end{equation*}
The linear antiautomorphism associated to ${\rm Ad}_\Theta$ is then
\begin{equation}
{\rm Ad}_\Theta\circ*=\mathrm{t}\otimes\sharp\label{eqn:sharp.matrix}
\end{equation}
where $\mathrm{t}$ is the transpose operation on the $M_n(\CC)$ factor, and $\sharp$ is Eq.\ \eqref{eqn:sharp}. \emph{As a convenient shorthand, we will abbreviate $\mathrm{t}\otimes\sharp$ to $\sharp$ when we pass to quaternion matrix algebras as above. }

\end{example}

\begin{rem}
The sharp operation on $M_{2n}(\CC)$ can be restricted to an involution on the unitary subgroup ${\rm U}(2n)\subset M_{2n}(\CC)$, or on the special unitary subgroup  ${\rm SU}(2n)$,
\begin{equation}
\sharp:u\mapsto \Theta u^*\Theta^{-1}.\label{eqn:sharp.unitary.formula}
\end{equation}
Recall that $\Theta=J\circ\kappa$ where $J$ is the special unitary matrix $\begin{pmatrix} 0 & -\mathbbm{1}_n \\
\mathbbm{1}_n & 0 \end{pmatrix}$. Multiplication by $J$ is a homeomorphism on ${\rm (S)U}(2n)$. Under this homeomorphism, it is straightforward to verify that $\sharp$ transforms into the ``minus-transpose'' operation,
\begin{equation}\label{eqn:sharp.equivalent}
J\circ\sharp\circ J^{-1}:u\mapsto -u^{\rm t}.
\end{equation}
So Eq.\ \eqref{eqn:sharp.unitary.formula} and Eq.\ \eqref{eqn:sharp.equivalent} define equivalent involutions on ${\rm (S)U}(2n)$, in the sense that the resulting involutive spaces are equivariantly homeomorphic. We can carry either involution to the direct limits
\begin{equation*}
{\rm U}(\infty)=\varinjlim {\rm U}(2n),\qquad {\rm SU}(\infty)=\varinjlim {\rm SU}(2n).
\end{equation*}
\end{rem}

\begin{example}[Quaternionifying real $C^*$-algebras]\label{ex:matrix.algebra}
Further generalising Example \ref{ex:quaternion.higher}, let $B=\CC\otimes_\RR A$ be a complex $C^*$-algebra with real structure $\kappa_B$ and real subalgebra $A$.

Write $A^\HH$ for $\HH\otimes_\RR A$. Then $M_n(A^\HH)$ complexifies into
\begin{align*}
\CC\otimes_\RR M_n(A^\HH)&\cong \CC\otimes_\RR (M_n(\HH)\otimes_\RR A)\\
&\cong (\CC\otimes_\RR M_n(\HH))\otimes(\CC\otimes_\RR A)\\
&\cong M_{2n}(\CC)\otimes A^\CC \cong M_{2n}(B).
\end{align*}
In reverse, taking the real structure on $M_{2n}(B)\cong M_{2n}(\CC)\otimes B$ to be ${\rm Ad}_\Theta\otimes \kappa_B$, we recover $A^\HH$ as the real subalgebra.

If the associated linear antiautomorphism on $B$ is denoted $\iota_B=\kappa_B\circ *$, then on $M_{2n}(B)\cong M_{2n}(\CC)\otimes B$, the associated linear antiautomorphism is $\sharp\otimes\iota_B$. 
\end{example}

\medskip

\begin{example}\label{ex:translation.quaternion.algebra}
Elaborating on Example \ref{ex:matrix.algebra}, take $B=C^*_r(\ZZ^d)$, and $A=C^*_{r,\RR}(\ZZ^d)$. That is, we want to quaternionify the abelian Example \ref{ex:group.algebra.Z}. In this case, $M_n((C^*_{r,\RR}(\ZZ^d))^\HH)$ is the subalgebra of $M_{2n}(C^*_r(\ZZ^d))$ commuting with the quaternionic structure $\Theta$ on the tight-binding Hilbert space $\CC^{2n}\otimes\ell^2(\ZZ^d)$.

Suppose we work in the Fourier transformed Hilbert space, $\CC^{2n}\otimes\ell^2(\ZZ^d)\overset{\rm Fourier}{\longrightarrow}\CC^{2n}\otimes L^2(\TT^d)$ instead. The standard quaternionic structure $\Theta=\Theta\otimes\kappa$ on $\CC^{2n}\otimes\ell^2(\ZZ^d)$ Fourier transforms into $\widehat{\Theta}:=\Theta\otimes(\iota_{\rm flip}\circ\kappa)$ acting on $\CC^{2n}\otimes L^2(\TT^d)$. Similarly, the complex $C^*$-algebra $M_{2n}(C^*_r(\ZZ^d))$ transforms into the algebra $M_{2n}(C(\TT^d))$ acting by multiplication. 
Recall that the induced (anti)automorphism on $C(\TT^d)$ is $\iota_{\rm flip}$, then on the matrix algebra of functions $M_{2n}(C(\TT^d))$, the antiautomorphism is $\sharp\otimes\iota_{\rm flip}$. Recalling Eq.\ \eqref{eqn:sharp.matrix}, the real subalgebra of $M_{2n}(C(\TT^d))$ which satisfies $f^{\sharp\otimes\iota_{\rm flip}}=f^*$ (where $(f^*)_{ij}=\overline{f_{ji}}$), is equivalently defined by $f^{(\sharp\otimes \iota_{\rm flip})\circ*}=({\rm Ad}_\Theta\otimes(\iota_{\rm flip}\circ\kappa))(f)=f$. It comprises those operators in $M_{2n}(C(\TT^d))$ which commute with the Fourier-transformed quaternionic structure $\widehat{\Theta}$.
\end{example}

\subsection{Bulk topological invariant}
Let $H=H^*$ be a local, translation-invariant Hamiltonian acting on some tight-binding Hilbert space $V\otimes \ell^2(\ZZ^d)$, which is time-reversal invariant, i.e.\ commutes with $\Theta$. For brevity, we will refer to such a Hamiltonian simply as a ``$\Theta$-invariant Hamiltonian''.

\subsubsection{Operator $K$-theory invariant of Fermi projection}
By the discussion in \S\ref{sec:tight.binding} and Example \ref{ex:translation.quaternion.algebra}, this means that
\begin{equation*}
H=H^*\in M_n((C^*_{r,\RR}(\ZZ^d))^\HH)\subset M_{2n}(C^*_r(\ZZ^d)).
\end{equation*}
Suppose further that $H$ has a \emph{spectral gap} around $\lambda\in\RR$, which we can set to $\lambda=0$ by a suitable overall energy shift $H\mapsto H-\lambda$. Formally, this means 
\begin{equation}
\exists \epsilon>0\;\;:\;\;
{\rm Spec}(H)\cap (-\infty,-\epsilon]\neq \emptyset \neq {\rm Spec}(H)\cap [\epsilon,\infty),\quad {\rm Spec}(H)\cap (-\epsilon,\epsilon)=\emptyset.\label{eqn:spectral.gap.definition}
\end{equation}
With the appropriate energy shift, we can arrange for $(-\epsilon,\epsilon)$ to be the entire spectral gap around $0$.

The \emph{Fermi spectral projection} $P_{\rm Fermi}$ onto the negative energy states of $H$ can now be written as $f(H)$ for some suitable continuous function $f:\RR\rightarrow\RR$. Simply choose $f(x)=1$ for $x\leq-\epsilon$, $f(x)=0$ for $x\geq \epsilon$, and continuously (even smoothly) interpolate over the spectral gap $(-\epsilon,\epsilon)$.

By standard properties of the continuous functional calculus, $P_{\rm Fermi}$ commutes with $\Theta$, and we have $P_{\rm Fermi}\in M_n((C^*_{r,\RR}(\ZZ^d))^\HH)$, thus defining a $K$-theory class
\begin{equation*}
[P_{\rm Fermi}]\in KO_0(C^*_{r,\RR}(\ZZ^d))^\HH).
\end{equation*}
The right side can be computed using the real Baum--Connes conjecture \cite{BK} (a ``$K$-theoretic Fourier transform'') and Poincar\'{e} duality, together with a stable splitting of the classifying space for $\ZZ^d$ (which is a $d$-torus $T^d=\RR^d/\ZZ^d$),
\begin{align*}
KO_0(C^*_{r,\RR}(\ZZ^d))^\HH)\cong KO_4(C^*_{r,\RR}(\ZZ^d))\cong KO_4(T^d)&\cong KO^{d-4}(T^d)\\
&\cong \widetilde{KO}^{d-4}(\bigvee_{j=0}^d (S^{j})^{\vee \binom{d}{j}})\\
&\cong \bigoplus_{j=0}^{d} (KO^{d-4-j})^{\oplus \binom{d}{j}}
\end{align*}
where $KO^{-n}$ denotes the $n$-th $KO$-group of a point.

For $d=1, 2, 3$, this gives
\begin{align}
KO_0(C^*_{r,\RR}(\ZZ)^\HH)&\cong \underbrace{KO^{-3}}_{0}\oplus \underbrace{KO^{-4}}_{\ZZ},\nonumber\\
KO_0(C^*_{r,\RR}(\ZZ^2)^\HH)
&\cong \underbrace{KO^{-2}}_{\ZZ_2}\oplus \underbrace{(KO^{-3})^2}_{0}\oplus \underbrace{KO^{-4}}_{\ZZ}, \nonumber\\
KO_0(C^*_{r,\RR}(\ZZ^3)^\HH)
&\cong \underbrace{KO^{-1}}_{{\rm ``strong"}\;\ZZ_2}\oplus \underbrace{(KO^{-2})^3}_{{\rm ``weak"}\;(\ZZ_2)^3}\oplus \underbrace{(KO^{-3})^3}_{0}\oplus\underbrace{KO^{-4}}_{\ZZ}.\label{eqn:KO.torus}
\end{align}
In the above, the first factors of $\ZZ_2$ are the ``strong'' $\ZZ_2$-invariants. For $d=3$, the extra three $\ZZ_2$ factors are the ``weak'' $\ZZ_2$-invariants \cite{Kitaev}, coming from the three standard ways in which $C^*_{r,\RR}(\ZZ^2)^\HH$ is included inside $C^*_{r,\RR}(\ZZ^3)^\HH$. Formally, this computation has the following meaning:

\begin{lem}
In dimension $d\geq 2$, there exist stably non-homotopic translation-invariant and $\Theta$-invariant Fermi projections, distinguished by mod-2 invariants.
\end{lem}

\subsubsection{$KR$-theory formulation: Brillouin zone picture}\label{sec:KR.formulation}
Recall from Example \ref{ex:translation.quaternion.algebra}, that in the Fourier transformed picture, the standard quaternionic strucure on $\CC^{2n}\otimes L^2(\TT^d)$ has the form $\widehat{\Theta}=\Theta\otimes(\iota_{\rm flip}\circ\kappa)$. If we think of this Hilbert space as the $L^2$-sections of the trivial Hermitian vector bundle $\TT^d\times \CC^{2n}$, then $\widehat{\Theta}$ makes $\TT^d\times \CC^{2n}$ into a trivial ``Quaternionic'' vector bundle over $(\TT^d,\iota_{\rm flip})$, in the sense of \cite{Dupont,D-G-cohom}. Namely, $\widehat{\Theta}$ is an antiunitary lift of $\iota_{\rm flip}$ to a bundle map squaring to $-1$.

Consider the Fourier transform of the gapped Hamiltonian $H$. This is a continuous Hermitian matrix assignment
\begin{equation*}
\TT^d\ni\vect{z}\mapsto H^{\rm 0D}(\vect{z})\in M_{2n}(\CC),
\end{equation*}
where we have written $\vect{z}:=(z_1,\ldots,z_d)$. We may think of $H^{\rm 0D}$ as an endomorphism of the Hermitian vector bundle $\TT^d\times\CC^{2n}$. Similarly, the Fermi projection $P_{\rm Fermi}=f(H)$ Fourier transforms to the continuous projection matrix family
\begin{equation*}
P_{\rm Fermi}^{\rm 0D}:\vect{z}\mapsto \frac{1-{\rm sgn}(H^{\rm 0D}(\vect{z}))}{2}.
\end{equation*}
Here, ${\rm sgn}$ is the sign function, which is continuous on the spectrum of $H^{\rm 0D}(\vect{z})$ due to the spectral gap at $0$. 

By commutativity with $\widehat{\Theta}$, both $H^{\rm 0D}$ and $P_{\rm Fermi}^{\rm 0D}$ are endomorphisms in the ``Quaternionic'' sense. The range bundle of $P_{\rm Fermi}^{\rm 0D}$ (sometimes called the \emph{Bloch bundle}), equipped with the restricted $\widehat{\Theta}$, is a ``Quaternionic'' subbundle over $(\TT^d,\iota_{\rm flip})$. This means that it is a Hermitian vector bundle with an antiunitary lift $\widehat{\Theta}$ of the base involution $\iota_{\rm flip}$. Unlike the ambient trivial bundle, the Bloch subbundle may not be trivial in the `Quaternionic' category. Indeed, it defines a class in Dupont's `Quaternionic' $K$-theory \cite{Dupont}, $KQ^{0}(\TT^d,\iota_{\rm flip})$.

To calculate $KQ^{0}(\TT^d,\iota_{\rm flip})$, we can use the equivariant stable splitting of $(\TT^d,\iota_{\rm flip})$ into a wedge sum of involutive $j$-spheres $S^{1,j}$ (the notation $S^{1,j}$ means the unit sphere in $\RR^{1+j}$ in which the last $j$-coordinates are flipped):
\begin{equation*}
KQ^{0}(\TT^d,\iota_{\rm flip})\cong KR^{-4}(\TT^d,\iota_{\rm flip})\cong \widetilde{KR}^{-4}(\bigvee_{j=0}^d (S^{1,j})^{\vee \binom{d}{j}})\cong \bigoplus_{j=0}^{d} (KO^{-4+j})^{\oplus \binom{d}{j}}.
\end{equation*}
For $d=1, 2, 3$, this gives
\begin{align}
KQ^0(\TT,\iota_{\rm flip})&\cong \underbrace{KO^{-4}}_{\ZZ}\oplus \underbrace{KO^{-3}}_{0},\nonumber\\
KQ^{0}(\TT^2,\iota_{\rm flip})
&\cong  \underbrace{KO^{-4}}_{\ZZ}\oplus \underbrace{(KO^{-3})^2}_{0}\oplus\underbrace{KO^{-2}}_{\ZZ_2}, \nonumber\\
KQ^{0}(\TT^3,\iota_{\rm flip})
&\cong   \underbrace{KO^{-4}}_{\ZZ}\oplus \underbrace{(KO^{-3})^3}_{0}\oplus \underbrace{(KO^{-2})^3}_{{\rm ``weak"}\;(\ZZ_2)^3} \oplus \underbrace{KO^{-1}}_{{\rm ``strong"}\;\ZZ_2}.\label{eqn:KR.torus}
\end{align}
This is a ``T-dual'' computation to that in \eqref{eqn:KO.torus}, as explained in \cite{MT-higher}.

%%%%%%%%%%%%%%%%%%%%%%%%%%%%%

\section{Spectral gap-filling for half-space topological insulators}\label{sec:gap.filling}
The $\ZZ_2$-invariants of the (Fermi projection of) time-reversal invariant Hamiltonians, are somewhat abstract labels. The deeper (and physically observable) meaning is to be found in the unusual boundary-localized states which fill up the bulk spectral gap, whenever the Hamiltonian is restricted to a half-space. In the absence of the $\Theta$-symmetry requirement, this ``topologically protected'' \emph{gap-filling phenomenon} was investigated in detail in \cite{PSB}, and further in \cite{T-edge} for general geometrically imperfect half-spaces.

Formally, the tight-binding Hilbert space for a Euclidean half-space is the subspace $\jmath:V\otimes \ell^2(\ZZ^{d-1}\times\NN)\hookrightarrow V\otimes \ell^2(\ZZ^d)$, with $\jmath^*$ the projection. Note that $\jmath^*\circ\jmath={\rm id}_{V\otimes\ell^2(\ZZ^{d-1}\times\NN)}$. Let us suppress the internal Hilbert space $V$ for the moment.

Given a bounded operator $S$ on $\ell^2(\ZZ^d)$, its truncation to $\ell^2(\ZZ^{d-1}\times\NN)$ is $\widetilde{S}:=\jmath^*\circ S\circ\jmath$. This is $*$-preserving on the bounded operators, $\widetilde{S^*}=(\widetilde{S})^*$, but may not be multiplicative. For instance, the unitary translation $T_d$ truncates into an isometry $\widetilde{T}_d$ which is not unitary, $\widetilde{T_d}\widetilde{T_d}^*\neq 1$, although the remaining truncated translations, $\widetilde{T_i}, i=1,\ldots,d-1$, are still unitary.

Let $C^*_r(\ZZ^{d-1}\times\NN)$ denote the unital $C^*$-subalgebra of $\mathcal{B}(\ell^2(\ZZ^{d-1}\times\NN))$ generated by $\widetilde{T}_i, i=1,\ldots, d$. This is also called the semigroup $C^*$-algebra for $\ZZ^{d-1}\times\NN$. The $*$-linear map defined on generators by $\widetilde{T}_i\mapsto T_i$ gives a $*$-homomorphism
\begin{equation*}
q:C^*_r(\ZZ^{d-1}\times\NN)\rightarrow C^*_r(\ZZ^d).
\end{equation*}
Observe that $C^*_r(\ZZ^{d-1}\times\NN)\cong C^*_r(\ZZ^{d-1})\otimes C^*_r(\NN)$. The $C^*_r(\NN)$ factor is just the classical Toeplitz $C^*$-algebra, which fits into the short exact sequence
\begin{equation*}
0\longrightarrow \mathcal{K}\longrightarrow C^*_r(\NN)\overset{q}{\longrightarrow} C^*_r(\ZZ)\longrightarrow 0,
\end{equation*}
with $\mathcal{K}$ the compact operators on $\ell^2(\NN)$. Tensoring back $C^*_r(\ZZ^{d-1})$ gives the short exact sequence
\begin{equation*}
0\longrightarrow C^*_r(\ZZ^{d-1})\otimes\mathcal{K}\longrightarrow C^*_r(\ZZ^{d-1}\times\NN)\overset{q}{\longrightarrow} C^*_r(\ZZ^d)\longrightarrow 0.
\end{equation*}
Taking matrix algebras, we also have a short exact sequence
\begin{equation}
0\longrightarrow M_{2n}(C^*_r(\ZZ^{d-1})\otimes\mathcal{K})\longrightarrow M_{2n}(C^*_r(\ZZ^{d-1}\times\NN))\overset{q}{\longrightarrow} M_{2n}(C^*_r(\ZZ^d))\longrightarrow 0,\label{eqn:complex.bulk.boundary.algebras}
\end{equation}
so we can include the internal $V\cong \CC^{2n}$ in this discussion.

A local, translation invariant Hamiltonian $H\in M_{2n}(C^*_r(\ZZ^d))$ is, by definition, the norm-limit of some linear combination of powers of $T_i, T_j^*$. Correspondingly, its half-space version, the truncation $\widetilde{H}=\jmath^*\circ H\circ\jmath$, is approximated by the same combination of $\widetilde{T}_i, \widetilde{T}_j^*$, and lies in $M_{2n}(C^*_r(\ZZ^{d-1}\times\NN))$.

\subsection{$K$-theory connecting map and gap-filling argument}
\label{sec:gap-filling.invariant}
The spectrum of $\widetilde{H}$ as a bounded operator on $V\otimes \ell^2(\ZZ^{d-1}\times\NN)$ is also its spectrum in the $C^*$-algebra $M_{2n}(C^*_r(\ZZ^{d-1}\times\NN))$ (this is known as \emph{spectral permanence}); likewise for $H$. For any $\mu\in\RR$, if $\widetilde{H}-\mu$ is invertible, then so is $H-\mu$. Thus the spectrum of $\widetilde{H}$ contains the spectrum of $H$. Now suppose $H$ is gapped in the sense Eq.\ \eqref{eqn:spectral.gap.definition}, thus some open interval $(-\epsilon,\epsilon)$ is in the complement of ${\rm Spec}(H)$.

\noindent
{\bf Fundamental question:} Does $\widetilde{H}$ retain any gap inside the bulk spectral gap $(-\epsilon,\epsilon)$ of $H$?

\begin{dfn}\label{dfn:gap.filling.property}
With $H$ and $\widetilde{H}$ as above, we say that $\widetilde{H}$ has the \emph{spectral gap-filling property}, or simply \emph{gap-filling property}, if it has no spectral gaps in $(-\epsilon,\epsilon)$. 
\end{dfn}

Suppose $\lambda\not\in {\rm Spec}(\widetilde{H})$ for some $\lambda\in(-\epsilon,\epsilon)$. Thus $(\lambda-\epsilon^\prime,\lambda+\epsilon^\prime)\subset (-\epsilon,\epsilon)$ is a spectral gap of $\widetilde{H}$ for some $\epsilon^\prime$. In the expression $P_{\rm Fermi}=f(H)$, we could have chosen the continuous $f$ such that it interpolates between the value 1 and 0 only within this sub-gap $(\lambda-\epsilon^\prime,\lambda+\epsilon^\prime)$. Thus $f$ only takes the values $0,1$ on ${\rm Spec}(\widetilde{H})$, meaning that $\widetilde{P}_{\rm Fermi}=f(\widetilde{H})\in M_{2n}(C^*_r(\ZZ^{d-1}\times\NN))$ remains a \emph{projection}. Since 
\begin{equation*}
q(\widetilde{P}_{\rm Fermi})=q(f(\widetilde{H}))=f(q(\widetilde{H}))=f(H)=P_{\rm Fermi},
\end{equation*}
the bulk Fermi projection $P_{\rm Fermi}$ admits a lift to a \emph{projection} $\widetilde{P}_{\rm Fermi}$ living inside \mbox{$M_{2n}(C^*_r(\ZZ^{d-1}\times\NN))$}.

Conversely, suppose we are able to find an obstruction to $P_{\rm Fermi}$ possessing  a \emph{projection} lift in $M_{2n}(C^*_r(\ZZ^{d-1}\times\NN))$. \emph{Then $\widetilde{H}$ cannot have any spectral gap whatsoever} inside $(-\epsilon,\epsilon)$!

In the absence of $\Theta$, the complex $K$-theory connecting map (which is also called an \emph{exponential map}),
\begin{equation*}
\partial:K_0(C^*_r(\ZZ^d))\ni [P_{\rm Fermi}]\mapsto [{\rm exp}(-2\pi i \widetilde{P}_{\rm Fermi})]\in K_1(C^*_r(\ZZ^{d-1})),
\end{equation*}
provides an obstruction to the existence of any \emph{projection} lift of $P_{\rm Fermi}$, as explained in detail in \cite{T-edge, PSB}. Thus it suffices to evaluate the non-triviality of the connecting map on $[P_{\rm Fermi}]$, to deduce the gap-filling property of $\widetilde{H}$.

\subsubsection*{Gap-filling argument in $\Theta$-invariant case}
The quaternionic structure $\Theta$ may be restricted to $\CC^{2n}\otimes \ell^2(\ZZ^{d-1}\times\NN)$. Suppose $H$ is $\Theta$-invariant, then $\widetilde{H}$ is also $\Theta$-invariant. If we write the ($\Theta$-invariant) bulk Fermi projection as $P_{\rm Fermi}=f(H)$, then $\widetilde{P}_{\rm Fermi}=f(\widetilde{H})$ is likewise a $\Theta$-invariant lift of $P_{\rm Fermi}$ in the real subalgebra \mbox{$M_n(C^*_{r,\RR}(\ZZ^{d-1}\times\NN)^\HH)\subset M_{2n}(C^*_{r}(\ZZ^{d-1}\times\NN))$}. 

Generally, $\widetilde{P}_{\rm Fermi}$ does not remain a projection, and the question now is whether there exists any projection lift of $P_{\rm Fermi}$ inside \mbox{$M_n(C^*_{r,\RR}(\ZZ^{d-1}\times\NN))^\HH$} at all. We will explain how a certain connecting morphism in $KO$-theory provides an obstruction. This requires some preliminaries on formulations of the $KO_i$ functors, following \cite{BL}.

\subsection{$KO_i$ functors and connecting maps}
\begin{dfn}
[$KO_0$ functor]
For a real unital $C^*$-algebra $A$, the group $KO_0(A)$ is just the Grothendieck group of the monoid of projections in $M_\infty(A)$ under direct sum. 
\end{dfn}
Here, $M_\infty(A)$ is the inductive limit of $M_n(A)$ as $n\rightarrow\infty$, with structure maps being ``upper-left corner embedding''. As in the complex case, for $KO_0$ of non-unital algebras, one passes to the unitisation $A^+$, and takes the kernel of the induced map under $A^+\rightarrow A^+/A\cong\RR$.

\begin{rem}
In \cite{BL}, each of the $KO_i, i\in\ZZ_8$, functors is formulated in a ``unitary convention''. In particular, the unitary definition of $KO_0$ (see Definition 5.1 of \cite{BL}) is as follows. Take self-adjoint unitaries in $\cup_{n\in\NN} M_{2n}(A)$ and mod out by homotopy, and direct summing of ``trivial unitaries'' of the form $\begin{pmatrix}1_m & 0 \\ 0 & -1_m\end{pmatrix}$. The inverse of $[u]$ in this definition of $KO_0(A)$ is $[-u]$.

The relation between the projection and unitary formulations of $KO_0$ is as follows. If there is a representative projection $p=p^2=p^*\in M_n(A)$, take the self-adjoint unitary 
\begin{equation}
u=\begin{pmatrix} 1_n-2p & 0 \\ 0 & -1_n \end{pmatrix}\in  M_{2n}(A).\label{eqn:projection.unitary}
\end{equation}
For instance, the zero projection corresponds to $\begin{pmatrix}1_n & 0 \\ 0 & -1_n\end{pmatrix}$.

In reverse, if $u$ is a self-adjoint unitary (i.e.\ a grading operator) in $M_{2n}(A)$, then $\frac{1}{2}(1_{2n}-u)$ is a projection in $M_n(A)$. So associate to $u$ the virtual projection $\frac{1}{2}(1_{2n}-u)\ominus 1_n$. This will recover the projection definition of $KO_0$, see Theorem 5.6 of \cite{BL} for details.
\end{rem}

\begin{rem}\label{rem:unitary.vs.proj}
We can equivalently work with projections/unitaries in the complexification $A^\CC$, at the expense of an additional reality requirement that $p^\flat=p$ or $u^\flat=u$. Since the $p$ or $u$ is self-adjoint, this additional requirement is equivalently $p^\iota=p$ or $u^\iota=u$; see Definition 5.9 of \cite{BL}.
\end{rem}

\begin{dfn}[$KO_4$ functor]
The $KO_4$ functor on real $C^*$-algebras may be defined as $KO_4(A)=KO_0(A^\HH)$. 
\end{dfn}

\begin{dfn}[$KO_3$ functor, Definition 6.17 of \cite{BL}]\label{defn:KO3}
Let $A$ be a unital real $C^*$-algebra, whose complexification $B=A^\CC$ has real structure denoted $\kappa_B$. The associated linear involutive antiautimorphism $\iota_B=\kappa_B\circ*$ is simply denoted $\iota$. Let $U_{2n}^{(3)}(B)$ be the subset of unitaries in $M_{2n}(B)$ satisfying
\begin{equation}
u^{\sharp\otimes\iota}=u,\label{eqn:KO3.condition}
\end{equation}
where the involutive antiautomorphism $\sharp\otimes\iota$ is given as in Example \ref{ex:matrix.algebra}.
Define
\begin{equation*}
U_\infty^{(3)}(B)=\lim_{\longrightarrow} U_{2n}^{(3)}(B),
\end{equation*}
where the embedding $U_{2n}^{(3)}(B)\hookrightarrow U_{2(n+m)}^{(3)}(B)$ is by direct summing with the identity $\mathbbm{1}_{2m}$ of $M_{2m}(B)$. Then $KO_3(A)$ is defined to be the set of path-components of $U_\infty^{(3)}(B)$, with composition law $[u]+[v]=\left[\begin{pmatrix} u & 0 \\ 0 & v \end{pmatrix}\right]$.
\end{dfn}
If $A$ is non-unital with unitization $A^+$ (and correspondingly for the complexifications), then $KO_3(A)$ is defined as the kernel of the map $KO_3(A^+)\rightarrow KO_3(\RR)$ induced by the map extracting the scalar part (Definition 6.20 of \cite{BL}).

As discussed in Example \ref{ex:matrix.algebra}, the condition $u^{\sharp\otimes\iota}=u$ in Eq.\ \eqref{eqn:KO3.condition} above is just the condition of being fixed under ${\rm Ad}_\Theta\otimes\kappa_B$, where $\Theta$ is the canonical quaternionic structure on $\CC^{2n}$ on which the $M_{2n}(\CC)$ factor acts.

\medskip

\begin{example}\label{ex:KR3.classifying.map}[Commutative case]
According to Definition \ref{defn:KO3}, a representative of a class in $KO_3(C(X,\iota))$ is a unitary $u\in U_{2n}(C(X))$ such that $u^{\sharp\circ\iota}=u$. Think of $u$ as a continuous map $X\rightarrow {\rm U}(2n)$. Then by construction,
\begin{equation*}
u^{\sharp\circ\iota}(x)=(u(\iota(x)))^\sharp,\qquad{\rm in\;\;} {\rm U}(2n).
\end{equation*}
Thus the condition $u^{\sharp\circ\iota}=u$ is just the statement that $u$ is a $\ZZ_2$-equivariant map from $(X,\iota)$ to $({\rm U}(2n),\sharp)$. In other words, we may equivalently formulate $KO_3(C(X,\iota))$ as the equivariant homotopy classes of maps
\begin{equation*}
X\rightarrow {\rm U}(\infty):=\varinjlim{\rm U}(2n).
\end{equation*}
In $KR$-theory language, we may also write
\begin{equation}
KO_3(C(X,\iota))\cong KR^{-3}(X,\iota)=[X,{\rm U}(\infty)]_{\ZZ_2}.\label{eqn:KR3.classifying.map}
\end{equation}
\end{example}

\subsubsection*{$\delta_4:KO_4\rightarrow KO_3$ connecting map}
Let 
\begin{equation}
0\rightarrow I\rightarrow \tilde{A}\rightarrow A\rightarrow 0\label{eqn:SES.real}
\end{equation}
be an exact sequence of real $C^*$-algebras, with $\tilde{A}, A$ unital. The connecting map $\delta_4:KO_4\rightarrow KO_3$ can be explicitly expressed as an exponential map, following Definition 8.3 of \cite{BL}.
For $KO_4(A)$, we are taking the $KO_0$ functor (unitary, complexified definition, see Remark \ref{rem:unitary.vs.proj}) applied to $A^\HH\equiv\HH\otimes_\RR A$. Thus the representatives are self-adjoint unitaries in
\begin{equation*}
U_{2n}((A^\HH)^\CC)\subset M_{2n}((\HH\otimes_\RR A)^\CC)\cong M_{2n}(M_2(A^\CC))\cong M_{4n}(A^\CC),
\end{equation*}
satisfying $u^{\sharp\otimes\iota}=u$. (Here $\iota$ is the canonical involutive antiautomorphism on $A^\CC$, given by $\kappa\circ*$.) Then define
\begin{equation}
\delta_4:KO_4(A)\rightarrow KO_3(I),\qquad [u]\mapsto [-{\rm exp}(\pi i \tilde{u})],\label{eqn:exponential.map}
\end{equation}
where $\tilde{u}\in M_{4n}((\tilde{A}^\HH)^\CC)$ is any (self-adjoint\footnote{The self-adjointness of the lift seems to have been left implicit in \cite{BL}.}) lift of $u$ satisfying $\tilde{u}^{\sharp\otimes\iota}=\tilde{u}$ and $-1\leq \tilde{u}\leq 1$. Here, one checks that the unitary $-{\rm exp}(\pi i \tilde{u})$ does represent a $KO_3(I)$ element, as in the proof of \cite{BL} Lemma 8.4, as follows. Under the map $M_{2n}(\tilde{A}^\CC)\rightarrow M_{2n}(A^\CC)$, the element $-{\rm exp}(\pi i \tilde{u})$ is mapped to $-{\rm exp}(\pi i u)=1_{4n}$ since the self-adjoint unitary $u$ has only $\pm 1$ eigenvalues. By exactness of \eqref{eqn:SES.real}, $-{\rm exp}(\pi i \tilde{u})$ can be regarded as an element of $(M_{2n}(I^\CC))^+$. The condition Eq.\ \eqref{eqn:KO3.condition} follows from $-{\rm exp}(\pi i \tilde{u})^{\sharp\otimes\iota}=-{\rm exp}(\pi i \tilde{u})^{\flat\circ*}=-{\rm exp}(-\pi i \tilde{u})^\flat=-{\rm exp}(\pi i \tilde{u})$.

\medskip

Let us rewrite Eq.\ \eqref{eqn:exponential.map} in the projection point of view for $KO_4$. The self-adjoint unitary $u\in M_{4n}(A^\CC)$ corresponds via Eq.\ \eqref{eqn:projection.unitary} to some projection $p\in M_{2n}(A^\CC)$ satisfying $p^{\sharp\otimes\iota}=p$ ($=p^\flat$ since $p^*=p$). The same formula, Eq.\ \eqref{eqn:projection.unitary}, for passing between unitaries and projections, will give a self-adjoint lift $\tilde{p}$ in $M_{2n}(\tilde{A}^\CC)$ of the projection $p$, satisfying $\tilde{p}^\flat=\tilde{p}$. Note that $\tilde{p}$ will not generally remain a projection. We may compute
\begin{align}
\delta_4[p]&=[-{\rm exp}(\pi i\,{\rm diag}(1_{2n}-2\tilde{p}, -1_{2n}))]\nonumber\\
&=[{\rm diag}(-{\rm exp}(-2\pi i\tilde{p}){\rm exp}(\pi i 1_{2n}),1_{2n})]\nonumber\\
&=[-{\rm exp}(-2\pi i\tilde{p}){\rm exp}(\pi i 1_{2n})]\nonumber\\
&=[{\rm exp}(-2\pi i\tilde{p})],\label{eqn:exponential.map.q}
\end{align}
which is just the exponential map applied to the lift $\tilde{p}$ of $p$ in $\tilde{A}$.

\subsection{Proof of time-reversal invariant bulk-boundary correspondence}\label{sec:BBC}
The $\Theta$-invariant part of the short exact sequence Eq.\ \eqref{eqn:complex.bulk.boundary.algebras} is the sequence of real $C^*$-algebras,
\begin{equation}
0\longrightarrow M_n((C^*_{r,\RR}(\ZZ^{d-1})\otimes\mathcal{K_\RR})^\HH)\longrightarrow M_n(C^*_{r,\RR}(\ZZ^{d-1}\times\NN)^\HH)\overset{q}{\longrightarrow} M_n(C^*_{r,\RR}(\ZZ^d)^\HH)\longrightarrow 0.\label{eqn:q.SES}
\end{equation}
Here $\mathcal{K}_\RR$ denotes the compact operators on the real Hilbert space $\ell^2(\NN;\RR)$. Let us study the exponential connecting map, Eq.\ \eqref{eqn:exponential.map.q}, for the sequence Eq.\ \eqref{eqn:q.SES}.

We had, for a spectrally gapped $\Theta$-invariant Hamiltonian $H$, that $P_{\rm Fermi}\in M_n(C^*_r(\ZZ^d)^\HH)$, thus defining a class
\begin{equation*}
[P_{\rm Fermi}]\in KO_0(C^*_{r,\RR}(\ZZ^d)^\HH)\cong KO_4(C^*_{r,\RR}(\ZZ^d))\cong\begin{cases}\ZZ_2\oplus\ZZ,\qquad\qquad\quad\;\; d=2,\\
\ZZ_2\oplus(\ZZ_2)^3\oplus\ZZ,\qquad d=3.\end{cases}
\end{equation*}
Applying the connecting map $\delta_4$ gives
\begin{equation*}
\delta_4:[P_{\rm Fermi}]\mapsto [{\rm exp}(-2\pi i \widetilde{P}_{\rm Fermi})]\in KO_3(C^*_{r,\RR}(\ZZ^{d-1}))\cong\begin{cases}
\ZZ_2,\qquad\qquad\;\; d=2,\\
\ZZ_2\oplus(\ZZ_2)^2,\quad d=3.
\end{cases}
\end{equation*}

\medskip

\subsubsection*{Computation of $\delta_4$} For computing $\delta_4:KO_4(C^*_{r,\RR}(\ZZ^{d}))\rightarrow KO_3(C^*_{r,\RR}(\ZZ^{d-1}))$, we observe that the short exact sequence of interest,
\begin{equation}
0\longrightarrow C^*_{r,\RR}(\ZZ^{d-1})\otimes\mathcal{K_\RR}\longrightarrow C^*_{r,\RR}(\ZZ^{d-1}\times\NN)\overset{q}{\longrightarrow} C^*_{r,\RR}(\ZZ^d)\longrightarrow 0,\label{eqn:basis.real.SES}
\end{equation}
is simply a Toeplitz-type extension. In more detail, it is the standard real Toeplitz extension,
\begin{equation}
0\longrightarrow \mathcal{K}_\RR\longrightarrow C^*_{r,\RR}(\NN)\longrightarrow C^*_{r,\RR}(\ZZ)\longrightarrow 0,\label{eqn:real.Toeplitz}
\end{equation}
tensored with $C^*_{r,\RR}(\ZZ^{d-1})$. Equivalently, for ${\rm id}:\ZZ\rightarrow {\rm Aut}(C^*_{r,\RR}(\ZZ^{d-1}))$ the trivial action, we have $C^*_{r,\RR}(\ZZ^d)\cong  C^*_{r,\RR}(\ZZ^{d-1})\rtimes_{\rm id}\ZZ$ as the associated crossed product algebra. Then \eqref{eqn:basis.real.SES} is the associated Toeplitz-type extension in the sense of \cite{PV}.

The point now is that the long exact sequence in $KO$-theory for \eqref{eqn:basis.real.SES} simplifies to the Pimsner--Voiculescu (PV) sequence (the real PV-sequence is explained in \cite{Ros-PV}, Theorem 2.7),
\begin{equation*}
\cdots KO_4(C^*_{r,\RR}(\ZZ^{d-1}))\overset{1-{\rm id}_*}{\longrightarrow}KO_4(C^*_{r,\RR}(\ZZ^{d-1}))\overset{j_*}{\rightarrow} KO_4(C^*_{r,\RR}(\ZZ^{d}))\overset{\delta_4}{\rightarrow} KO_3(C^*_{r,\RR}(\ZZ^{d-1}))\overset{1-{\rm id}_*}{\longrightarrow}\cdots
\end{equation*}
where $j$ denotes the inclusion of $C^*_{r,\RR}(\ZZ^{d-1})$ into the crossed product algebra \mbox{$C^*_{r,\RR}(\ZZ^{d-1})\rtimes_{\rm id}\ZZ\cong C^*_{r,\RR}(\ZZ^{d})$}. The trivial (auto)morphism ${\rm id}$ induces the identity map in $K$-theory, thus $1-{\rm id}_*=0$, and the PV sequence simplifies into short exact sequences such as
\begin{equation}
0\rightarrow  KO_4(C^*_{r,\RR}(\ZZ^{d-1}))\overset{j_*}{\rightarrow} KO_4(C^*_{r,\RR}(\ZZ^{d}))\overset{\delta_4}{\rightarrow} KO_3(C^*_{r,\RR}(\ZZ^{d-1}))\rightarrow 0.\label{eqn:PV.simplify}
\end{equation}

\subsubsection*{Connecting map as a Gysin map}
There is a geometric way to understand $\delta_4$ as a Gysin map in $KR$-theory. For convenience, we write $\tilde{S}^1$ for the circle with flip involution, sometimes called a `Real' circle. Let $(X,\iota)$ be the compact involutive space which is Gelfand-dual to a real unital commutative $C^*$-algebra $A$. Then the trivial crossed product $A\rtimes_{\rm id}\ZZ$ is another real unital commutative $C^*$-algebra corresponding to $(X\times \tilde{S}^1,\iota\times\iota_{\rm flip})$. The inclusion $j:A\rightarrow A\rtimes_{\rm id}\ZZ$ corresponds to the projection $\pi:X\times\tilde{S}^1\rightarrow X$ defining $X\times\tilde{S}^1$ as a (trivial) `Real' circle bundle over $X$. Associated to $\pi$ is the `Real' Gysin sequence in $KR$-theory \cite{Atiyah-KR}, which may be derived from the Thom isomorphism\footnote{Let $D^{1,1}$ be the unit disc in the plane with involution $(x,y)\mapsto (x,-y)$, whose boundary $S^{1,1}$ is $\tilde{S}^1$. In the $KR$-theory exact sequence for the pair $(X\times D^{1,1}, X\times S^{1,1})$, substitute the Thom isomorphism $KR^{-n}(X\times D^{1,1}, X\times S^{1,1})\cong KR^{-n}(X)$ to obtain the Gysin sequence.}
\begin{equation*}
\cdots \overset{\pi_*}{\rightarrow} KR^{-4}(X)\overset{\cup\chi}{\rightarrow} KR^{-4}(X)\overset{\pi^*}{\rightarrow} KR^{-4}(X\times\tilde{S}^1)\overset{\pi_*}{\rightarrow} KR^{-3}(X)\overset{\cup\chi}{\rightarrow}\cdots.
\end{equation*}
Here $\chi$ is the reduced $KR^0$ class of the associated `Real' line bundle, which is trivial in this case. So the Gysin sequence simplifies to
\begin{equation}
0\rightarrow KR^{-4}(X)\overset{\pi^*}{\rightarrow} KR^{-4}(X\times\tilde{S}^1)\overset{\pi_*}{\rightarrow} KR^{-3}(X)\rightarrow 0.\label{eqn:gysin.simplify}
\end{equation}
For $X=\TT^{d-1}$, Eq.\ \eqref{eqn:gysin.simplify} is the geometric version of Eq.\ \eqref{eqn:PV.simplify}.

\medskip

\medskip

For explicit description of the connecting map $\delta_4$, we fix some generators for the $K$-theory groups involved. First, there is a natural splitting $KR^{-4}(\TT^2)\cong \widetilde{KR}^{-4}(\TT^2)\oplus KR^{-4}(\pt)=\ZZ_2\oplus\ZZ$; accordingly we write $(\nu; n)\in KR^{-4}(\TT^2)$ with $n\in\ZZ$, and $\nu\in\ZZ_2$ the mod-2 invariant. 

For $\widetilde{KR}^{-4}(\TT^3)\cong(\ZZ_2)^4$, the three ``weak'' generators are pulled back from the generator of $\widetilde{KR}^{-4}(\TT^2)\cong\ZZ_2$ along the maps $\pi_i$ projecting out the $i$-th coordinate, $i=1,2,3$. There remains a ``strong'' generator (a possible choice is described in Remark \ref{rem:weak.strong}, or Eq.\ \eqref{eqn:KR.torus}). Accordingly, we write $(\nu_0;\nu_1,\nu_2,\nu_3;n)\in \ZZ_2\oplus(\ZZ_2)^3\oplus\ZZ\cong KR^{-4}(\TT^3)$, and call $\nu_0$, $\nu_i$ the strong (mod-2) invariant and $i$-th weak invariant, respectively. In both the $d=2$ and $d=3$ cases, the integer invariant $n$ is the ``quaternionic rank'' invariant.

To translate the above into the operator formulation, let $j_i$ be the inclusion $C^*_{r,\RR}(\ZZ^2)\rightarrow C^*_{r,\RR}(\ZZ^3)$ induced from the sublattice $\ZZ^2$ orthogonal to the $i$-th generator of $\ZZ^3$. Then $\nu_i$ lies in the subgroup $(j_i)_*\widetilde{KO}_4(C^*_{r,\RR}(\ZZ^2))\subset \widetilde{KO}_4(C^*_{r,\RR}(\ZZ^3))$.

\begin{thm}\label{thm:BBC}
Let $P_{\rm Fermi}$ be the Fermi projection for a $\Theta$-invariant gapped, local Hamiltonian $H$ in dimension $d=2,3$. For $d=2$, suppose the mod-2 invariant of $P_{\rm Fermi}$ is non-trivial. For $d=3$, suppose any of the strong invariant $\nu_0$ or the weak invariants $\nu_1$, $\nu_2$ of $P_{\rm Fermi}$ is/are non-trivial. Then the half-space Hamiltonian $\widetilde{H}$ necessarily has the spectral gap-filling property in the sense of Definition \ref{dfn:gap.filling.property}.
\end{thm}
\begin{proof}
Since the connecting map $\delta_4$ is an exponential map (Eq.\ \eqref{eqn:exponential.map.q}), as explained in \S\ref{sec:gap-filling.invariant}, gap-filling will be guaranteed whenever $\delta_4[P_{\rm Fermi}]\neq 0$. From Eq.\ \eqref{eqn:PV.simplify}, ${\rm ker}\,\delta_4=(j_d)_*(KO_4(C^*_{r,\RR}(\ZZ^{d-1})))$. Equivalently, from Eq.\ \eqref{eqn:gysin.simplify}, ${\rm ker}\,((\pi_d)_*)=\pi_d^*(KR^{-4}(\TT^{d-1}))$. So for $d=2$, $(\nu;n)\not\in {\rm ker}\,\delta_4$ whenever $\nu\neq 0$; similarly, for $d=3$, $(\nu_0;\nu_1,\nu_2,\nu_3;n)\not\in {\rm ker}\,\delta_4$ whenever $\nu_0,\nu_1,\nu_2\neq 0$. 
\end{proof}

\begin{rem}
A ``bulk-boundary correspondence'' is typically of the form ``non-trivial bulk invariant $\Rightarrow$ non-trivial boundary invariant''. Later, we will see that the half-space Hamiltonian $\widetilde{H}$ has a ``boundary invariant'' in $KO_3(C^*_{r,\RR}(\ZZ^{d-1}))$, see Definition \ref{defn:boundary.K.class}. We show that $\delta_4[P_{\rm Fermi}]=[\widetilde{H}]$ (Corollary \ref{cor:BBC}), whence Theorem \ref{thm:BBC} may be rephrased in the above ``bulk-boundary correspondence'' form.
\end{rem}

\begin{rem}
Theorem \ref{thm:BBC} in its $K$-theoretic form may appear rather abstract, but we will make it more descriptive by passing to cohomology. For the bulk mod-2 invariants of $P_{\rm Fermi}$, they localize to sign invariants at the fixed points of $\TT^d$ as explained in \S\ref{sec:FKMM.invariant}. For the boundary invariants $\delta_4[P_{\rm Fermi}]$ (or $[\widetilde{H}]$) in $KO_3(C^*_{r,\RR}(\ZZ^{d-1}))\cong KR^{-3}(\TT^{d-1})$, these will later be understood in cohomology as the `Real' gerbe invariants associated to $\widetilde{H}$.
\end{rem}

\begin{rem}[Related literature on Class AII bulk-edge correspondence]\label{rem:BBC.literature}
Bulk-edge correspondence for a certain subclass of 2D time-reversal invariant Hamiltonians defined by second-order difference equations, was proved in \cite{GP} using methods which address the spectrum more directly. $K$-theoretic methods, generalising \cite{PSB} to the Class AII case in general dimensions, were adopted in \cite{T-class,MT-higher,Kellendonk,BKR,GSB,AMZ}, and recalled here briefly. In \cite{T-class}, operator $K$-theory techniques for studying Class AII bulk topological insulators (and also general symmetry classes) were introduced, while in \cite{MT-higher}, non-trivial connecting maps to the boundary were proposed and analysed with the aid of T-duality. However, it was left implicit that these maps coincided with the Toeplitz extensions' connecting map $\delta_4$. In \cite{AMZ}, working with formulations of $K$-theory due to van Daele and Kasparov, the physical meaning of a non-trivial $K$-theory connecting map was given in terms of gaplessness of the boundary theory, although the connecting map was not explicitly computed. In \cite{BKR}, the bivariant Kasparov theory provided $\ZZ_2$-valued pairings of the ``strong'' $K$-theory classes with certain $K$-homology classes, which are preserved under the Toeplitz connecting maps from bulk invariants to abstract boundary invariants. However, the explicit Kasparov pairing computation (which is generally very hard), and the physical meaning of the boundary invariants, were not provided. Other (non-trivial) $\ZZ_2$-valued pairings for the bulk theory were studied in \cite{Kellendonk, GSB}, but a duality theory for passing these on to the boundary is not (yet) available. 

We also mention \cite{Gaw}, which, in $d=2,3$, regards $u_p=1-P_{\rm Fermi}$ as an equivariant field $\TT^d\rightarrow {\rm U}(2n)$ with involution $u\mapsto\Theta u \Theta^{-1}$ on the latter (note: this involution happens to coincide with our Eq.\ \eqref{eqn:sharp.unitary.formula} for $u_p$ self-adjoint). By considering an equivariant extension of the basic gerbe over ${\rm U}(2n)$ (with connection), he constructs, in $d=2$, a notion of ``square root of gerbe holonomy'' along $u_p$, as well as a related index for $d=3$. These indices were shown to be $\ZZ_2$-valued, and were identified with the Fu--Kane--Mele formulation of the strong \emph{bulk} $\ZZ_2$-invariant for topological insulators (see Eq.\ \eqref{eqn:FKM.POS}). In \S\ref{sec:real.gerbes}, we study the rather different notion of `Real' gerbes, with the aim of constructing a gerbe out of spectral data of $\widetilde{H}$, living over the \emph{boundary} Brillouin zone $\TT^{d-1}$.
\end{rem}

\subsection{Sign invariants of `Quaternionic' Bloch bundles}\label{sec:FKMM.invariant}

\subsubsection{FKMM invariant}
Let $(X,\iota)$ be an involutive space with the structure of a $\ZZ_2$-CW-complex. 
We write $H^n_\pm(X)$ for the equivariant cohomology with local coefficients $H^n_{\Z_2}(X;\Z(1))$, where $\Z(1)$ is the local system $X\times\ZZ$ made equivariant by the action $(x,l)\mapsto (\iota(x),-l)$. A \v{C}ech formulation of this cohomology is briefly discussed in \S\ref{sec:Real.gerbe.class}.

In \cite{D-G-AII, D-G-cohom}, a characteristic class theory was developed for `Quaternionic' bundles over $(X,\iota)$ --- the so-called \emph{FKMM invariant} $\kappa$, which we briefly recall. 

\begin{itemize}
\item
Let $W\subset X$ be an invariant subspace. Consider a pair $(L,s)$ comprising a `Real' line bundle\footnote{This means a complex line bundle with an antiunitary lift of $\iota$ squaring to $+1$.} $L\rightarrow X$ and a nowhere vanishing equivariant section $s\in\Gamma(W, L|_W)$. Two pairs $(L,s)$ and $(L', s')$ are isomorphic if there is an isomorphism $f:L\rightarrow L'$ of `Real' line bundles such that $f|_W\circ s=s'$. Then the isomorphism classes of pairs $(L,s)$ are classified by the relative cohomology $H^2_\pm(X,W)$. (If $W=\emptyset$, this is the `Real' Chern class of $L$ introduced in \cite{Kahn}.)
\item 
Let $\pi:E\rightarrow X$ be a `Quaternionic' vector bundle of even rank (over $\CC$). If the fixed-point set $X^\iota$ is non-empty, then over $X^\iota$, there uniquely exists a nowhere vanishing (normalized equivariant) section $\sigma_E\in\Gamma(X^\iota,{\rm det}\, E|_{X^\iota})$ of the `Real' determinant line bundle. 
\item By definition, the FKMM invariant $\kappa(E)\in H^2_\pm(X,X^\iota)$ is the isomorphism class of the pair $({\rm det}\, E, \sigma_E)$.
\end{itemize}
Let ${\rm Vect}^{2m}_Q(X)$ denote the set of isomorphism classes of `Quaternionic' vector bundles over $X$ of (complex) rank $2m$. For each even rank $2m$, the FKMM invariant gives a map,
\begin{equation*}
\kappa: {\rm Vect}^{2m}_Q(X)\rightarrow H^2_\pm(X, X^\iota).
\end{equation*}
Alternatively, we can think of the FKMM invariant as a homomorphism
\begin{equation*}
\kappa:\widetilde{KQ}^0(X)=\widetilde{KR}^{-4}(X)\rightarrow H^2_\pm(X, X^\iota).
\end{equation*}

\subsubsection{Sign formulae for FKMM invariant for low dimensional tori}\label{sec:FKMM.sign} 
Generally speaking, for a given dimension $d$ of $X$, there is a \emph{stable rank} $2\sigma$ such that every `Quaternionic' bundle of rank $2m>2\sigma$ over $X$ splits into the direct sum of a rank $2\sigma$ `Quaternionic' bundle and a trivial rank $2(m-\sigma)$ `Quaternionic' bundle. 

We are interested in $X=\TT^d, d=2,3$ with the flip involution in each circle factor. Then $X^\iota$ is the finite set $\{\pm 1\}^d$ with $2^d$ points. In these cases, the stable rank condition is already attained for $\sigma=1$ (Theorem 4.2 of \cite{D-G-cohom}), and $\kappa$ is actually an \emph{isomorphism} of groups,
\begin{equation*}
\kappa:\widetilde{KQ}^0(X)=\widetilde{KR}^{-4}(X)\overset{\cong}{\longrightarrow} H^2_\pm(X, X^\iota).
\end{equation*}
Furthermore, one may verify that $H^2_\pm(X)=0$ for $X=\TT^d$ (Prop A.2 of \cite{D-G-AII}), so the relative cohomology sequence gives
\begin{equation*}
H^2_\pm(X, X^\iota)\cong H^1_\pm(X^\iota)/j^*H^1_\pm(X)
\end{equation*}
where $j^*$ is the restriction from $X$ to $X^\iota$ (we will often leave $j^*$ implicit in what follows). Generally speaking, we have $H^1_\pm(X)\cong [X,{\rm U}(1)]_{\ZZ_2}$ where ${\rm U}(1)$ is given the complex conjugation involution (Prop.\ A.2 of \cite{Gomi}). In particular, $H^1_\pm(X^\iota)$ is just the set of $\pm 1$-valued \emph{sign maps}, ${\rm Map}(X^\iota,\ZZ_2)$.

We can write the FKMM invariant
\begin{equation*}
\kappa:KR^{-4}(X)\rightarrow H^2_\pm(X,X^\iota)\cong {\rm Map}(X^\iota,\ZZ_2)/[X,{\rm U}(1)]_{\ZZ_2}
\end{equation*}
in terms of sign maps as follows. By Prop.\ 4.3 of \cite{D-G-AII}, a `Quaternionic' bundle $E$ over $\TT^d$ is necessarily trivial in the complex category, so that it has a global orthonormal frame of sections $\{t_1,\ldots,t_{2m}\}$. If $\widehat{\Theta}$ denotes the `Quaternionic' structure on the bundle $E$, we can construct its associated \emph{sewing matrix} \mbox{$w:X\rightarrow {\rm U}(2m)$} with components
\begin{equation*}
w_{ji}(x)=\langle t_j(\iota(x)) | \widehat{\Theta}(t_i(x))\rangle,
\end{equation*}
satisfying $w(\iota(x))=-w^{\rm t}(x), x\in X$, see \S4.2 of \cite{D-G-AII}. In particular, $w$ is antisymmetric at each fixed point. The determinant line bundle ${\rm det}\,E$ is a trivial `Real' line bundle $X\times\CC$ whose `Real' structure is
\begin{equation*}
{\rm det}(\widehat{\Theta}):(x,\lambda)\mapsto (\iota(x),{\rm det}(w)(x)\overline{\lambda}),
\end{equation*}
with ${\rm det}(w)(x)={\rm det}(w)(\iota(x))$. Triviality of ${\rm det}\,E$ means that there is a map $q_w:X\rightarrow {\rm U}(1)$ such that ${\rm det}(w)(x)=q_w(x)q_w(\iota(x))$. Thus at a fixed point, $q_w$ is a square root of ${\rm det}(w)$, and we obtain a sign map
\begin{equation*}
\mathfrak{d}_w:X^\iota\ni y\mapsto \frac{q_w(y)}{{\rm Pf}(w(y))}\in\ZZ_2,
\end{equation*}
which may be verified to be independent of the choice of orthonormal frame defining $w$. We observe that there is still a gauge freedom to modify $q_w$ by equivariant maps $\varepsilon\in [X,{\rm U}(1)]_{\ZZ_2}$. So only the class of $\mathfrak{d}_w$ in ${\rm Map}(X^\iota,\ZZ_2)/[X,{\rm U}(1)]_{\ZZ_2}$ is gauge-invariant.

With the above identifications, for $X=\TT^d, d=2,3$, the FKMM invariant is the assignment (Prop.\ 4.4 of \cite{D-G-AII})
\begin{equation*}
KR^{-4}(X)\ni [E]\mapsto [\mathfrak{d}_w]\in {\rm Map}(X^\iota,\ZZ_2)/[X,{\rm U}(1)]_{\ZZ_2}
\end{equation*}

\begin{lem}[Prop.\ A.2 of \cite{D-G-AII}]\label{lem:tori.sign.maps}
Let $\TT^d$ have the complex conjugation/flip involution in each coordinate, so the fixed points are labelled by $\{\pm 1\}^d$. Then
\begin{equation*}
{\rm Map}((\TT^d)^\iota,\ZZ_2)/[\TT^d,{\rm U}(1)]_{\ZZ_2}\cong (\ZZ_2)^{2^d-(d+1)}=\begin{cases}\ZZ_2,\;\;\quad d=2,\\ (\ZZ_2)^4,\;\;\quad d=3.\end{cases}
\end{equation*}
The equivariant maps in $[\TT^d,{\rm U}(1)]_{\ZZ_2}$ are generated by the $i$-th coordinate maps $x_i:\TT^d\rightarrow \TT\cong {\rm U}(1)$, $i=1,\ldots,d$, together with the constant map $-1$.
\end{lem} 

With a little thought, we see that for $X=\TT^2$, the class of a sign map $\mathfrak{d}$ in ${\rm Map}((\TT^d)^\iota,\ZZ_2)/[\TT^2,{\rm U}(1)]_{\ZZ_2}\cong\ZZ_2$ is given by the gauge-invariant product of signs
\begin{equation}
\Pi\mathfrak{d}:=\Pi_{y\in(\TT^2)^\iota}\mathfrak{d}(y)\label{eqn:POS}
\end{equation}
over the four fixed points. Thus the FKMM invariant is the map
\begin{align}
\ZZ_2\cong \widetilde{KR}^{-4}(\TT^2)&\rightarrow {\rm Map}((\TT^2)^\iota,\ZZ_2)/[\TT^2,{\rm U}(1)]_{\ZZ_2}\cong\ZZ_2 \nonumber \\
[E]&\mapsto \Pi\mathfrak{d}_w.\label{eqn:FKM.POS}
\end{align}
Historically, the product-of-signs formula on the right side was written down by physicists Fu--Kane--Mele \cite{FKM}, although its role in the `Quaternionic' characteristic class theory was only investigated later on \cite{D-G-AII, D-G-cohom}. 

\begin{rem}[Weak and strong generators]\label{rem:weak.strong}
For $X=\TT^3$, we can pull back the above 2D FKMM invariant on $\TT^2$ in three canonical ways. This accounts for the three ``weak'' generators $\mathfrak{d}_1, \mathfrak{d}_2, \mathfrak{d}_3$ of $(\ZZ_2)^3\subset\widetilde{KR}^{-4}(\TT^3)\cong(\ZZ_2)^4$. Consider the product-of-signs map over the eight fixed points,
\begin{equation*}
\Pi:\widetilde{KR}^{-4}(\TT^3)\overset{\cong}{\rightarrow} {\rm Map}((\TT^3)^\iota,\ZZ_2)/[\TT^3,{\rm U}(1)]_{\ZZ_2}\rightarrow\ZZ_2.
\end{equation*}
Note that for each $i$, we have $\Pi\mathfrak{d}_i=+1$ (that is, $\mathfrak{d}_i\in{\rm ker}\,\Pi$), since the fixed point set of $\TT^3$ is a double cover of the fixed point set of $\TT^2$, and $\mathfrak{d}_i$ is pulled back under the covering map.
To fix a splitting $\widetilde{KR}^{-4}(\TT^3)\cong{\rm ker}\,\Pi\oplus\ZZ_2$, pick the sign map $\mathfrak{d}_{\rm strong}$ whose value is $-1$ at the fixed point $(+1,+1,+1)$ and $+1$ at the remaining seven fixed points (so $\Pi\mathfrak{d}_{\rm strong}=-1$). We call $\mathfrak{d}_{\rm strong}$ a strong generator in $\widetilde{KR}^{-4}(\TT^3)$.
\end{rem}

\section{`Real' gerbes}\label{sec:real.gerbes}

There are a number of equivalent realisations of gerbes \cite{Hi,Murray}, but in all cases, they are classified by the third integral cohomology group $H^3(X;\ZZ)$ \cite{MS}. From this viewpoint, a `Real' gerbe on a space $X$ with involution can be considered to be any geometric object classified by the equivariant cohomology $H^3_{\pm}(X) = H^3_{\Z_2}(X; \Z(1))$ with local coefficients.

\subsection{A definition of `Real' gerbes}
Let $(X,\iota)$ be a space with involution. Given a collection $\{U_i\}_{i\in I}$ of open subsets, we write $U_{i_1 i_2\ldots i_n}=U_{i_1}\cap U_{i_2}\cap\ldots\cap U_{i_n}$.

We define a `Real' gerbe on $X$ to be
$$
\G = (\U, L_{ij}, \phi_{ijk}, J_i, \tau_i, \rho_{ij})
$$
consisting of the following data:

\begin{enumerate}
\item
$\U = \{ U_i \}_{i \in I}$ is an invariant open cover of $X$ (that is, $\iota(U_i)=U_i$ for all $i\in I$).

\item
$L_{ij} \to U_{ij}$ is a Hermitian line bundle. We assume that $L_{ii}$ is the trivial line bundle.

\item
$\phi_{ijk} : L_{ij} \otimes L_{jk} \to L_{ik}$ is a unitary isomorphism on $U_{ijk}$ which makes the following diagram on $U_{ijkl}$ commutative:
$$
\begin{CD}
L_{ij} \otimes L_{jk} \otimes L_{kl}
@>{1 \otimes \phi_{jkl}}>> 
L_{ij} \otimes L_{jl} \\
@V{\phi_{ijk} \otimes 1}VV @VV{\phi_{ijl}}V \\
L_{ik} \otimes L_{kl} 
@>{\phi_{ikl}}>>
L_{il}.
\end{CD}
$$

\item
$J_i \to U_i$ is a Hermitian line bundle.

\item
$\rho_{ij}$ is a unitary isomorphism on $U_{ij}$
$$
\rho_{ij} : J_i \otimes \overline{L}_{ij} \to \iota^*L_{ij} \otimes J_j
$$
which makes the following diagram on $U_{ijk}$ commutative
$$
\begin{CD}
J_i \otimes \overline{L}_{ij} \otimes \overline{L}_{jk}
@>{1 \otimes \overline{\phi}_{ijk}}>>
J_i \otimes \overline{L}_{ik} @=
J_i \otimes \overline{L}_{ik} \\
@V{\rho_{ij} \otimes 1}VV @. @VV{\rho_{ik}}V \\
\iota^*L_{ij} \otimes J_j \otimes \overline{L}_{jk}
@>{1 \otimes \rho_{jk}}>>
\iota^*L_{ij} \otimes \iota^*L_{jk} \otimes J_k
@>{\iota^*\phi_{ijk}}>>
\iota^*L_{ik} \otimes J_k.
\end{CD}
$$

\item
$\tau_i$ is a unitary isomorphism $\tau_i : J_i \to \iota^*J_i$ such that the composition
$$
\begin{CD}
J_i @>{\tau_i}>> \iota^*J_i @>{\iota^*\tau_i}>>
\iota^*\iota^*J_i = J_i
\end{CD}
$$
is the identity on $J_i$ (hence $J_i$ is a $\Z_2$-equivariant line bundle) and the following diagram on $U_{ij}$ is commutative
$$
\begin{CD}
J_i \otimes \overline{L}_{ij} @>{\tau_i \otimes 1}>>
\iota^*J_i \otimes \overline{L}_{ij} @=
L_{ji} \otimes \iota^*J_i \\
@V{\rho_{ij}}VV @. @AA{\iota^*\rho_{ji}}A \\
\iota^*L_{ij} \otimes J_j @>{1 \otimes \tau_j}>>
\iota^*L_{ij} \otimes \iota^*J_j
@=
\iota^*J_j \otimes \iota^*\overline{L}_{ji},
\end{CD}
$$
where the identification $L_{ji} = \overline{L}_{ij}$ is the composition of the unitary isomorphism $L_{ji} \cong L_{ij}^*$ induced from $\phi_{iji}$ and $L_{ij}^* \cong \overline{L}_{ij}$ from the Hermitian metric.

\end{enumerate}

\begin{rem}\label{rem:Gao.Hori}
We mention some special cases and generalisations of our `Real' gerbes, as well as their relation to some other constructions in the literature.
\begin{enumerate}
\item If the involution $\iota$ is absent, then the data $J_i,\tau_i,\rho_{ij}$ are not included, and $(\U, L_{ij}, \phi_{ijk})$ gives the data of a usual non-equivariant gerbe.
\item A `Real' gerbe in this paper is a special case of a $(\ZZ_2,{\rm id})$-\emph{twisted equivariant gerbe} \cite{GSW}, also called a \emph{Jandl gerbe} \cite{SSW}. It is also a `Real' version of the Hitchin--Chatterjee formulation of gerbes \cite{Hi} in terms of ``clutching line bundles''.
\item Given a $\Z_2$-equivariant line bundle $J \to X$ with its $\Z_2$-equivariant structure $\tau : J \to \iota^*J$, there is an associated `Real' gerbe, via the so-called Gao--Hori construction in \cite{HMSV2}. This puts a `Real' structure on the trivial gerbe. In our formulation of `Real' gerbes, we take the invariant open cover $\U = \{ X \}$ and the Hermitian line bundle $J \to X$ together with its $\Z_2$-equivariant structure $\tau : J \to \iota^*J$.
\item Given a map $\tau : X \to {\rm U}(1)$ such that $\iota^*\tau = \tau^{-1}$, we can make the product line bundle $J = X \times \C$ into a $\Z_2$-equivariant one, via $(x, z) \mapsto (\iota(x), \tau(x) z)$, and hence the Gao--Hori construction yields the associated `Real' gerbe.
\item By passing to a sufficiently fine open cover, we may take the $J_i$ to be trivial line bundles. Then the unitary isomorphism $\rho_{ij}:\overline{L}_{ij}\rightarrow \iota^* L_{ij}$ is just an antiunitary map on $L_{ij}$ covering the involution $\iota$, while the $\tau_i$ which makes $J_i$ into a $\ZZ_2$-equivariant line bundle amounts to a map $\tau_i:U_i\rightarrow {\rm U}(1)$ such that $\iota^*\tau_i = \tau_i^{-1}$.
\item Let $G$ be a finite group acting on $X$, and $\phi:G\rightarrow \ZZ_2$ a homomorphism. In the context of twisted equivariant $K$-theory of Freed--Moore \cite{FM, G3}, a \emph{$\phi$-twist} \cite{G3} consists of a local equivalence of groupoids $F : \mathcal{X} \to X//G$, and a $\phi$-twisted extension of $\mathcal{X}$. Our definition of `Real' gerbes above is an explicit description of a $\phi$-twist in the special case where $G=\ZZ_2$ and $\phi:G\rightarrow\ZZ_2$ is the identity map. In more detail: from an invariant open cover $\mathfrak{U} = \{ U_i \}$, we can construct a groupoid  $\mathcal{X}$ such that the space of objects is $\mathcal{X}_0 = \bigsqcup U_i$ and the space of morphisms is $\mathcal{X}_1 = \Z_2 \times \bigsqcup U_{ij}$. This groupoid has the obvious map to the quotient groupoid $X//G$, and this map turns out to be a local equivalence. Now, a $\phi$-twisted extension of $\mathcal{X}$ consists of a Hermitian line bundle $\mathcal{L}$ on $\mathcal{X}_1$ and an extra datum about a compatibility condition. The Hermitian line bundle $\mathcal{L}$ corresponds to the data $L_{ij}$ and $J_i$, while the datum about the compatibility corresponds to the remaining data $\phi_{ijk}$, $\rho_{ij}$ and $\tau_i$.
\end{enumerate}
\end{rem}

%%%%%%%%%%%%%%%%%%%%%%%%%
%%%%%%%%%%%%%%%%%%%%%%%%%

\subsection{Equivalences of `Real' gerbes}
Similarly to the case of non-equivariant bundle gerbes \cite{MS}, we will consider `Real' gerbes up to \emph{stable equivalence}, defined as that generated by ``refinement of open covers'' and by the following equivalence. Suppose two `Real' gerbes
\begin{align*}
\G &= (\U, L_{ij}, \phi_{ijk}, J_i, \tau_i, \rho_{ij}), &
\G' &= (\U, L'_{ij}, \phi'_{ijk}, J'_i, \tau'_i, \rho'_{ij}),
\end{align*}
share the same open cover $\U$. They are said to be equivalent if we have the following data:
\begin{enumerate}
\item
A Hermitian line bundle $K_i \to U_i$ for each $i \in I$.

\item
A unitary isomorphism  $\psi_{ij} : K_i \otimes L_{ij} \to L'_{ij} \otimes K_j$ for $i, j \in I$ which makes the following diagram commutative on $U_{ijk}$
$$
\begin{CD}
K_i \otimes L_{ij} \otimes L_{jk} @>{\psi_{ij} \otimes 1}>>
L'_{ij} \otimes K_j \otimes L_{jk} @>{1 \otimes \psi_{jk} }>> 
L'_{ij} \otimes L'_{jk} \otimes K_k \\
@V{1 \otimes \phi_{ijk}}VV @. @VV{\phi'_{ijk} \otimes 1}V \\
K_i \otimes L_{ik} @= K_i \otimes L_{ik} @>>{\psi_{ik}}> L'_{ik} \otimes K_k.
\end{CD}
$$

\item
A unitary isomorphism $\pi_i : J'_i \otimes \overline{K}_i \to \iota^* K_i \otimes J_i$ for $i \in I$ which makes the following diagram on $U_{ij}$ commutative
$$
\begin{CD}
\iota^*K_i \otimes J_i \otimes \overline{L}_{ij}
@>{1 \otimes \rho_{ij}}>> 
\iota^*K_i \otimes \iota^*L_{ij} \otimes J_j
@>{\iota^*\psi_{ij} \otimes 1}>>
\iota^*L'_{ij} \otimes \iota^*K_j \otimes J_j \\
@A{\pi_i \otimes 1}AA @. @AA{1 \otimes \pi_j}A \\
J_i' \otimes \overline{K}_i \otimes \overline{L}_{ij}
@>{1 \otimes \overline{\psi}_{ij}}>>
J'_i \otimes \overline{L}'_{ij} \otimes \overline{K}_j
@>{\rho'_{ij} \otimes 1}>>
\iota^*L'_{ij} \otimes J'_j \otimes \overline{K}_j,
\end{CD}
$$
and also the following diagram on $U_i$ commutative,
$$
\begin{CD}
J'_i \otimes \overline{K}_i @>{\tau'_i \otimes 1}>>
\iota^*J'_i \otimes \overline{K}_i
@=
\iota^*(J'_i \otimes \iota^*\overline{K}_i) \\
@V{\pi_i}VV @. @VV{\iota^*\pi_i^\iota}V \\
\iota^*K_i \otimes J_i 
@>{1 \otimes \tau_i}>>
\iota^*K_i \otimes \iota^*J_i
@=
\iota^*(K_i \otimes J_i),
\end{CD}
$$
where the right vertical isomorphism is the pull-back of the map $\pi_i^\iota : J'_i \otimes \iota^*\overline{K}_i \to K_i \otimes J_i$ induced from $\pi_i : J'_i \otimes \overline{K}_i \to \iota^* K_i \otimes J_i$.

\end{enumerate}

\subsection{Cohomology classification of `Real' gerbes}\label{sec:Real.gerbe.class}

`Real' gerbes in our formulation are classified up to stable equivalence by the equivariant cohomology with local coefficients $H^3_\pm(X) := H^3_{\Z_2}(X; \Z(1))$. In this section, we explain how to define a \v{C}ech cohomology class from a given `Real' gerbe. Once this is clarified, the classification by $H^3_\pm(X)$ can be completed along the lines of \S4 of \cite{GSW}.
 
As explained in Appendix \ref{appendix:Borel.Cech.twisted}, we can identify $H^n_\pm(X)$ with the $(n - 1)$th cohomology of the total complex associated to the double complex $(\check{C}^{p, q}, \delta, \partial)$, where $\check{C}^{p, q}(\U) = \check{C}^q(\U; \underline{\T})$ is the \v{C}ech complex associated to a sufficiently fine invariant open cover $\U$ and $\underline{\TT}$ is the sheaf of continuous functions with values in $\T = {\rm U}(1)$, $\delta : \check{C}^{p, q} \to \check{C}^{p, q+1}$ the usual differential of the \v{C}ech complex, and $\partial : \check{C}^{p, q} \to \check{C}^{p + 1, q}$ is defined by $\partial c = c^{(-1)^p} \cdot \iota^*c$. (One can also obtain this double complex as the \v{C}ech complex associated to a complex of sheaves $\underline{\mathbb{T}} \overset{\partial}{\to} \underline{\mathbb{T}} \overset{\partial}{\to} \underline{\mathbb{T}} \overset{\partial}{\to} \cdots$.)

A $2$-cocycle is then given by
\begin{equation}
(\phi_{ijk}, \rho_{ij}, \tau_i)
\in \check{C}^{0, 2} \oplus \check{C}^{1, 1} \oplus \check{C}^{2, 0}
= \check{C}^2(\U, \underline{\T}) \oplus 
\check{C}^1(\U, \underline{\T}) \oplus \check{C}^0(\U, \underline{\T})\label{eqn:Cech.2.cocycle.decomposition}
\end{equation}
subject to the cocycle conditions 
\begin{align}
\phi_{ijk} \phi_{ikl} &= \phi_{jkl} \phi_{ijl}, \nonumber \\
\phi_{ijk} \cdot \iota^*\phi_{ijk} &= 
\rho_{ij}^{-1} \rho_{ik} \rho_{jk}^{-1}, \nonumber \\
\rho_{ij}^{-1} \cdot \iota^*\rho_{ij} &= \tau_j^{-1} \tau_i, \nonumber \\
\tau_i \cdot \iota^*\tau_i &= 1. \label{eqn:2cocycle.conditions}
\end{align}
Two $2$-cocycles $(\phi_{ijk}, \rho_{ij}, \tau_i)$ and $(\phi'_{ijk}, \rho'_{ij}, \tau'_i)$ define the same cohomology class when there is a \v{C}ech $1$-cochain
$$
(\psi_{ij}, \pi_i) \in \check{C}^{0, 1} \oplus \check{C}^{1, 0}
= \check{C}^1(\U, \underline{\T}) \oplus \check{C}^0(\U, \underline{\T})
$$
such that
\begin{align*}
\phi'_{ijk} \phi_{ijk}^{-1} &= \psi_{ij}^{-1} \psi_{ik} \psi_{jk}^{-1}, \\
\rho'_{ij} \rho_{ij}^{-1} 
&= (\psi_{ij} \cdot \iota^*\psi_{ij}) \pi_j^{-1} \pi_i, \\
\tau'_i \tau_i^{-1} &= \pi_i^{-1} \cdot \iota^*\pi_i.
\end{align*}

\medskip

The \v{C}ech $2$-cocycle representing a `Real' gerbe $\G = (\U, L_{ij}, \phi_{ijk}, J_i, \tau_i, \rho_{ij})$ is given as follows: First of all, we can assume $\V$ is a refinement of $\U$. Taking a further refinement if necessary, we may assume that, in the gerbe $(\V, L_{ij}, \phi_{ijk}, J_i, \tau_i, \rho_{ij})$ that is equivalent to $\G$ through the refinement, the line bundles $L_{ij}$ and $J_i$ are topologically trivial. Choosing trivializations $L_{ij} \cong U_{ij} \times \C$ and $J_i\cong U_i\times \C$, we can identify the unitary isomorphism $\phi_{ijk}$ with a map $\phi_{ijk} : U_{ijk} \to \T$. Similarly, we identify the antiunitary isomorphism $\rho_{ij}$ with $(x, z) \mapsto (\iota(x), \rho_{ij}(x) \bar{z})$, regarding it as a map $\rho_{ij} : U_{ij} \to \T$; the isomorphism $\tau_i$ is also identified with a map $U_i\to\T$ as in Remark \ref{rem:Gao.Hori}. Then we have a \v{C}ech $2$-cocycle $(\phi_{ijk}, \rho_{ij}, \tau_i)$ in the form of Eq.\ \eqref{eqn:Cech.2.cocycle.decomposition}, satisfying the conditions Eq.\ \eqref{eqn:2cocycle.conditions}.

\medskip

\begin{rem}
The above \v{C}ech cohomology description of $H^3_\pm(X)$ is compatible with that in \cite{GSW}, where an invariant open cover is used. On the contrary, our description is different from that in \cite{HMSV1}: there, one needs an open cover $\{ U_i \}_{i \in I}$ with an index set $I$ admitting an involution $\iota : I \to I$ such that $\iota(U_i) \subset U_{\iota(i)}$ (and there is no fixed point in $I$). Therefore it is not direct to relate the \v{C}ech cocycles in this paper with those in \cite{HMSV1}. 
\end{rem}

\begin{rem} \label{rem:equivariant_cohomology}
For $n \ge 1$, the equivariant cohomology $H^n_{\Z_2}(X)$ can similarly be identified with the $(n-1)$th cohomology of the total complex associated to the double complex $(\check{C}^{p, q}, \delta, \partial')$, where $\partial' c = c^{(-1)^{p+1}} \cdot \iota^*c$, see Appendix \ref{appendix:Borel.Cech.untwisted}.
\end{rem}

\subsection{Sign invariants of `Real' gerbes}\label{sec:sign.invariants}

Suppose $Y=Y^\iota$ is a space with trivial involution. As shown in Appendix \ref{appendix:equiv.cohom.fixed}, there is a natural isomorphism
$$
H^3_\pm(Y) \cong H^2(Y; \Z_2) \oplus H^0(Y; \Z_2).
$$
In particular, when $Y$ is path connected, then $H^0(Y; \Z_2) \cong \Z_2$, and this direct summand is identified with the direct summand $H^3_\pm(\pt) \cong \Z_2$ in the decomposition of $H^3_\pm(Y) \cong \tilde{H}^3_\pm(Y) \oplus H^3_\pm(\pt)$.

Following ideas in \cite{HMSV2}, we introduce the following \emph{sign invariants} of a `Real' gerbe.

\begin{dfn}\label{dfn:sign.invariant}
Let $X$ be a space with an involution $\iota : X \to X$. For any `Real' gerbe $\G$ over $X$ and any path-connected subspace $Y \subset X^\iota$, we define
$$
\sigma(\G, Y) \in \{ +1,-1\}$$
to be the sign in the component $H^0(Y; \Z_2) \cong \Z_2=\{\pm 1\}$ of $[\G|_Y] \in H^3_\pm(Y)$ with respect to the decomposition $H^3_\pm(Y) \cong H^2(Y; \Z_2) \oplus H^0(Y; \Z_2)$.
\end{dfn}

The cohomological definition of the sign invariant admits the following equivalent definitions:

\begin{itemize}
\item
Suppose that $\G$ consists of data $(\U, L_{ij}, \phi_{ijk}, J_i, \tau_i, \rho_{ij})$. Then $(J_i, \tau_i)$ gives rise to a $\Z_2$-equivariant line bundle on $U_i$. On $U_i \cap Y$, the $\Z_2$-action on $J_i$ can only be multiplication by $\pm 1$. This sign is independent of the choice of the index $i \in I$, and provides $\sigma(\G, Y)$. In particular, if $\G$ is the `Real' gerbe associated with a $\Z_2$-equivariant line bundle $J \to X$, then the sign of the $\Z_2$-action on $J|_Y$ defines $\sigma(\G, Y)$.

\item
Suppose, as in Remark \ref{rem:Gao.Hori}, that a `Real' gerbe is specified by maps $\tau_i : U_i \to {\rm U}(1)$ such that $\iota^*\tau_i = \tau_i^{-1}$. Thus, on $U_i \cap Y$, this map takes its values in $\{\pm 1\}$. This value is independent of the choice of the index $i \in I$, and defines the sign $\sigma(\G, Y)$.

\end{itemize}

\begin{rem}\label{rem:POS}
In \S\ref{sec:on.the.point}, we explain that $H^3_\pm(\pt)\cong\ZZ_2$ is generated by the gerbe $\G_{\tau_Q}$ induced by the so-called `Quaternionic' 2-cocycle $\tau_Q$, and that it has sign invariant $-1$. Tensoring by $\G_{\tau_Q}$ (which modifies the `Real' DD-invariant by the non-trivial element $\eta\in H^3_\pm(\pt)$) changes the sign invariant at each component $Y\subset X^\iota$. So the product of sign invariants at any \emph{pair} of components $Y_1, Y_2\subset X^\iota$ is a $\ZZ_2$-invariant factoring through $H^3_\pm(X)/H^3_\pm(\pt)$.
\end{rem}

\subsubsection{Completeness of sign invariants in low dimensions}

\begin{prop}\label{prop:low.dim.gerbe.signs}
Let $X$ be a compact and connected smooth manifold, and $\iota : X \to X$ a smooth involution such that the fixed point set $X^\iota$ consists of a finite number of points. Then the homomorphism 
$$
H^3_\pm(X) \to H^3_\pm(X^\iota) \cong H^0(X^\iota; \Z_2)
$$
induced from the inclusion $X^\iota \to X$ is:
\begin{itemize}
\item
surjective if $4 > \dim X$; and

\item
bijective if $3 > \dim X$.

\end{itemize}
\end{prop}

\begin{proof}
We have the exact sequence for the pair $(X, X^\iota)$
$$
\cdots \to
\cdots
H^3_\pm(X, X^\iota) \to
H^3_\pm(X) \to 
H^3_\pm(X^\iota) \to
H^4_\pm(X, X^\iota) \to
\cdots.
$$
Therefore the proposition will follow from $H^n_\pm(X, X^\iota) = 0$ for $n > \dim X$, which is shown below. 

By the so-called slice theorem \cite{Hsiang}, we have open invariant disks $D_x$ centered at $x \in X^\iota$ such that $\overline{D}_x \cap \overline{D}_y = \emptyset$ for $x \neq y$. Let $D = \bigcup_{x \in X^\iota} D_x$ be the (disjoint) union of the disks. We put $X' = X \backslash D$, which is a compact and connected manifold with $\iota$-invariant boundary $\partial X'$. By the excision axiom, we have an isomorphism of groups
$$
H^n_\pm(X, X^\iota) \cong H^n_\pm(X', \partial X').
$$
The latter group fits into the exact sequence (Prop.\ 2.3 of \cite{Gomi})
$$
\cdot\cdot \to
H^n_{\Z_2}(X', \partial X') \overset{f}{\to}
H^n(X', \partial X') \to
H^n_\pm(X', \partial X') \to
H^{n+1}_{\Z_2}(X', \partial X') \to
\cdot\cdot,
$$
with $f$ the map forgetting the $\ZZ_2$-action. Now, the induced involution on $X'$ is fixed point free. Therefore $X'/\Z_2$ is also a compact and connected manifold with boundary $\partial X'/\Z_2$, so that
\begin{align*}
H^n(X', \partial X') &\cong
0, \qquad\qquad\qquad\qquad\qquad\quad (n > \dim X)
\\
H^n_{\Z_2}(X', \partial X') &\cong
H^n(X'/\Z_2, \partial X'/\Z_2) \cong
0. \quad (n > \dim X)
\end{align*}
Hence $H^n_\pm(X, X^\iota) \cong H^n_\pm(X', \partial X') \cong 0$ for $n > \dim X$.
\end{proof}

\begin{cor} \label{cor:classification_via_sign}
Let $X$ be a compact and connected smooth manifold of dimension $dim\, X < 3$, and $\iota : X \to X$ a smooth involution such that the fixed point set $X^\iota$ consists of $N$ points. Then the group of equivalence classes of `Real' gerbes on $X$ is isomorphic to $\Z_2^N$ via the signs $\sigma(\G, x)$, ($x \in X^\iota$). 
\end{cor}

Under the setup of the above corollary, let $x_0 \in X^\iota$ be a fixed point. A `Real' gerbe $\G_{x_0}$ such that 
$$
\sigma(\G_{x_0}, x) =
\left\{
\begin{array}{ll}
-1, & (x_0 = x \in X^\iota) \\
+1, & (x_0 \neq x \in X^\iota)
\end{array}
\right.
$$
can be constructed by taking an invariant open covering $\{ X \backslash \{ x_ 0 \}, D_{x_0} \}$, where $D_{x_0}$ is an invariant open disk around $x_0$ which does not contain any other fixed points; see the 2-torus example in \S \ref{sec:low.examples.gerbe}.

\section{Examples of `Real' gerbes}\label{sec:Real.gerbe.examples}
\subsection{On the point $\pt$}\label{sec:on.the.point}

Let $\pt$ denote the point with trivial involution. Because $H^3_\pm(\pt) \cong \Z_2$, there is a non-trivial `Real' gerbe over $\pt$. This is described by the following data:
\begin{enumerate}
\item
We take $U_0 = \pt$ to form an open cover $\{ U_0 \}$. 

\item
We take $L_{00} \to U_{00}$ to be the product bundle $L_{00} = \C$.

\item
There is no need to consider $\phi_{ijk}$.

\item 
We take $J_0\rightarrow U_0$ to be the product bundle $L_{00} = \C$.

\item
We take $\rho_{00}(z) = \bar{z}$. 

\item
$\tau_0 = \tau_Q(-1, -1)$, where $\tau_Q:\ZZ_2\times\ZZ_2\rightarrow {\rm U}(1)$
is the `Quaternionic' group 2-cocycle
\begin{align*}
\tau_Q(1, 1) &= \tau_Q(-1, 1) = \tau_Q(1, -1) = 1, &
\tau_Q(-1, -1) &= -1.
\end{align*}
\end{enumerate}
Note that the sign invariant of this `Real' gerbe is $-1$.

\medskip

From the general viewpoint of $\phi$-twists, an element in $H^3_\pm(\pt)$ is represented by a `twisted' group $2$-cocycle, namely, $\tau : \Z_2 \times \Z_2 \to {\rm U}(1)$ such that
$$
\tau(h, k)^g \cdot \tau(gh, k)^{-1} \cdot \tau(g, hk) \cdot \tau(g, h)^{-1} 
= 1,\qquad g,h,k\in\ZZ_2,
$$
and the $\tau_Q$ above is a non-trivial 2-cocycle. When $\tau$ is normalized, only $\tau(-1,-1)$ may be non-trivial, and can only equal $\pm 1 $. The latter sign may be identified with the sign invariant of Definition \ref{dfn:sign.invariant}.

\begin{dfn}\label{defn:non-trivial.gerbe.point}
We write $\mathcal{G}_{\tau_Q}$ for the above non-trivial `Real' gerbe on $X=\pt$, and $\eta\in H^3_\pm(\pt)$ its `Real' DD-invariant. For a space with involution $(X,\iota)$, the pullbacks of $\tau_Q, \mathcal{G}_{\tau_Q}, \eta$ and $H^3_\pm(\pt)$ under the collapse map $X\rightarrow\pt$ are denoted with the same symbols.
\end{dfn}

\begin{rem}
Later on, we will generally be interested in `Real' gerbes up to shifts by $\eta$, i.e.\ $H^3_\pm(X)/H^3_\pm(\pt)$. If $X$ has a fixed point $\pt$ regarded as a basepoint, the \emph{reduced} cohomology fits into
\begin{equation*}
0\rightarrow \tilde{H}^3_\pm(X)\rightarrow H^3_\pm(X)\overset{i^*}{\rightarrow} H^3_\pm(\pt)\rightarrow 0,
\end{equation*}
where $i:\pt\rightarrow X$ is the inclusion. The collapse map $X\rightarrow \pt$ induces a splitting $H^3_\pm(X)\cong\tilde{H}^3_\pm(X)\oplus H^3_\pm(\pt)$.
\end{rem}

\subsection{On the circle $\tilde{S}^1$, 2-sphere $\tilde{S}^2$, and 2-torus $\TT^2$}\label{sec:low.examples.gerbe}
Let $\tilde{S}^1$ be the circle with the flip involution, which we identify with the unit circle in the complex plane with complex conjugation involution and basepoint $+1$.

For a $\Z_2$-space $X$ with a fixed point, the smash product $X \wedge \tilde{S}^1$ naturally inherits an involution. Then the following formulae for the reduced cohomology hold true (\S2 of \cite{Gomi}):
\begin{align*}
\tilde{H}^n_\pm(X \wedge \tilde{S}^1) &\cong \tilde{H}^{n-1}_{\Z_2}(X), &
\tilde{H}^n_{\Z_2}(X \wedge \tilde{S}^1) &\cong \tilde{H}^{n-1}_{\pm}(X),
\end{align*}
where $H^n_{\Z_2}(X) = H^n_{\Z_2}(X; \Z)$ is the standard Borel equivariant cohomology. Thus, we can compute the reduced cohomology of $\tilde{S}^d = \underbrace{\tilde{S}^1 \wedge \cdots \wedge \tilde{S}^1}_{d}$ as 
$$
\tilde{H}^3_{\pm}(\tilde{S}^d)
\cong
\left\{
\begin{array}{ll}
\Z_2, & (d = 1, 2) \\
\Z, & (d = 3) \\
0. & (d \ge 4)
\end{array}
\right.
$$

\medskip

{\bf On $\tilde{S}^1$.} 
Since $\tilde{H}^3_\pm(\tilde{S}^1)\cong\ZZ_2$, there exists a non-trivial `Real' gerbe $\G$ over $\tilde{S}^1$. This is specified by the following data:

\begin{enumerate}
\item
We take an invariant open cover $\{ U_+, U_- \}$ of $\tilde{S}^1$ so that each of $U_\pm$ is equivariantly contractible to a fixed point $\pm 1$, and $U_{+-}:=U_+ \cap U_-$ is equivariantly homotopic to the two-point space $\{i,-i\}$ with free involution.

\item
We take $L_{+-}$ to be the product line bundle over $U_{+-}$ (and $L_{-+} = L_{+-}^*$).

\item
We need not consider $\phi_{ijk}$.

\item We take $J_\pm$ to be the product line bundle over $U_\pm$. 

\item
Identifying $L_{+-}$ with the product bundle $\{i,-i\} \times \C \to \{i,-i\}$, we define $\rho_{+-}$ by
\begin{align*}
\rho_{+-} &: \{i,-i\} \times \C \to \{i,-i\} \times \C, &
\rho_{+-}(\pm i, z) &= (\mp i, \pm\bar{z}),
\end{align*}
and similarly for $\rho_{-+}$. 

\item
We take $\tau_+ = \tau_Q(-1, -1)=-1$ and $\tau_- = 1$ to be constant maps, where $\tau_Q$ is the `Quaternionic' $2$-cocycle.

\end{enumerate}
Note that $\rho_{-+}\rho_{+-} = -1$, so the ``clutching'' line bundle $L_{+-} \to U_{+-}$ is a `Quaternionic' line bundle over the two-point space $\{i,-i\}$ with free involution.

\begin{prop}\label{prop:non-triviality.gerbe.S1}
The `Real' gerbe $\G$ on $\tilde{S}^1$ constructed above is non-trivial.
\end{prop}

\begin{proof}
Noting that
\begin{align*}
H^n_\pm(U_\pm) &\cong H^n_\pm(\pt), &
H^n_\pm(U_+ \cap U_-) &\cong H^n(\pt),
\end{align*}
we make use of the Mayer--Vietoris exact sequence for $\{ U_+, U_- \}$,
$$
\overbrace{H^2_\pm(U_+ \cap U_-)}^0 \longrightarrow
\overbrace{H^3_\pm(\tilde{S}^1)}^{\Z_2 \oplus \Z_2} 
\overset{(j_+^*, j_-^*)}{\longrightarrow}
\overbrace{H^3_\pm(U_+) \oplus H^3_\pm(U_-)}^{\Z_2 \oplus \Z_2}  
\longrightarrow
\overbrace{H^3_\pm(U_+ \cap U_-)}^0,
$$
where $j_\pm : U_\pm \to \tilde{S}^1$ are the inclusions. Since $(j_+^*, j_-^*)$ is bijective, it suffices to see that the gerbe constructed on $\tilde{S}^1$ restricts to a non-trivial one in $H^3_\pm(U_+)$ or in $H^3_\pm(U_-)$. It is clear that the restriction to $U_-$ is trivial, but that to $U_+$ is non-trivial.
\end{proof}

We can also show Prop.\ \ref{prop:non-triviality.gerbe.S1} using the sign invariants introduced in \S\ref{sec:sign.invariants}. At the fixed points $\pm 1$, we have
\begin{equation*}
\sigma(\G,\{+1\})=\tau_+=-1,\qquad \sigma(\G,\{-1\})=\tau_-=1.
\end{equation*}

\begin{rem}
If we had chosen $\tau_+ = 1$ and $\tau_- = \tau_Q(-1,-1)=-1$, we would have obtained another non-trivial `Real' gerbe, which is not isomorphic to the previous one. However, they differ only by changing the signs of $\tau_\pm$, i.e.\ multiplying by $\tau_Q=-1$. This modifies the `Real' DD-invariant by $\eta\in H^3_\pm(\pt)$, and is not visible in $\tilde{H}^3_\pm(X)$.
\end{rem}

\medskip

{\bf On $\tilde{S}^2$.}
The sphere $\tilde{S}^2=\tilde{S}^1\wedge\tilde{S}^1$ can be identified with the unit 2-sphere in $\RR^3$ with the involution $(p_0,p_1,p_2)\mapsto (p_0,-p_1,-p_2)$. Since $\tilde{H}^3_\pm(\tilde{S}^2)\cong\ZZ_2$, there is a non-trivial `Real' gerbe over $\tilde{S}^2$, which is not simply obtained by pulling back along $X\rightarrow \pt$. 

Its construction is almost identical to that over $\tilde{S}^1$, the only differences being the following. We take an invariant open cover $\{ U_+, U_- \}$ of $\tilde{S}^2$ so that $U_\pm$ is equivariantly contractible to the fixed point $(\pm 1,0,0)$, and $U_+ \cap U_-$ is equivariantly homotopic to the circle $S^1_{\mathrm{free}} = \{ u \in \C\, |\ \lvert u \rvert = 1 \}$ with the free involution $u \mapsto -u$. Identifying $L_{+-}$ with  $S^1_{\mathrm{free}} \times \C \to S^1_{\mathrm{free}}$, we define
\begin{align*}
\rho_{+-} &: 
S^1_{\mathrm{free}} \times \C \to S^1_{\mathrm{free}} \times \C, &
\rho_{+-}(u, z) &= (-u, u \bar{z}),
\end{align*}
and similarly for $\rho_{-+}$. Since $\rho(-1)_{-+}\rho(-1)_{+-} = -1$, this ``clutching'' line bundle $L_{+-} \to U_{+-} \simeq S^1_{\mathrm{free}}$ is a `Quaternionic' line bundle.

The `Real' gerbe $\G$ described above is non-trivial, since it restricts to the non-trivial gerbe on $\tilde{S}^1 \subset \tilde{S}^2$ (set $p_2=0$). In terms of sign invariants,
\begin{equation*}
\sigma(\G,\{(\pm 1,0,0)\})=\tau_\pm=\mp 1.
\end{equation*}

\medskip

{\bf On $\TT^2$.} Consider the 2-torus $\TT^2 = \tilde{S}^1\times\tilde{S}^1$ with involution $(z_1, z_2) \mapsto (\overline{z_1}, \overline{z_2})$. Pick any one of the four fixed points to be $\pt \in \TT^2$ and consider an open cover $\{ U_0, U_\infty \}$ such that $U_0$ is a $\Z_2$-invariant disk centered at $\pt$ and $U_0 \cap U_\infty \simeq S^1_{\mathrm{free}}$. For this open cover, we can take the remaining data of a `Real' gerbe in the same way as in the case of $\tilde{S}^2$ above. This construction results in a sign invariant $-1$ at $\pt$, and $+1$ at the other three fixed points of $\TT^2$.

It is known \cite{G3} that
$$
H^3_\pm(\TT^2) \cong
\Z_2 \oplus \Z_2 \oplus \Z_2 \oplus \Z_2.
$$
In terms of sign invariants, this fact also follows from Corollary \ref{cor:classification_via_sign}. A basis of $H^3_\pm(\TT^2)$ is obtained by the above construction of a `Real' gerbe localized at each choice of fixed point.

\medskip

\subsection{On the 3-sphere $\tilde{S}^3\cong {\rm SU}(2)$}
When equipped with the sharp involution, Eq.\ \eqref{eqn:sharp}, \eqref{eqn:sharp.unitary.formula}, ${\rm SU}(2)$ is identified with the unit 3-sphere in $\RR^4$ with the involution $(p_0,p_1,p_2,p_3)\mapsto (p_0,-p_1,-p_2,-p_3)$.

Set $U_\pm=\tilde{S}^3\setminus \{(\pm 1,0,0,0)\}$, which is equivariantly homotopic to the fixed point $(\mp 1,0,0,0)$. The double intersection $U_{+-}$ is equivariantly homotopic to the equatorial 2-sphere $S^2_{\rm free}$ with free antipodal involution. Let $L_{+-}\rightarrow U_{+-}$ be the Hopf line bundle. With only two open sets in the cover, $\phi_{ijk}$ need not be considered. We also set $J_\pm\rightarrow U_\pm$ to be the product line bundle.

To specify $\rho_{+-}$, it is convenient to think of $U_{+-}\simeq S^2_{\rm free}$ as the Bloch sphere, so that the fiber $(L_{+-})_{\vect{n}}$ above a unit 3-vector $\vect{n}\in S^2_{\rm free}$ is the $+1$-eigenspace in $\CC^2$ of the spin matrix $\vect{n}\cdot\vect{\sigma}$ along the direction $\vect{n}$. The involution takes $\vect{n}$ to $\vect{n}^\sharp=-\vect{n}$. If $(\vect{n},\xi)\in (L_{+-})_{\vect{n}}$, then 
\begin{equation*}
(\vect{n}^\sharp\cdot\vect{\sigma})(\Theta \xi)=(-\vect{n}\cdot\vect{\sigma})(\Theta \xi)=(\Theta(\vect{n}\cdot\vect{\sigma})\Theta^{-1})(\Theta\xi)=\Theta\xi,
\end{equation*}
so that $\Theta\xi\in (L_{+-})_{\vect{n}^\sharp}$. Thus, we can set $\rho_{+-}:(\vect{n},\xi)\mapsto(\vect{n}^\sharp,\Theta\xi)$, which turns $L_{+-}$ into a `Quaternionic' clutching line bundle. 

Finally, we set $\tau_\pm=\mp 1$, to complete the data of a `Real' gerbe over $\tilde{S}^3={\rm SU}(2)$. This generates $\tilde{H}^3_{\pm}(\tilde{S}^3)\cong\ZZ$.

\medskip
By forgetting the `Real' data, the underlying non-equivariant gerbe is just the basic gerbe over ${\rm SU}(2)$. By restricting to $p_3=0$ (resp.\ $p_2=0=p_3$), we recover the `Real' gerbe over $\tilde{S}^2$ (resp.\ $\tilde{S}^1$) constructed earlier in \S\ref{sec:low.examples.gerbe}.

\subsection{Basic `Real' gerbe on ${\rm SU}(2n)$}\label{sec:basic.Real.SU}
We briefly review the basic (non-equivariant) gerbe on ${\rm SU}(N)$ constructed in \cite{Gaw-R,Mein}: A key to the construction is the fact that the eigenvalues of $u \in {\rm SU}(N)$ are uniquely expressed as $\mathrm{Spec}(u) = \{ e^{2\pi i \lambda_1}, \ldots, e^{2\pi i \lambda_N} \}$ in terms of $\lambda_1, \ldots, \lambda_N \in \R$ such that
\begin{equation}
\lambda_1 \ge \lambda_2 \ge \cdots \ge \lambda_N \ge \lambda_1 - 1,\qquad \sum_{i=1}^N\lambda_i=0.\label{eqn:SU.eigenvalues}
\end{equation}
As a result, for $i = 1, \ldots, N-1$, we have open sets $U_i \subset {\rm SU}(N)$,
$$
U_i = \bigg\{ u \in {\rm SU}(N) \bigg|\
\begin{array}{l}
\mathrm{Spec}(u) =\{ e^{2\pi i \lambda_1}, \ldots, e^{2\pi i \lambda_N} \}, \\
\lambda_1 \ge \cdots \ge \lambda_i > \lambda_{i+1} \ge 
\cdots \ge \lambda_N \ge \lambda_1 - 1
\end{array}
\bigg\},
$$
as well as $U_N \subset {\rm SU}(N)$ by setting
$$
U_N = \bigg\{ u \in {\rm SU}(N) \bigg|\
\begin{array}{l}
\mathrm{Spec}(u) = \{ e^{2\pi i \lambda_1}, \ldots, e^{2\pi i \lambda_N} \}, \\
\lambda_1 \ge \cdots \ge \lambda_N > \lambda_1 - 1
\end{array}
\bigg\}.
$$
These open sets form an open cover $\U = \{ U_1, \ldots, U_N \}$ of ${\rm SU}(N)$. For $i, j$ such that $1 \le i < j \le N$, we can construct a vector bundle $E_{ij} \to U_i \cap U_j$ of rank $j - i$ as a subbundle of the product bundle,
$$
E_{ij} = \bigcup_{u \in U_i \cap U_j}
\bigoplus_{i < p \le j} \mathrm{Ker}(u - e^{2\pi i \lambda_p})
\subset (U_i \cap U_j) \times \C^N.
$$
Taking the determinant we get line bundles $L_{ij} \to U_i \cap U_j$. If $i > j$, then we put $L_{ij} = L_{ji}^*$. Now, we have an obvious isomorphism $\phi_{ijk} : L_{ij} \otimes L_{jk} \to L_{ik}$ on $U_i \cap U_j \cap U_k$ compatible on the quadruple intersections. These data constitute the basic gerbe $\mathcal{G} = (\{ U_i \}, L_{ij}, \phi_{ijk})$ on ${\rm SU}(N)$, whose Diximer--Douady class generates $H^3({\rm SU}(N)) \cong \Z$.

\subsubsection{`Real' structure on the basic gerbe over ${\rm SU}(2n)$}

Suppose that $N = 2n$ is even. In this case, we have the involution on ${\rm SU}(2n)$ induced by a quaternionic structure, $\sharp:u\mapsto \Theta u^* \Theta^{-1}$ (Eq.\ \eqref{eqn:sharp.unitary.formula}). The key to the construction of a `Real' structure on the basic gerbe over ${\rm SU}(2n)$ is:

\begin{lem}
Each of the open sets $U_1, \ldots, U_{2n} \subset {\rm SU}(2n)$ is invariant under the $\sharp$-involution.
\end{lem}

\begin{proof}
Notice that a unitary matrix $u \in U_i$ is characterized by its standard expression of the eigenvalues. It is clear that $\mathrm{Spec}(u) = \mathrm{Spec}(u^\sharp)$. Hence $u \in U_i$ implies $u^\sharp \in U_i$. 
\end{proof}

The product bundle $(U_i \cap U_j) \times \C^{2n} \to U_i \cap U_j$ has the trivial `Quaternionic' structure $(u, \xi) \mapsto (u^\sharp, \Theta \xi) $, which restricts to the subbundle $E_{ij} \to U_i \cap U_j$. As a result, $L_{ij} = \det E_{ij}$ is a `Real' (resp.\ `Quaternionic') line bundle if and only if $\lvert i - j \rvert$ is even (resp.\ odd), and we take $\rho_{ij}$ to be the `Real' (resp.\ `Quaternionic') structure. Finally, set $J_i\rightarrow U_i$ to be the product line bundle, and define $\tau_i : U_i \to {\rm U}(1)$ by
$$
\tau_i =
\left\{
\begin{array}{ll}
-1, & (\mbox{$i$ : odd}) \\
1. & (\mbox{$i$ : even})
\end{array}
\right.
$$

\begin{rem}\label{rem:SU.Pfaffian}
To understand the fixed point set of ${\rm SU}(2n)$, recall that conjugation by $J=\begin{pmatrix} 0 & -\mathbbm{1}_n \\
\mathbbm{1}_n & 0 \end{pmatrix}$ converts $\sharp$ to the involution $u\mapsto -u^{\rm t}$ (Eq.\ \eqref{eqn:sharp.equivalent}). So $u$ is $\sharp$-invariant iff $J^{-1}u$ is skew-symmetric. The skew-symmetric elements of ${\rm SU}(2n)$ form two connected components, distinguished by the Pfaffian. Write 
\begin{equation*}
{\rm SU}(2n)^\sharp_\pm=\{u\in{\rm SU}(2n)\,:\,{\rm Pf}(J^{-1}u)=\pm 1\},
\end{equation*}
then we can check that the sign invariants of the basic `Real' gerbe are
$$
\sigma(\G, {\rm SU}(2n)^\sharp_\pm) = -\mathrm{Pf} = \mp 1.
$$
\end{rem}

\begin{rem}
When labelling the ${\rm SU}(N)$ eigenvalues, we could have imposed $\sum_{i=1}^N{\lambda_i}=1$ instead of $\sum_{i=1}^N{\lambda_i}=0$ as in Eq.\ \eqref{eqn:SU.eigenvalues}. Then $\lambda_N$ needs to be replaced by $\lambda_N+1$, and the decreasing labels $\lambda_i^\prime$ are
\begin{equation*}
\underbrace{(\lambda_N+1)}_{\lambda_1^\prime}\geq\underbrace{\lambda_1}_{\lambda_2^\prime}\geq\underbrace{\lambda_2}_{\lambda_3^\prime}\geq\ldots\geq\underbrace{\lambda_{N-1}}_{\lambda_N^\prime}\geq (\lambda_N+1)-1.
\end{equation*}
\end{rem}
The open sets $U_i^\prime$ defined via the $\lambda_i^\prime$ labels are just a cyclic relabelling of the original $U_i$, with $i\mapsto i+1$ mod $N$. Accordingly, we have
$$
\tau'_i =
\left\{
\begin{array}{ll}
1, & (\mbox{$i$ : odd}) \\
-1, & (\mbox{$i$ : even})
\end{array}
\right. 
$$
which provides a `Real' structure on $\mathcal{G}$ inequivalent to that provided by $\tau_i$. The difference in `Real' DD-invariants is $\eta\in H^3_\pm(\pt)$. Up to this choice, there is a canonical `Real' gerbe over ${\rm SU}(2n)$.

\begin{rem}
For any central element $z \in {\rm SU}(N)$, we can define an involution $\iota$ on ${\rm SU}(N)$ to be $\iota(g) = (zg)^{-1}$. Such involutions and their related ``Jandl gerbes'' were studied in \cite{SSW} and subsequent works, but they are different from the involution considered in this paper. 
\end{rem}

\section{`Real' Fermi gerbe and higher spectral flow}\label{sec:Real.gerbe.spectral.flow}
We will give a spectral construction of `Real' gerbes associated to certain self-adjoint Fredholm families, motivated by topological insulator examples and applications.

\subsection{Spectral flow as an obstruction to gap-opening}\label{sec:sf.gap}
Quite generally, a norm-continuous self-adjoint Fredholm family $F:X\rightarrow \mathcal{F}^{\rm sa}_*$ may exhibit \emph{spectral flow} along loops in $X$. Here, $\mathcal{F}^{\rm sa}_*$ denotes the non-contractible component of the space of self-adjoint Fredholm operators $\mathcal{F}^{\rm sa}$ on a separable Hilbert space, with \emph{both} positive and negative essential spectrum \cite{AS}. Let us recall how this works.

Following \cite{Phillips}, there is a well-defined spectral-flattening map $\varphi$ on $\mathcal{F}^{\rm sa}_*$, which collapses all the positive (resp.\ negative) spectrum outside some small interval $[-\delta,\delta]$ to $+1$ (resp.\ $-1$), and rescales the essential spectral gap to $(-1,1)$. So there remains only finitely many eigenvalues inside $(-1,1)$. This is at first a local construction over neighbourhoods for which $A\mapsto\chi_{[-\delta,\delta]}(A)$ is continuous, finite-rank projection-valued ($\chi$ denotes the characteristic function), and is effected there by applying to $A$ the function which is $+1$ when $t\geq \delta$, $-1$ when $t\leq -\delta$, and interpolates linearly between $-1$ and $+1$ when $t\in(-\delta,\delta)$. Over $\mathcal{F}^{\rm sa}_*$, different choices of neighbourhoods and $\delta$ will be needed, but they can be patched into a global construction over $\mathcal{F}^{\rm sa}_*$ by a partition-of-unity argument, as detailed in the proof of Prop.\ 4 of \cite{Phillips}. This $\varphi$ is shown to implement a homotopy equivalence between $\mathcal{F}^{\rm sa}_*$ and
\begin{equation*}
\hat{F}^\infty_*:=\{A\in \mathcal{F}^{\rm sa}_*\,:\,||A||=1,\;\text{Spec}(A)\;\text{finite},\;\text{ess-spec}(A)=\{-1,+1\}\}.
\end{equation*}
For each $A\in\mathcal{F}^{\rm sa}_*$, the spectrally-flattened $\varphi(A)\in\hat{F}^\infty_*$ is obtainable from $A$ by applying continuous functional calculus.

Given a continuous $F:X\rightarrow \mathcal{F}^{\rm sa}_*$, let $v_F$ denote the unitary map 
\begin{equation}
v_F:{x}\mapsto{\rm exp}\left(\pi i (1+\varphi(F(x))\right)=-{\rm exp}(\pi i(\varphi(F(x)))\in {\rm U}(\infty),\label{eqn:Fredholm.to.unitary}
\end{equation}
whose homotopy class defines an element of $K^{-1}(X)$. This passage from $\mathcal{F}^{\rm sa}_*$ to ${\rm U}(\infty)$, taking $A\mapsto -{\rm exp}(\pi i\varphi(A))$, is a homotopy equivalence \cite{AS, Phillips}. 

Spectral flows of $F$ are due to the ``determinant part'' of $v_F$,
\begin{equation*}
{\rm det}(v_F):X\rightarrow {\rm U}(1).
\end{equation*}
Since ${\rm U}(1)=B\ZZ$, the homotopy class of ${\rm det}(v_F)$ is an element of $H^1(X;\ZZ)\cong {\rm Hom}(H_1(X),\ZZ)$. For a continuous loop $\ell:S^1\rightarrow X$, we can pullback $\ell^*[{\rm det}(v_F)]=[\ell^*({\rm det}(v_F))]\in H^1(S^1;\ZZ)\cong\ZZ$, with the latter isomorphism given by the winding number. According to \cite{Phillips}, this winding number gives the spectral flow of $F(\ell(\cdot))$ along the loop $\ell$.

\medskip
The existence of non-trivial spectral flow along some loop $\ell:S^1\rightarrow X$ in the parameter space, can be thought of as a first obstruction to the existence of a global spectral gap inside the common essential spectral gap of the family $F$. Intuitively, along such a loop in $X$, discrete spectrum robustly connects the negative essential spectrum to the positive essential spectrum, ``flowing'' across 0.

\subsection{Equivariant self-adjoint Fredholm families from half-space Hamiltonians}\label{sec:half.Hamiltonian.family}
Recall from \S\ref{sec:KR.formulation}, that Fourier transform of a gapped Hamiltonian $H$ and its Fermi projection $P_{\rm Fermi}$ gave the Hermitian matrix assignments
\begin{equation*}
\TT^d\ni\vect{z}\equiv ({x};z_d)\mapsto H^{\rm 0D}({x};z_d)\in M_{2n}(\CC),
\end{equation*}
and
\begin{equation*}
P_{\rm Fermi}^{\rm 0D}:({x};z_d)\mapsto \frac{1-{\rm sgn}(H^{\rm 0D}({x};z_d))}{2}.
\end{equation*}
Here we have abbreviated the first $d-1$ coordinates of $\TT^d$ into a single symbol $x=(z_1,\ldots,z_{d-1})$.
If we only Fourier transformed $H$ along the first $d-1$ directions, we get a norm-continuous assignment of gapped self-adjoint operators
\begin{equation*}
\TT^{d-1}\ni{x}\mapsto H^{\rm 1D}({x})\in M_{2n}(C^*_{r}(\ZZ))\subset\mathcal{B}(V\otimes\ell^2(\ZZ)),
\end{equation*}
and projections
\begin{equation*}
P^{\rm 1D}_{\rm Fermi}:{x}\mapsto\frac{1-{\rm sgn}(H^{\rm 1D}({x}))}{2}\in M_{2n}(C^*_{r}(\ZZ))\subset\mathcal{B}(V\otimes\ell^2(\ZZ)),
\end{equation*}
each acting on the ``1D Hilbert space'' $V\otimes\ell^2(\ZZ)$ transverse to the first $d-1$ directions.

\medskip

The truncated $\widetilde{H}$ on the half-space can likewise be Fourier transformed in the first $d-1$ directions, giving a truncated self-adjoint family,
\begin{equation*}
\TT^{d-1}\ni{x}\mapsto \widetilde{H}^{\rm 1D}({x})\in M_{2n}(C^*_{r}(\NN))\subset \mathcal{B}(V\otimes\ell^2(\NN)),
\end{equation*}
in which $\widetilde{H}^{\rm 1D}({x})=\jmath^*\circ H^{\rm 1D}({x})\circ\jmath$ is the compression of $H^{\rm 1D}(x)$ to the transverse ``half-line Hilbert space'' $V\otimes\ell^2(\NN)$. 
Because the spectrum of $\widetilde{H}^{\rm 1D}({x})$ modulo the compacts on $V\otimes\ell^2(\NN)$ is the spectrum of $H^{\rm 1D}({x})$ (recall Eq.\ \eqref{eqn:real.Toeplitz}), and the latter is gapped, we see that the essential spectrum of each $\widetilde{H}^{\rm 1D}({x})$ is gapped. Thus $\widetilde{H}^{\rm 1D}$ is a continuous self-adjoint \emph{Fredholm} family
\begin{equation*}
\widetilde{H}^{\rm 1D}:\TT^{d-1}\rightarrow \mathcal{F}^{\rm sa}_*.
\end{equation*}

\medskip

\subsubsection{Time-reversal invariance and equivariant self-adjoint Fredholm families}
Suppose $H$, thus also $P_{\rm Fermi}$, is time-reversal invariant, i.e., it commutes with the quaternionic structure $\Theta$ from Eq.\ \eqref{eqn:standard.quaternionic} (extended to $V\otimes \ell^2(\ZZ^d)$). Then $\widetilde{H}$ also commutes with $\Theta$ acting on \emph{$V\otimes \ell^2(\ZZ^{d-1}\times\NN)$}. When we work with the Fourier transformed Hilbert space $L^2(\TT^{d-1})\otimes (V\otimes \ell^2(\NN))$, the quaternionic structure becomes $\widehat{\Theta}=\Theta\circ\iota_{\rm flip}$ where $\iota_{\rm flip}:x\mapsto \overline{x}$ and $\Theta$ is the standard one on $V\otimes \ell^2(\NN)$. Explicitly, this means that the family $\widetilde{H}^{\rm 1D}$ satisfies the `Quaternionic' condition
\begin{equation}
\widetilde{H}^{\rm 1D}({x})={\rm Ad}_\Theta (\widetilde{H}^{\rm 1D}(\overline{{x}})) =J\overline{\widetilde{H}^{\rm 1D}(\overline{{x}})}J^{-1}.\label{eqn:H.equivariance}
\end{equation}
The operation ${\rm Ad}_\Theta$ preserves the spectrum of self-adjoint operators on $V\otimes\ell^2(\NN)$, thus it induces an involution on $\mathcal{F}^{\rm sa}_*$. Then Eq.\ \eqref{eqn:H.equivariance} is simply restated as:
\begin{prop}\label{prop:H.equivariance}
Let $H$ be a gapped time-reversal invariant local Hamiltonian acting on $V\otimes \ell^2(\ZZ^d)$, and $\widetilde{H}^{\rm 1D}$ be its Fourier transform in the first $d-1$ directions. Then $\widetilde{H}^{\rm 1D}:\TT^{d-1}\rightarrow\mathcal{F}^{\rm sa}_*$ is a continuous equivariant self-adjoint family with respect to the flip involution on $\TT^{d-1}$ and ${\rm Ad}_\Theta$ on $\mathcal{F}^{\rm sa}_*$.
\end{prop}

\subsubsection{$K$-theory connecting map --- Fredholm formulation}
Recall from Eq.\ \eqref{eqn:sharp.unitary.formula} that we also have the involution $u\mapsto \Theta u^*\Theta^{-1}$ on ${\rm U}(\infty)$. Since
\begin{align*}
{\rm exp}(\pi i \varphi({\rm Ad}_\Theta A ))={\rm exp}(\Theta (-\pi i\varphi(A))\Theta^{-1})&=\Theta {\rm exp}(-\pi i\varphi(A)) \Theta^{-1}\\ 
&=\Theta ({\rm exp}(\pi i\varphi(A)))^*\Theta^{-1},
\end{align*}
the homotopy equivalence $\mathcal{F}^{\rm sa}_*\simeq {\rm U}(\infty)$ is actually $\ZZ_2$-equivariant. 

Thus, if we are given an involutive space $(X,\iota)$ and an equivariant self-adjoint Fredholm family $F:X\rightarrow \mathcal{F}^{\rm sa}_*$, the unitary map $v_F$ of Eq.\ \eqref{eqn:Fredholm.to.unitary} gives an equivariant map $X\rightarrow {\rm U}(\infty)$, whose equivariant homotopy class defines a $K$-theory element (recall Eq.\ \eqref{eqn:KR3.classifying.map}),
\begin{equation*}
v_F\in [X,{\rm U}(\infty)]_{\ZZ_2}=KR^{-3}(X).
\end{equation*}

\begin{dfn}\label{defn:boundary.K.class}
Let $H$ be a time-reversal invariant gapped Hamiltonian acting on $V\otimes \ell^2(\ZZ^d)$. The resultant equivariant family $\widetilde{H}^{\rm 1D}:\TT^{d-1}\rightarrow\mathcal{F}^{\rm sa}_*$ of Prop.\ \ref{prop:H.equivariance} determines a class in $KR^{-3}(\TT^{d-1})\cong KO_3(C^*_{r,\RR}(\ZZ^{d-1}))$, which we call the boundary $K$-theory class of $H$.
\end{dfn}

We saw in Eq.\ \eqref{eqn:exponential.map.q} that the connecting map \mbox{$\delta_4:KO_4(C^*_{r,\RR}(\ZZ^d))\rightarrow KO_3(C^*_{r,\RR}(\ZZ^{d-1}))$} applied to the Fermi projection $P_{\rm Fermi}$ for a $\Theta$-invariant $H$, has the formula
\begin{equation*}
\delta_4[P_{\rm Fermi}]=[{\rm exp}(-2\pi i\tilde{P}_{\rm Fermi})],
\end{equation*}
where $\tilde{P}_{\rm Fermi}$ is a(ny) self-adjoint lift of $P_{\rm Fermi}$ in $M_{2n}(C^*_{r}(\ZZ^{d-1})\otimes C^*_{r}(\NN))$ satisfying the $\Theta$-invariance condition $(\tilde{P}_{\rm Fermi})^\flat=\tilde{P}_{\rm Fermi}$ (equivalently, $\widehat{\Theta}$-invariance after Fourier transforming $\ZZ^{d-1}$ to $\TT^{d-1}$).

Let us rewrite $\delta_4$ in terms of the half-space Hamiltonian $\widetilde{H}$ regarded as an equivariant family $\widetilde{H}^{\rm 1D}:\TT^{d-1}\rightarrow \mathcal{F}^{\rm sa}_*$ after Fourier transform. As discussed in \S\ref{sec:sf.gap}, the spectral flattening $\varphi$ is implemented by functional calculus with some real-valued function. It commutes with ${\rm Ad}_\Theta$, and also \mbox{$q:M_{2n}(C^*_r(\NN))\rightarrow M_{2n}(C^*_{r}(\ZZ))$}. Thus
\begin{align*}
q\left(\frac{1-\varphi(\widetilde{H}^{\rm 1D}({x}))}{2}\right)&=\frac{1-\varphi(q(\widetilde{H}^{\rm 1D}({x})))}{2}\\
&=\frac{1-\varphi(H^{\rm 1D}({x}))}{2}=\frac{1-{\rm sgn}(H^{\rm 1D}({x}))}{2}=P^{\rm 1D}_{\rm Fermi}({x}),\quad x\in\TT^{d-1},
\end{align*}
so we have a desired $\widehat{\Theta}$-invariant lift $\frac{1-\varphi(\widetilde{H}^{\rm 1D}({\cdot}))}{2}$ of $P^{\rm 1D}_{\rm Fermi}({\cdot})$.
Now applying the exponential map,
\begin{align*}
{\rm exp}\left(-2\pi i \left(\frac{1-\varphi(\widetilde{H}^{\rm 1D}({x}))}{2}\right)\right)&={\rm exp}\left(\pi i (1+\varphi(\widetilde{H}^{\rm 1D}({x})))\right)\nonumber \\
&=-{\rm exp}\left(\pi i\varphi(\widetilde{H}^{\rm 1D}({x}))\right)\in {\rm U}(\infty).
\end{align*}
Thus the connecting map $\delta_4$ can be rewritten as
\begin{equation*}
\delta_4[P_{\rm Fermi}]=\left[x\mapsto -{\rm exp}\left(\pi i \varphi(\widetilde{H}^{\rm 1D}(x))\right)\right]\in [\TT^{d-1}, {\rm U}(\infty)]_{\ZZ_2}=KR^{-3}(\TT^{d-1}).
\end{equation*}
Under the equivariant homotopy equivalence $\mathcal{F}^{\rm sa}_*\simeq {\rm U}(\infty)$ described in \ref{sec:equiv.sa.Fredholm}, we equivalently have
\begin{equation*}
\delta_4[P_{\rm Fermi}]=\left[x\mapsto \widetilde{H}^{\rm 1D}(x)\right]\in [\TT^{d-1},\mathcal{F}^{\rm sa}_*]_{\ZZ_2}\cong KR^{-3}(\TT^{d-1}).
\end{equation*}
To summarize: 
\begin{cor}\label{cor:BBC}
\emph{$\delta_4$ maps the bulk Fermi projection $P_{\rm Fermi}$ to the half-space Hamiltonian $\widetilde{H}^{\rm 1D}$ at the level of their $K$-theory classes.}
\end{cor}

\begin{rem}
For the physical interpretation, it is important that we may replace $\widetilde{H}^{\rm 1D}$ above by a more realistic $\widetilde{H}^{\rm 1D}+V$ where $V$ is any ($\widehat{\Theta}$-invariant) self-adjoint perturbation term in the kernel of $q$. Such a $V$ is a family of compact operators parametrized by $\TT^{d-1}$, thus may be thought of as a ``transversally compact'' time-reversal invariant boundary perturbation. ``Turning on $V$'' does not change the homotopy invariants of $\widetilde{H}^{\rm 1D}$ such as the $K$-theory class. So conclusions drawn from computations involving such invariants are stable against such perturbations.
\end{rem}

\subsubsection{Equivariant self-adjoint Fredholm families and vanishing spectral flows}\label{sec:equiv.sa.Fredholm}

Abstracting the above discussion, consider a compact, connected space $X$ with involution $\iota$, and a continuous equivariant map
\begin{equation*}
F:X\rightarrow \mathcal{F}^{\rm sa}_*.
\end{equation*}
By the Fredholm condition, there exists some common essential spectral gap for the family, say $(-\epsilon,\epsilon)$. We shall think of $F$ as an ``abstract $\Theta$-invariant half-space Hamiltonian which has a gap $(-\epsilon,\epsilon)$ in the bulk''. 

The full spectrum of the half-space Hamiltonian is the union of the spectra of $F(x)$ over $x\in X$, and  includes contributions from the discrete spectra of the $F(x)$ inside $(-\epsilon,\epsilon)$. Thus the half-space Hamiltonian will not generally retain the entire interval $(-\epsilon,\epsilon)$ as a spectral gap. Our interest is in whether any spectral gap inside $(-\epsilon,\epsilon)$ persists at all.

A first obstruction could occur at the level of spectral flows, which for equivariant self-adjoint Fredholm families $F$, are equivariant spectral flows, corresponding to a class ${\rm Sf}(F)\in H^1_{\ZZ_2}(X)$. If ${\rm Sf}(F)$ is non-trivial, gap-filling occurs. 

In some situations, we automatically have ${\rm Sf}(F)=0$, and we cannot conclude anything about gap-filling from such an absence of spectral flow. 
For example, the involutive space $\TT^{d-1}$ has $H^1_{\ZZ_2}(\TT^{d-1})=0$. Let us therefore look for a secondary obstruction to spectral gap-opening for the equivariant family $F$.

\subsection{`Real' Fermi gerbe obstruction to spectral gap opening}\label{sec:Real.Fermi}
Even if $F$ has no spectral flows at all, it could still be possible that a global spectral gap is prevented from existing inside $(-\epsilon,\epsilon)$. That a higher obstruction exists in the form of a ``Fermi gerbe'' associated to $F$, was observed in \cite{C-T}. There, it was explained that the Dixmier--Douady invariant of the Fermi gerbe detects the gap-filling property for $F$. 
This is a \emph{non-equivariant} obstruction living in $H^3(X)$, occuring only when ${\rm dim}(X)\geq 3$.

For equivariant self-adjoint Fredholm families, the Fermi gerbe construction initially proceeds in the same way as in \cite{C-T}, but it acquires the extra structure of a `Real' gerbe under a certain assumption. As we saw in \S\ref{sec:Real.gerbe.examples}, there are non-trivial `Real'  gerbes even if ${\rm dim}(X)<3$, and this is important for physically important examples with $X=\TT^{d-1}, d=2,3$.

\subsubsection{Underlying non-equivariant Fermi gerbe}
In \cite{C-T}, the \emph{Fermi gerbe} was constructed from a continuous family of self-adjoint Fredholm operators parametrized by a compact space $X$. Without loss of generality, the common essential spectral gap of the family is assumed to be $(-1,1)$.

As a warm-up, let us reformulate the spectral flow of $F$ as follows. For each $\mu \in (-1, 1)$, define the open set
$$
U_\mu = \{ x \in X |\ \mu \not\in \mathrm{Spec}(F(x)) \},
$$
so that $\U=\{U_\mu\}_{\mu\in(-1,1)}$ covers $X$. For $\mu, \nu \in (-1, 1)$ such that $\mu < \nu$, let $E_{\mu \nu} \to U_\mu \cap U_\nu$ be the eigenspace bundle of $F$ for the eigenvalues lying inside $(\mu, \nu)$. On $U_\mu \cap U_\nu \cap U_\xi$ with $\mu < \nu < \xi$, there is a natural isomorphism $E_{\mu \nu} \oplus E_{\nu \xi} \cong E_{\mu \xi}$. Thus, the rank $r_{\mu \nu} = \mathrm{rank}E_{\mu \nu}$ defines a locally constant function $r_{\mu \nu} : U_\mu \cap U_\nu \to \Z$, and we get a \v{C}ech cocycle $r = (r_{\mu \nu}) \in \check{Z}^1(\U; \Z)$, and upon passing to refinements, a first cohomology class ${\rm Sf}(F)\in H^1(X)$.

Now suppose there exists some $\lambda_0 \in (-1, 1)$ such that $\lambda_0 \not\in \mathrm{Spec}(F(x))$ for all $x\in X$. Then $U_{\lambda_0} = X$ is enough to cover $X$, the \v{C}ech cocycle $r$ trivializes, and so $\mathrm{Sf}(F) = 0$. In other words, if $\mathrm{Sf}(F) \neq 0$, then we can conclude that no such global gap $\lambda_0 \in (-1, 1)$ can occur in the ${\rm Spec}(F(x))$. As the spectral flow is homotopy invariant, for any other map $F'$ homotopic to $F$, the gap-filling still occurs for $F'$.

\medskip

Conversely, suppose $\mathrm{Sf}(F) = 0$. This does \emph{not} necessarily imply that there exists a global gap $\lambda_0 \in (-1, 1)$ in ${\rm Spec}(F(x))$. Here, a gerbe obstruction can occur. With $\U=\{U_\mu\}_{\mu\in(-1,1)}$ as before, take the determinant line bundles $L_{\mu \nu} = \det E_{\mu \nu}$ over $U_\mu\cap U_\nu$. There are isomorphisms $\phi_{\mu \nu \xi}:L_{\mu\nu}\otimes L_{\nu\xi}\rightarrow L_{\mu\xi}$ on triple intersections, induced from $E_{\mu \nu} \oplus E_{\nu \xi} \cong E_{\mu \xi}$. We get a \v{C}ech cocycle in $\check{Z}^2(\U;\underline{{\rm U}(1)})$ defining a gerbe, and then its Dixmier--Douady (DD) invariant in $H^3(X)$. 

\medskip
Suppose $\mu_0\in(-1,1)$ is a global spectral gap for $F$. Then $U_{\mu_0}$ suffices to cover $X$ and globally trivialize the Fermi gerbe. Thus a non-vanishing DD-invariant of the Fermi gerbe detects gap-filling.

\subsubsection{`Real' structure on Fermi gerbe}
 When $F$ is an equivariant self-adjoint family, the $U_\mu$ would give an invariant cover of $X$, and the equivariant spectral flow gives a class  ${\rm Sf}(F)\in H^1_{\ZZ_2}(X)$ obstructing the existence of a global gap in ${\rm Spec}(F(x))$.  Suppose ${\rm Sf}(F)=0$; we shall look for the next obstruction in the form of a `Real' Fermi gerbe.

We have the invariant open cover $\{ U_\mu \}$, line bundles $L_{\mu \nu} = \det E_{\mu \nu}$, and also isomorphisms $\phi_{\mu \nu \xi}$, as in the non-equivariant Fermi gerbe (We take $J_\mu$ to be trivial.) The `Quaternionic' structure on the Hilbert space restricts to each $E_{\mu\nu}$, and induces a `Real' or `Quaternionic' structure $\rho_{\mu\nu}$ on the determinant line bundles $L_{\mu\nu}$ depending on whether the rank $r_{\mu\nu}$ of $E_{\mu\nu}$ is even or odd. Here, we observe that $\rho_{\mu\nu}$ furnishes an antiunitary isomorphism between $L_{\mu\nu}$ and $\iota^*L_{\mu\nu}$, forming part of the data of a `Real' gerbe. In a local trivialization, this would be implemented as
\begin{equation*}
\rho_{\mu\nu}:(x,z)\mapsto(\iota(x),\rho_{\mu\nu}(x)\overline{z}),
\end{equation*}
with $\rho_{\mu\nu}(x)\in{\rm U}(1)$. Squaring it gives
\begin{equation*}
\pm {\rm id}=\rho_{\mu\nu}^2:(x,z)\mapsto(x,\rho_{\mu\nu}(\iota(x))\overline{\rho_{\mu\nu}(x)}z)
\end{equation*}
according to the parity of $r_{\mu\nu}$.
Thus $\rho_{\mu\nu}^{-1}\iota^*\rho_{\mu\nu}=\pm 1=e^{\pi i r_{\mu\nu}}$.

To complete the data of a `Real' Fermi gerbe, the remaining datum \mbox{$\tau_\mu : U_\mu \to {\rm U}(1)$} satisfying
\begin{equation*}
\rho_{\mu\nu}^{-1}\iota^*\rho_{\mu\nu}=e^{\pi i r_{\mu\nu}}=\tau_\nu^{-1}\tau_\mu,\qquad \tau_\mu\iota^*\tau_\mu=1,
\end{equation*}
(i.e.\ the remaining cocycle conditions Eq.\ \eqref{eqn:2cocycle.conditions})
needs to be provided. 

In general, this last datum is not available, but under the assumption $\mathrm{Sf}(F) = 0$, we can always find such a desired $\tau_\mu$, after possibly passing to a refinement, see Lemma \ref{lem:vanishing.equivariant.H1} in the Appendix. Even so, the choice of $\tau_\mu$ is not unique, and is not determined by the spectral data of $F$. Nevertheless, when $X$ is connected the ambiguity is just an overall sign in the collection of $\tau_\mu$. Thus the DD-class of the resulting `Real' gerbe is well-defined up to $\eta\in H^3_\pm(\pt) \cong \Z_2$, see Remark \ref{rem:vanishing.equivariant.H1}.

As in the non-equivariant case, if the above `Real' Fermi gerbe has non-trivial `Real' DD-invariant in $H^3_\pm(X)/H^3_\pm({\rm pt})$, then there cannot exist any $\lambda_0 \in (-1, 1)$ such that $\lambda_0 \not\in \mathrm{Spec}(F(x))\;\forall x\in X$, to wit: gap-filling occurs. This `Real' gerbe obstruction is an equivariant homotopy invariant of $F$.

\subsubsection*{Sign invariants of `Real' Fermi gerbe}
Let $y$ be a fixed point of $(X,\iota)$. Then $F(y)$ is a $\Theta$-invariant self-adjoint Fredholm operator. Suppose $\mu_y\in(-1,1)$ is an eigenvalue of $F(y)$, then it has even multiplicity (``Kramers degeneracy''). Also, $y\in U_{\mu_y\pm\varepsilon}$ for sufficiently small $\varepsilon$ since $\mu_y$ is isolated in the spectrum of $F(y)$. At $y$, we have $\tau_{\mu_y-\varepsilon}=\tau_{\mu_y+\varepsilon}$ since $r_{\mu_y-\varepsilon,\mu_y+\varepsilon}$ is even, and this value of $\tau_{\mu_y\pm\varepsilon}$ is the sign invariant of $\G_F$ at $\{y\}$. This ``invariant'' inherits an ambiguity from that of $\tau_\mu$, but the product of sign invariants at any \emph{pair} of components of $X^\iota$ is unambiguous, see Remark \ref{rem:POS}.

Some examples of `Real' Fermi gerbes and their sign invariants, are given in Fig.\ \ref{fig:phase.function3} and Fig.\ \ref{fig:phase.function}.

\subsection{Relation between $KR^{-3}$ and `Real' gerbes}

\begin{lem}\label{lem:U.to.SU} 
Let $(X,\iota)$ be an involutive space such that $H^1_{\ZZ_2}(X)=0$. Suppose $q:X\rightarrow {\rm U}(2n)$ is an equivariant map, i.e., $q(\iota(x))=q(x)^\sharp$. Then $q$ is equivariantly homotopic to a map $X\rightarrow {\rm SU}(2n)$.
\end{lem}
\begin{proof}
We have $ [{\rm det}(q)]\in [X,{\rm U}(1)_{\rm triv}]_{\ZZ_2}\cong H^1_{\ZZ_2}(X)=0$ where ${\rm U}(1)$ is given the trivial involution, see e.g.\ Lemma \ref{lem:H1.equivariant}. So there exists a homotopy ${g}:X\times [0,1]\rightarrow {\rm U}(1)$ such that 
\begin{equation*}
{g}(x,0)=1,\quad {g}(x,1)={\rm det}\,q(x)^{-1},\quad {g}(x,t)={g}(\iota(x),t),\quad \forall x\in X, t\in[0,1].
\end{equation*}
For each $t\in[0,1]$, the map ${g}(\cdot,t):X\rightarrow{\rm U}(1)$ has vanishing winding numbers, so we may choose a continuous $2n$-th root function ${h}(\cdot,t):X\rightarrow{\rm U}(1)$,
\begin{equation*}
{h}(x,t)={h}(\iota(x),t),\quad ({h}(x,t))^{2n}=g(x,t),\quad \forall x\in X, t\in [0,1].
\end{equation*}
For $t\in[0,1]$, set
\begin{equation*}
\tilde{q}(x,t)=q(x) (h(x,t)\mathbbm{1}_{2n}),\quad x\in X,
\end{equation*}
to obtain the desired equivariant homotopy from $q$ to $\tilde{q}(\cdot,1)\in [X,{\rm SU}(2n)]_{\ZZ_2}$.
\end{proof}

\medskip

Recall that a $KR^{-3}(X)$ class may be represented by an equivariant map $X\rightarrow {\rm U}(2n)$ for some $n$. Lemma \ref{lem:U.to.SU} says that if $H^1_{\ZZ_2}(X)=0$, the representative classifying map may be taken to land in ${\rm SU}(2n)$, and the `Real' basic gerbe over ${\rm SU}(2n)$ may be pulled back to $X$. Despite this, we do not obtain a well-defined map from $KR^{-3}(X)$ to $H^3_\pm(X)$. This is because two different choices of representative ${\rm SU}(2n)$-valued classifying maps may be equivariantly homotopic only over ${\rm U}(2n)$-valued maps rather than over ${\rm SU}(2n)$-valued maps. The basic gerbe cannot generally be pulled back over the homotopy, and it is possible that different choices of representative ${\rm SU}(2)$-valued classifying maps lead to inequivalent pullback `Real' gerbes.

The ambiguity can be seen with a simple example. Take $X=\pt$. Then $q_t:\pt\mapsto e^{\pi it}\mathbbm{1}_2\in{\rm U}(2)$ gives a homotopy between the constant ${\rm SU}(2)$-valued maps $q_0=\mathbbm{1}_2$ and $q_1=-\mathbbm{1}_2$. The basic `Real' gerbe on ${\rm SU}(2)$ restricts to the trivial `Real' gerbe over $\{\mathbbm{1}_2\}$, but restricts to the non-trivial `Real' gerbe over $\{-\mathbbm{1}_2\}$. Yet $q_0$ and $q_1$ represent the same (trivial) $KR^{-3}(\pt)$ class.

For a slightly more interesting example, recall that the non-trivial `Real' gerbe over $\tilde{S}^1$ is a restriction of the basic `Real' gerbe over ${\rm SU}(2)$; thus it has the classifying map
\begin{equation*}
q_0:\tilde{S}^1\ni z\mapsto \begin{pmatrix} z & 0 \\ 0 &\overline{z}\end{pmatrix}\in{\rm SU}(2).
\end{equation*}
Setting $q_t:=e^{\pi it}q_0$ gives a homotopy, through ${\rm U}(2)$-valued maps, from $q_0$ to $q_1=-q_0$. 
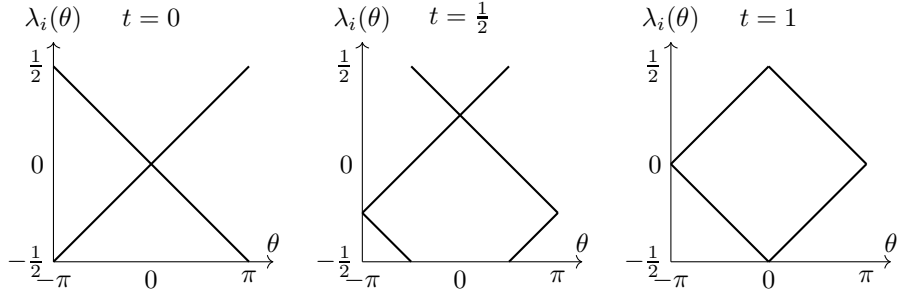
\begin{figure}[h!]
\begin{tikzpicture}[scale=0.65]
\draw[->] (0,0) -- (4.5,0) node[anchor=south] {$\theta$};
\draw[->] (0,0) -- (0,4.5) node[anchor=south] {$\lambda_i(\theta)$};
\draw[thick] (0,4) -- (4,0);
\draw[thick] (0,0) --  (4,4);
\node[below] at (2,0) {$0$};
\node[left] at (0,2) {$0$};
\node[left] at (0,0) {$-\frac{1}{2}$};
\node[below] at (0,0) {$-\pi$};
\node[below] at (4,0) {$\pi$};
\node[left] at (0,4) {$\frac{1}{2}$};
\node at (2,5) {$t=0$};
\end{tikzpicture}~
\begin{tikzpicture}[scale=0.65]
\draw[->] (0,0) -- (4.5,0) node[anchor=south] {$\theta$};
\draw[->] (0,0) -- (0,4.5) node[anchor=south] {$\lambda_i(\theta)$};
\draw[thick] (1,4) -- (4,1);
\draw[thick] (0,1) -- (1,0);
\draw[thick] (0,1) -- (3,4);
\draw[thick] (3,0) -- (4,1);
\node[below] at (2,0) {$0$};
\node[left] at (0,2) {$0$};
\node[left] at (0,0) {$-\frac{1}{2}$};
\node[below] at (0,0) {$-\pi$};
\node[below] at (4,0) {$\pi$};
\node[left] at (0,4) {$\frac{1}{2}$};
\node at (2,5) {$t=\frac{1}{2}$};
\end{tikzpicture}~
\begin{tikzpicture}[scale=0.65]
\draw[->] (0,0) -- (4.5,0) node[anchor=south] {$\theta$};
\draw[->] (0,0) -- (0,4.5) node[anchor=south] {$\lambda_i(\theta)$};
\draw[thick] (0,2) -- (2,0);
\draw[thick] (2,4) -- (4,2);
\draw[thick] (0,2) --  (2,4);
\draw[thick] (2,0) --  (4,2);
\node[below] at (2,0) {$0$};
\node[left] at (0,2) {$0$};
\node[left] at (0,0) {$-\frac{1}{2}$};
\node[below] at (0,0) {$-\pi$};
\node[below] at (4,0) {$\pi$};
\node[left] at (0,4) {$\frac{1}{2}$};
\node at (2,5) {$t=1$};
\end{tikzpicture}
\caption{For the unitary maps $q_t:\tilde{S}^1\ni e^{i\theta}\equiv z\mapsto {\rm diag}(e^{\pi i t}z, e^{\pi i t}\overline{z})={\rm diag}(e^{i(\pi t+\theta)},e^{i(\pi t-\theta)})$, the arguments $\lambda_1,\lambda_2$ of their eigenvalue functions $e^{2\pi i\lambda_1(\theta)}, e^{2\pi i\lambda_2(\theta)}$ at $t=0,\frac{1}{2},1$ are plotted. For the `Real' gerbe at $t=0$ regarded as a restriction of the basic `Real' gerbe on ${\rm SU}(2)$, the open set $U_1=\{\theta:\lambda_1(\theta)>\lambda_2(\theta)\}$ is $(0,2\pi)$, while \mbox{$U_2=\{\theta:\lambda_2(\theta)>\lambda_1(\theta)-1\}$} is $(-\pi,\pi)$. For the `Real' gerbe at $t=1$, the open set labelling $U_i$, and thus the attached sign $\tau_i$, is exchanged. Notice that the homotopy moves the Dirac point (eigenvalue crossing) from $\theta=0$ to $\theta=\pi$. Furthermore, at intermediate $t=\frac{1}{2}$, there are \emph{two} Dirac points: one at $\theta=0$ and one at $\theta=\pi$.}
\label{fig:phase.function2}
\end{figure}
It is not hard to see that $q_1$ and $q_0$ define inequivalent pullback `Real' gerbes on $\tilde{S}^1$, differing by the non-trivial element of $H^3_\pm(\pt)$, see Fig.\ \ref{fig:phase.function2}.

\begin{rem}
While not detailed here, one may compute that in contrast to ${\rm SU}(2n)$, we have $H^1_{\ZZ_2}({\rm U}(2n))\neq 0$, so that the non-vanishing spectral flows provide an obstruction to the existence of `Real' structures on the basic gerbe over ${\rm U}(2n)$. In other words, there is no ``universal'' `Real' gerbe over \emph{all} of ${\rm U}(\infty)$ to be pulled back under a general $KR^{-3}$-classifying map.
\end{rem}

\medskip

In the Fredholm formulation, a $KR^{-3}(X)$ class is represented by an equivariant self-adjoint Fredholm family. As we saw in \S\ref{sec:Real.Fermi}, with $X$ compact, connected, and $H^1_{\ZZ_2}(X)=0$ satisfied, the `Real' Fermi gerbe $\G_F$ for any equivariant $F:X\rightarrow \mathcal{F}^{\rm sa}_*$ is well-defined, up to a choice of the $\tau_\mu$ corresponding to a $H^3_\pm(\pt)$ ambiguity. Then taking the `Real' DD-invariant of $\G_F$ modulo $H^3_\pm(\pt)$ gives a well-defined homomorphism
\begin{equation}
DD:KR^{-3}(X)\rightarrow H^3_\pm(X)/H^3_\pm(\pt).\label{eqn:KR.to.DD}
\end{equation}
Also, for any pair of path-connected subsets $Y_1,Y_2\subset X^\iota$, the product of their sign invariants makes sense,
\begin{equation}
\sigma_{Y_1,Y_2}:KR^{-3}(X)\rightarrow H^3_\pm(X)/H^3_\pm(\pt)\rightarrow\ZZ_2,\qquad F\mapsto \sigma(\G_F,Y_1)\sigma(\G_F,Y_2),\label{eqn:POS.from.KR}
\end{equation}
independently of the indeterminacy of $\tau_\mu$.

\section{Applications to topological insulators}
\subsection{Non-triviality of `Real' Fermi gerbe}\label{sec:non.triviality}
{\bf Chern insulator.} In complex $K$-theory, we have $\widetilde{K}^0(\TT^2)\cong H^2(\TT^2)\cong\ZZ$. The generator can be represented by the Hopf line bundle over $S^2$ pulled back to $\TT^2$ under a degree-1 map. In terms of $C^*$-algebras, there is a projection $P_{\rm Chern}\in M_2(C(\TT^2))\cong M_2(C^*_r(\ZZ^2))$ representing the non-trivial generator of $K_0(C(\TT^2))\cong\ZZ\oplus\ZZ$. Such a projection arises, basically by definition, as the Fermi projection of a ``Chern insulator'' Hamiltonian $H_{\rm Chern}$, see e.g.\ \cite{T-edge}.

Let $\delta: K_0(C(\TT^2))\rightarrow K_0(C(\TT))$ be the connecting ``exponential'' map for the Toeplitz sequence
\begin{equation*}
0\rightarrow C^*_r(\ZZ)\otimes\mathcal{K}\rightarrow C^*_r(\ZZ\times\NN)\rightarrow C^*_r(\ZZ^2)\rightarrow 0,
\end{equation*}
or equivalently,
\begin{equation*}
0\rightarrow C(\TT)\otimes\mathcal{K}\rightarrow C(\TT)\otimes C^*_r(\NN)\rightarrow C(\TT^2)\rightarrow 0.
\end{equation*}
As in \S\ref{sec:half.Hamiltonian.family}, $\delta[P_{\rm Chern}]$ may be represented by the self-adjoint Fredholm family $\widetilde{H}^{\rm 1D}_{\rm Chern}:\TT\rightarrow\mathcal{F}^{\rm sa}_*$ built from the half-space Hamiltonian $\widetilde{H}_{\rm Chern}$. We may compute (see, e.g., \cite{T-edge, Thiang-sf}) that $\delta[P_{\rm Chern}]=[\widetilde{H}^{\rm 1D}_{\rm Chern}]$ is a generator of $K_1(C(\TT))\cong\ZZ$, with the latter isomorphism given by the spectral flow of $\widetilde{H}^{\rm 1D}_{\rm Chern}$. Furthermore, $[\overline{P}_{\rm Chern}]=-[P_{\rm Chern}]$ in $\widetilde{K}^0(\TT)$ (the conjugate Bloch bundle has the opposite Chern class), so $\overline{P}_{\rm Chern}$ leads to an opposite spectral flow compared to $P_{\rm Chern}$.

\medskip
\noindent
{\bf Quantum spin Hall insulator.}
Now consider the direct sum $P_{\rm QSH}:=P_{\rm Chern}\oplus\overline{P}_{\rm Chern}$, which is invariant under ${\rm Ad}_\Theta$ where $\Theta$ is the standard quaternionic structure $\Theta=\begin{pmatrix} 0 & -\mathbbm{1}_2 \\ \mathbbm{1}_2 & 0\end{pmatrix}\circ\kappa$. The subscript `QSH' means that $P_{\rm QSH}$ is supposed to be the Fermi projection for a `Quantum Spin Hall' Hamiltonian $H_{\rm QSH}=H_{\rm Chern}\oplus \overline{H}_{\rm Chern}$, following ideas in \cite{KM}. 

Thus $P_{\rm QSH}$ defines some class in $KO_0((C^*_{r,\RR}(\ZZ^2)^\HH)\cong KR^{-4}(\TT^2)$. We also have $\delta_4[P_{\rm QSH}]=[\widetilde{H}^{\rm 1D}_{\rm QSH}]$, and for $\widetilde{H}^{\rm 1D}_{\rm QSH}$, there are equal and opposite spectral flows coming from each factor. Diagramatically, with an appropriate deformation of $\widetilde{H}^{\rm 1D}_{\rm Chern}$ and therefore of $\widetilde{H}^{\rm 1D}_{\rm QSH}$, the discrete spectrum of $\widetilde{H}^{\rm 1D}_{\rm QSH}(x)$ flows in a manner depicted in Fig.\ \ref{fig:phase.function2}.

For brevity, write $F_{\rm QSH}=\widetilde{H}^{\rm 1D}_{\rm QSH}$ for the equivariant self-adjoint Fredholm family. Its (equivariant) spectral flow vanishes, so we can construct the `Real' Fermi gerbe $\G_{F_{\rm QSH}}$. We can directly see that the sign invariants at the two fixed points of $\TT$ differ (see Fig.\ \ref{fig:phase.function3}), so $\G_{F_{\rm QSH}}$ has non-trivial `Real' DD-invariant,
\begin{equation}
DD(\G_{F_{\rm QSH}})=-1\in\ZZ_2\cong H^3_\pm(\TT)/H^3_\pm(\pt)\cong {\rm Map}(\TT^\iota,\ZZ_2)/\{\pm 1\}.\label{eqn:QSH.nontrivial.gerbe}
\end{equation}
This shows that the $K$-theory class of $F_{\rm QSH}$ and reduced $K$-theory class of $P_{\rm QSH}$ are non-trivial.

The FKMM invariant of $P_{\rm QSH}$ can also be directly shown to be nontrivial (e.g.\ Cor.\ 4.12 of \cite{D-G-AII}); that is, 
\begin{equation}
\kappa[P_{\rm QSH}]=-1\in \ZZ_2\cong H^2_\pm(\TT^2, (\TT^2)^\iota)\cong {\rm Map}((\TT^2)^\iota,\ZZ_2)/[\TT^2,{\rm U}(1)]_{\ZZ_2}.\label{eqn:QSH.nontrivial.FKMM}
\end{equation}

\begin{figure}[h!]
\begin{center}
\begin{tikzpicture}[scale=0.7]
\draw[->] (0,0) -- (4.5,0) node[anchor=south] {$\theta$};
\draw[->] (0,0) -- (0,4.7) node[anchor=south] {Spectrum};
\draw[dashed] (0,2.8) -- (4,2.8);
\draw[dashed] (0,3.2) -- (4,3.2);
\draw[dotted] (1.2,2.8) -- (1.2,-2);
\draw[dotted] (2.8,2.8) -- (2.8,-2);
\draw[dotted] (0.8,3.2) -- (0.8,-3);
\draw[dotted] (3.2,3.2) -- (3.2,-3);
\draw[<->] (1.2,-2) -- (2,-2) node[anchor=south] {$-$} -- (2.8,-2);
\draw[->] (0,-2) -- (0.6,-2) node[anchor=south] {$+\;\;\;$} -- (1.2,-2);
\draw[<-] (2.8,-2) -- (3.4,-2) node[anchor=south] {$\;\;\;+$} -- (4,-2);
\draw[<->] (0.8,-3) -- (2,-3) node[anchor=south] {$-$} -- (3.2,-3);
\draw[->] (0,-3) -- (0.4,-3) node[anchor=south] {$+$} -- (0.8,-3);
\draw[<-] (3.2,-3) -- (3.6,-3) node[anchor=south] {$+$} -- (4,-3);
\node at (-0.5,-2) {$U_{0.4}$};
\node at (-0.5,-3) {$U_{0.6}$};
\filldraw[gray] (0,4) -- (4,4)-- (4,4.5) -- (0,4.5);
\filldraw[gray] (0,0) -- (4,0)-- (4,-0.5) -- (0,-0.5);
\draw[thick] (0,4) -- (4,0);
\draw[thick] (0,0) --  (4,4);
\node[below] at (2,-0.5) {$0$};
\node[left] at (0,2) {$0$};
\node[left] at (0,0) {$-1$};
\node[below] at (0,-0.5) {$-\pi$};
\node[below] at (4,-0.5) {$\pi$};
\node[left] at (0,4) {$1$};
\end{tikzpicture}~\hspace{2em}
\begin{tikzpicture}[scale=0.7]
\draw[fill=gray!50] (0,0) -- (0,-0.3) -- (4,-0.3) -- (6,3.16) -- (6,3.36) -- (2,3.36) -- (0,0) -- (4,0) -- (6,3.36);
\draw (4,-0.3) -- (4,0);
\draw[gray,fill=gray!35,shift={(3cm,1.73cm)},rotate=30] (0,0) ellipse (30pt and 25pt);
\node at (3,5.732) {$\bullet$};
\draw[shift={(3cm,5.73cm)},rotate=30] (0,0) ellipse (30pt and 25pt);
\draw (2.1,5.232) -- (3.9,2.232);
\draw (4,5.732) -- (2,1.732);
\draw (3,4.77) -- (3,0.77);
\draw[fill=gray!50, opacity=0.95] (0,4.3) -- (0,4) -- (4,4) -- (6,7.46) -- (6,7.66) -- (2,7.66) -- (0,4.3) -- (4,4.3) -- (6,7.66);
\draw (4,4) -- (4,4.3);
\node at (3,5.732) {$\bullet$};
\node at (6,7.36) {$\bullet$};
\node[above] at (3,5.7) {$(0,0)$};
\node[above] at (6.5,6.6) {$(\pi,\pi)$};
\draw[dotted] (0,4) -- (2,7.36) -- (6,7.36);
\draw[->] (5,0) -- (6,0);
\draw[->] (5,0) -- (5.5,0.86);
\node at (6.3,0) {$\theta_1$};
\node at (5.7,1.1) {$\theta_2$};
\end{tikzpicture}
\end{center}
\caption{[Left] For the equivariant self-adjoint Fredholm family $\widetilde{H}^{\rm 1D}_{\rm QSH}$, a possible dependence of the spectrum on $e^{i\theta}\in\TT$ is plotted. The essential spectrum is shown as gray bands, with common essential spectral gap $(-1,1)$. The net flow of discrete spectrum across this gap vanishes, but the `Real' Fermi gerbe is nontrivial. Two of the open sets in the open cover $\{U_\mu\}_{\mu\in(-1,1)}$ are shown, for $\mu=0.4$ and $\mu=0.6$, and together they suffice to refine $\{U_\mu\}_{\mu\in(-1,1)}$. A possible choice of $\tau_{0.4}$ and $\tau_{0.6}$ is indicated by the signs attached to each connected component of $U_{0.4}$ and $U_{0.6}$. At the fixed point $+1$ (i.e.\ $\theta=0$), the sign invariant of the Fermi gerbe is the value of $\tau_{0.4}$ (or $\tau_{0.6}$) there, which is $-1$. At the fixed point $-1$ (i.e.\ $\theta=\pi$), the sign invariant is $+1$. These sign invariants are ambiguous up to an overall sign. The (unambiguous) product of these sign invariants is $-1$, so the Fermi gerbe is non-trivial in $H^3_\pm(\TT)/H^3_\pm(\pt)$. [Right] A rotated version of the first diagram, depicting schematically the spectrum of an equivariant family $\widetilde{H}^{\rm 1D}_{\rm strong}:\TT^2\rightarrow\mathcal{F}^{\rm sa}_*$ associated to a strong 3D topological insulator. The sign invariants of the Fermi gerbe are similarly seen to be $-1$ at the fixed point $(\theta_1,\theta_2)=(0,0)$ and $+1$ at the remaining three fixed points $(0,\pi), (\pi,0), (\pi,\pi)$ of $\TT^2$.}
\label{fig:phase.function3}
\end{figure}
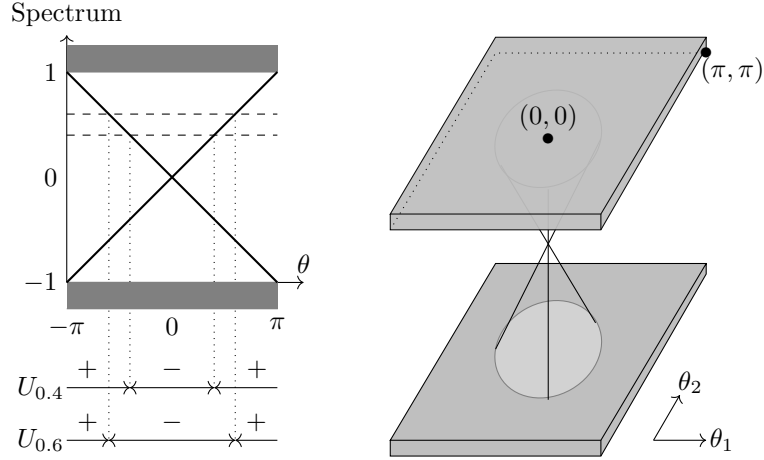

\medskip

Let $\pi=\pi_2:\TT^2\rightarrow\TT^1$ be the map projecting out the second coordinate, which also expresses $\TT^2$ as a `Real' circle bundle over $\TT$. Restricted to the fixed points, $\pi$ gives $(\TT^2)^\iota$ as a double cover of $\TT^\iota$. Then sign maps on $(\TT^2)^\iota$ can be ``pushed forward'' to sign maps on $\TT^\iota$ by multiplying along (the two-point) fibers (this is a genuine pushforward if we think of sign maps as elements of $H^0(X^\iota;\ZZ_2)$). 

Furthermore, by invoking Lemma \ref{lem:tori.sign.maps}, we see that any equivariant map in 
$[\TT^2,{\rm U}(1)]_{\ZZ_2}$ restricts on $(\TT^2)^\iota$ to a sign map whose push-forward is a constant sign map $\pm 1$ on $\TT^\iota$. This means that the push-forward descends to a well-defined map
\begin{equation*}
\pi_*:\underbrace{{\rm Map}((\TT^2)^\iota,\ZZ_2)/[\TT^2,{\rm U}(1)]_{\ZZ_2}}_{H^2_\pm(\TT^2, (\TT^2)^\iota)}\rightarrow \underbrace{{\rm Map}(\TT^\iota,\ZZ_2)/\{\pm 1\}}_{H^3_\pm(\TT)/H^3_\pm(\pt)}.
\end{equation*}

\begin{prop}\label{prop:commute.2torus}
The following diagram commutes:
\begin{equation*}
\begin{CD}
KR^{-4}(\TT^2) @>{\kappa}>> 
H^2_\pm(\TT^2, (\TT^2)^\iota) \\
@V{\pi_!=\delta_4}VV @VV{\pi_*}V \\
KR^{-3}(\TT) @>{DD}>> H^3_\pm(\TT)/H^3_\pm(\pt).
\end{CD}
\end{equation*}
\end{prop}
\begin{proof}
It suffices to check on the non-trivial generator $[P_{\rm QSH}]$ of $KR^{-4}(\TT^2)$ (since $\kappa$ and $\pi_!$ vanish on the trivial generator $[1]$). From Eq.\ \eqref{eqn:QSH.nontrivial.gerbe}, \mbox{$DD\circ\pi_![P_{\rm QSH}]$} is the non-trivial element of $H^3_\pm(\TT)/H^3_\pm(\pt)$, so it is represented by a non-constant sign map on $\TT^\iota$. By Eq.\ \eqref{eqn:QSH.nontrivial.FKMM}, $\kappa[P_{\rm QSH}]$ is non-trivial as a class in ${\rm Map}((\TT^2)^\iota,\ZZ_2)/[\TT^2,{\rm U}(1)]_{\ZZ_2}$, thus it is represented by a sign map on $(\TT^2)^\iota$ with product of signs $-1$ (Eq.\ \eqref{eqn:POS}). Then $\pi_*\circ\kappa[P_{\rm QSH}]$ is also represented by a sign map on $\TT^\iota$ with product of signs $-1$ (Eq.\ \eqref{eqn:POS.from.KR}), i.e.\ nonconstant.
\end{proof}

\medskip

Now let $\pi=\pi_3:\TT^3\rightarrow\TT^2$ project out the third coordinate. By similar arguments, the push-forward is well defined on equivalence classes of sign maps on the fixed points,
\begin{equation*}
\pi_*:\underbrace{{\rm Map}((\TT^3)^\iota,\ZZ_2)/[\TT^3,{\rm U}(1)]_{\ZZ_2}}_{H^2_\pm(\TT^3, (\TT^3)^\iota)}\rightarrow \underbrace{{\rm Map}((\TT^2)^\iota,\ZZ_2)/\{\pm 1\}}_{H^3_\pm(\TT^2)/H^3_\pm(\pt)}.
\end{equation*}

\begin{prop}\label{prop:commute.3torus}
The following diagram commutes:
\begin{equation*}
\begin{CD}
KR^{-4}(\TT^3) @>{\kappa}>> 
H^2_\pm(\TT^3, (\TT^3)^\iota) \\
@V{\pi_!=\delta_4}VV @VV{\pi_*}V \\
KR^{-3}(\TT^2) @>{DD}>> H^3_\pm(\TT^2)/H^3_\pm(\pt).
\end{CD}
\end{equation*}
\end{prop}
\begin{proof}
For the ``weak'' generators $\nu_1, \nu_2, \nu_3$ of $KR^{-4}(\TT^3)$, which are by definition pulled back from $\TT^2$, this follows from Prop.\ \ref{prop:commute.2torus} by naturality. For the ``strong'' generator $\nu_0$ of $KR^{-4}(\TT^3)$, we consider the section maps $s_{i,\pm}:\TT^2\rightarrow \TT^3, i=1,2$, which picks out the 2-torus inside $\TT^3$ at constant $i$-th coordinate $\pm 1$. The FKMM invariant of the strong generator $\nu_0$ is represented by the sign map $\mathfrak{d}_{\rm strong}$ on $(\TT^3)^\iota$ which is $-1$ at $(+1,+1,+1)$ and $+1$ elsewhere. This restricts on the 2-torus sections $s_{i,+}$ and $s_{i,-}$ to sign maps having different product-of-signs (this holds for $i=1$ and $i=2$). Each 2-torus section $s_{i,\pm}$ itself fibers (under $\pi=\pi_3$) over the circle in $\TT^2$ with $i$-th coordinate $\pm 1$. Invoking naturality again, we see that $DD\circ\pi_! (\nu_0)$ is represented by a sign map on $(\TT^2)^\iota$ whose restriction to the circle with $i$-th coordinate $+1$ (resp.\ $-1$) is non-trivial (resp.\ trivial) modulo the overall global sign; this holds for $i=1$ and $i=2$. Thus $DD\circ\pi_! (\nu_0)$ is representable by the sign map $\tilde{\mathfrak{d}}_{\rm strong}$ which is $-1$ at $(+1,+1)$ and $+1$ elsewhere in $(\TT^2)^\iota$. Thus we have verified that $\pi_*[\mathfrak{d}_{\rm strong}]=[\tilde{\mathfrak{d}}_{\rm strong}]$, which is the commutativity result for the strong generator.
\end{proof}

\subsection{How to count Dirac cone edge states?}\label{sec:counting.Dirac}
Suppose $H_{\rm TI}$ is the Hamiltonian for a 2D $\Theta$-invariant topological insulator. Thus its Fermi projection represents the generator of $\widetilde{KR}^{-4}(\TT^2)\cong\ZZ_2$. 

Figure \ref{fig:phase.function2} shows some possible spectral dispersions for the half-space Hamiltonian $\widetilde{H}_{\rm TI}$ (think of $2\lambda_i$ as the discrete spectrum inside the bulk gap $(-1,1)$). 

To define an ``edge index'', we could follow \cite{KM,GP} and take the mod-2 number of points $k\in (0,\pi)\subset\TT$ at which a (simple) eigenvalue of $\widetilde{H}^{\rm 1D}(k)$ crosses $0$. Physically, this is supposed to count the number of Kramers pairs (eigenstates related by $\Theta$) at the Fermi energy. Where applicable, this index coincides with our `Real' Fermi gerbe invariant, as is easily verified in the examples in Fig.\ \ref{fig:phase.function}. This method of counting, however, depends on the parameter space $\TT^{d-1}$ being one-dimensional. 

\medskip

For topological insulators in dimension $d=3$, for which the boundary Brillouin zone $\TT^{d-1}$ is two-dimensional, the physicists' rule-of-thumb is that the strong Fu--Kane--Mele $\ZZ_2$ invariant for $H$ is non-trivial iff $\widetilde{H}$ acquires an odd number of Dirac points (eigenvalue crossings) in the bulk gap, located at the fixed points $(\TT^2)^\iota$ (see e.g.\ \cite{FKM}). Alternatively, one says that a strong topological insulator should have an odd number of ``Dirac cone edge states'' (Fig.\ \ref{fig:phase.function3}).

To see that there might be a problem with such an informal statement, consider the analogous statement for $d=2$. We see from Fig.\ \ref{fig:phase.function2} that for a 2D topological insulator, neither the positions nor the total number (mod-2) of Dirac points is a homotopy invariant. Furthermore, as illustrated in Fig. \ref{fig:phase.function}, a Dirac point may merge into the bulk spectrum when $\Theta$-invariant perturbations are continuously turned on. 

By considering ``surfaces of rotation'', the same issues with Dirac point counting arises for $d=3$. An illustration of how a strong topological insulator might have no Dirac points is illustrated in Fig.\ \ref{fig:phase.function}, where the Dirac point has been manipulated in such a way that it disappears into the bulk spectrum. Conversely, it is possible to have a single Dirac cone but no gap-filling. We mention that Dirac cone manipulations are of experimental importance, see e.g.\ \cite{LYTH} on uncovering ``hiddden'' Dirac cones, although the goal is usually to move the Dirac point closer to the Fermi energy.

The upshot of the above discussion is that a more precise signature of topological insulators is the topologically protected gap-filling property, whether by a ``classic'' Dirac cone, or otherwise. Our `Real' Fermi gerbe invariant captures such a signature precisely and invariantly.

\begin{figure}[h!]
\begin{center}
\begin{tikzpicture}[scale=0.6, every node/.style={scale=0.8}]
\draw[->] (0,0) -- (4.5,0) node[anchor=south] {$\theta$};
\draw[->] (0,0) -- (0,4.7) node[anchor=south] {Spectrum};
\filldraw[gray=30] (0,4) -- (4,4)-- (4,4.5) -- (0,4.5);
\filldraw[gray=30] (0,0) -- (4,0)-- (4,-0.5) -- (0,-0.5);
\draw[thick] (0,4) -- (4,0);
\draw[thick] (0,0) --  (4,4);
\node[below] at (2,-0.5) {$0$};
\node[left] at (0,2) {$0$};
\node[left] at (0,0) {$-1$};
\node[below] at (0,-0.5) {$-\pi$};
\node[below] at (4,-0.5) {$\pi$};
\node[left] at (0,4) {$1$};
\end{tikzpicture}~
\begin{tikzpicture}[scale=0.6, every node/.style={scale=0.8}]
\draw[->] (0,0) -- (4.5,0) node[anchor=south] {$\theta$};
\filldraw[gray=30] (0,4) -- (4,4)-- (4,4.5) -- (0,4.5);
\filldraw[gray=30] (0,0) -- (4,0)-- (4,-0.5) -- (0,-0.5);
\draw[->] (0,0) -- (0,4.7);
\draw[thick] (0,4) -- (2,1) -- (4,0);
\draw[thick] (0,0) -- (2,1) -- (4,4);
\node[below] at (2,-0.5) {$0$};
\node[left] at (0,2) {$0$};
\node[left] at (0,0) {$-1$};
\node[below] at (0,-0.5) {$-\pi$};
\node[below] at (4,-0.5) {$\pi$};
\node[left] at (0,4) {$1$};
\end{tikzpicture}~
\begin{tikzpicture}[scale=0.6, every node/.style={scale=0.8}]
\draw[->] (0,0) -- (4.5,0) node[anchor=south] {$\theta$};
\filldraw[gray=30] (0,4) -- (4,4)-- (4,4.5) -- (0,4.5);
\filldraw[gray=30] (0,0) -- (4,0)-- (4,-0.5) -- (0,-0.5);
\draw[->] (0,0) -- (0,4.7);
\draw[thick] (0,4) -- (2,0) -- (4,0);
\draw[thick] (0,0) -- (2,0) -- (4,4);
\node[below] at (2,-0.5) {$0$};
\node[left] at (0,2) {$0$};
\node[left] at (0,0) {$-1$};
\node[below] at (0,-0.5) {$-\pi$};
\node[below] at (4,-0.5) {$\pi$};
\node[left] at (0,4) {$1$};
\end{tikzpicture}
\end{center}
\begin{center}
\begin{tikzpicture}[scale=0.6, every node/.style={scale=0.8}]
\draw[->] (0,0) -- (4.5,0) node[anchor=south] {$\theta$};
\filldraw[gray=30] (0,4) -- (4,4)-- (4,4.5) -- (0,4.5);
\filldraw[gray=30] (0,0) -- (4,0)-- (4,-0.5) -- (0,-0.5);
\draw[->] (0,0) -- (0,4.7);
\draw[thick] (0.2,4) -- (1.8,0) -- (4,0);
\draw[thick] (0,0) -- (2.2,0) -- (3.8,4);
\node[below] at (2,-0.5) {$0$};
\node[left] at (0,2) {$0$};
\node[left] at (0,0) {$-1$};
\node[below] at (0,-0.5) {$-\pi$};
\node[below] at (4,-0.5) {$\pi$};
\node[left] at (0,4) {$1$};
\draw[dashed] (0,2) -- (4,2);
\draw[dashed] (0,3) -- (4,3);
\draw[dotted] (0.6,3) -- (0.6,-3);
\draw[dotted] (1,2) -- (1,-2);
\draw[dotted] (3.4,3) -- (3.4,-3);
\draw[dotted] (3,2) -- (3,-2);
\draw[<->] (1,-2) -- (2,-2) node[anchor=south] {$-$} -- (3,-2);
\draw[->] (0,-2) -- (0.5,-2) node[anchor=south] {$+\;\;\;$} -- (1,-2);
\draw[<-] (3,-2) -- (3.5,-2) node[anchor=south] {$\;\;\;+$} -- (4,-2);
\draw[<->] (0.6,-3) -- (2,-3) node[anchor=south] {$-$} -- (3.4,-3);
\draw[->] (0,-3) -- (0.3,-3) node[anchor=south] {$+$} -- (0.6,-3);
\draw[<-] (3.4,-3) -- (3.7,-3) node[anchor=south] {$+$} -- (4,-3);
\node at (-0.5,-2) {$U_{0}$};
\node at (-0.5,-3) {$U_{0.5}$};
\end{tikzpicture}
\begin{tikzpicture}[scale=0.6, every node/.style={scale=0.8}]
\draw[fill=gray!50] (0,0) -- (0,-0.3) -- (4,-0.3) -- (6,3.16) -- (6,3.36) -- (2,3.36) -- (0,0) -- (4,0) -- (6,3.36);
\draw (4,-0.3) -- (4,0);
\draw[gray,fill=gray!35,shift={(3cm,1.73cm)},rotate=30] (0,0) ellipse (20pt and 15pt);
\node at (3,5.732) {$\bullet$};
\draw[shift={(3cm,5.73cm)},rotate=30] (0,0) ellipse (35pt and 20pt);
\draw (4,5.7) -- (3.65,1.8);
\draw (1.95,5.732) -- (2.35,1.6);
\draw (3,5) -- (3,1.15);
\draw[fill=gray!50, opacity=0.95] (0,4.3) -- (0,4) -- (4,4) -- (6,7.46) -- (6,7.66) -- (2,7.66) -- (0,4.3) -- (4,4.3) -- (6,7.66);
\draw (4,4) -- (4,4.3);
\node at (3,5.732) {$\bullet$};
\node at (6,7.36) {$\bullet$};
\node[above] at (3,5.7) {$(0,0)$};
\node[above] at (6,6.3) {$(\pi,\pi)$};
\draw[dotted] (0,4) -- (2,7.36) -- (6,7.36);
\draw[->] (5,0) -- (6,0);
\draw[->] (5,0) -- (5.5,0.86);
\node[below] at (6,0) {$\theta_1$};
\node at (5.7,1.1) {$\theta_2$};
\end{tikzpicture}
\begin{tikzpicture}[scale=0.65, every node/.style={scale=0.8}]
\draw[thick] (0.3,4) .. controls (1.6,0) .. (2.9,4);
\draw[thick] (1.1,4) .. controls (2.4,0) .. (3.7,4);
\draw[->] (0,0) -- (4.5,0) node[anchor=south] {$\theta$};
\filldraw[gray=30] (0,4) -- (4,4)-- (4,4.5) -- (0,4.5);
\filldraw[gray=30] (0,0) -- (4,0)-- (4,-0.5) -- (0,-0.5);
\draw[->] (0,0) -- (0,4.7);
\node[below] at (2,-0.5) {$0$};
\node[left] at (0,2) {$0$};
\node[left] at (0,0.6) {$-0.7$};
\node[left] at (0,0) {$-1$};
\node[below] at (0,-0.5) {$-\pi$};
\node[below] at (4,-0.5) {$\pi$};
\node[left] at (0,4) {$1$};
\draw[dashed] (0,2) -- (4,2);
\draw[dashed] (0,3) -- (4,3);
\draw[dashed] (0,0.6) -- (4,0.6);
\draw[dotted] (0.6,3) -- (0.6,-3);
\draw[dotted] (1,2) -- (1,-2);
\draw[dotted] (3.4,3) -- (3.4,-3);
\draw[dotted] (3,2) -- (3,-2);
\draw[dotted] (1.8,2) -- (1.8,-2);
\draw[dotted] (2.2,2) -- (2.2,-2);
\draw[dotted] (1.4,3) -- (1.4,-3);
\draw[dotted] (2.6,3) -- (2.6,-3);
\draw[->] (0,-2) -- (0.5,-2) node[anchor=south] {$-\;\;\;$} -- (1,-2);
\draw[<->] (1,-2) -- (1.4,-2) node[anchor=south] {$+$} -- (1.8,-2);
\draw[<->] (1.8,-2) -- (2,-2) node[anchor=south] {$-$} -- (2.2,-2);
\draw[<->] (2.2,-2) -- (2.6,-2) node[anchor=south] {$+$} -- (3,-2);
\draw[<-] (3,-2) -- (3.5,-2) node[anchor=south] {$\;\;\;-$} -- (4,-2);
\draw[->] (0,-3) -- (0.3,-3) node[anchor=south] {$-$} -- (0.6,-3);
\draw[<->] (0.6,-3) -- (1,-3) node[anchor=south] {$+$} -- (1.4,-3);
\draw[<->] (1.4,-3) -- (2,-3) node[anchor=south] {$-$} -- (2.6,-3);
\draw[<->] (2.6,-3) -- (3,-3) node[anchor=south] {$+$} -- (3.4,-3);
\draw[<-] (3.4,-3) -- (3.7,-3) node[anchor=south] {$-$} -- (4,-3);
\node at (-0.5,-2) {$U_{0}$};
\node at (-0.5,-3) {$U_{0.5}$};
\node at (-0.7,-4) {$U_{-0.7}$};
\draw (0,-4) -- (2,-4)node[anchor=south] {$-$} -- (4,-4);
\end{tikzpicture}
\end{center}
\caption{[Top row] Spectral dispersion plots for some equivariantly homotopic maps $\TT\rightarrow\mathcal{F}^{\rm sa}_*$. [Bottom left] Continuing the homotopy, the Dirac point (eigenvalue crossing) disappears into the lower essential spectrum, and the na\"{i}ve Dirac cone count drops from $1$ to $0$. The correct counting via the `Real' Fermi gerbe shows that its gerbe invariant (product of signs at $\theta=0$ and $\theta=\pi$) remains nontrivial ($-$) throughout, as befits a topological invariant, indicating that the essential spectral gap remains filled up by the family of discrete spectra. [Bottom centre] By considering surfaces of rotation, we can similarly homotope the classic Dirac cone of Fig.\ \ref{fig:phase.function3} into a frustum (truncated cone) with no Dirac point. [Bottom right] Conversely, there may be a Dirac point but no gap-filling: the gerbe invariant is manifestly trivial ($+$) in this example, by considering the single open set $U_{-0.7}=\TT$, for instance.}\label{fig:phase.function}
\end{figure}
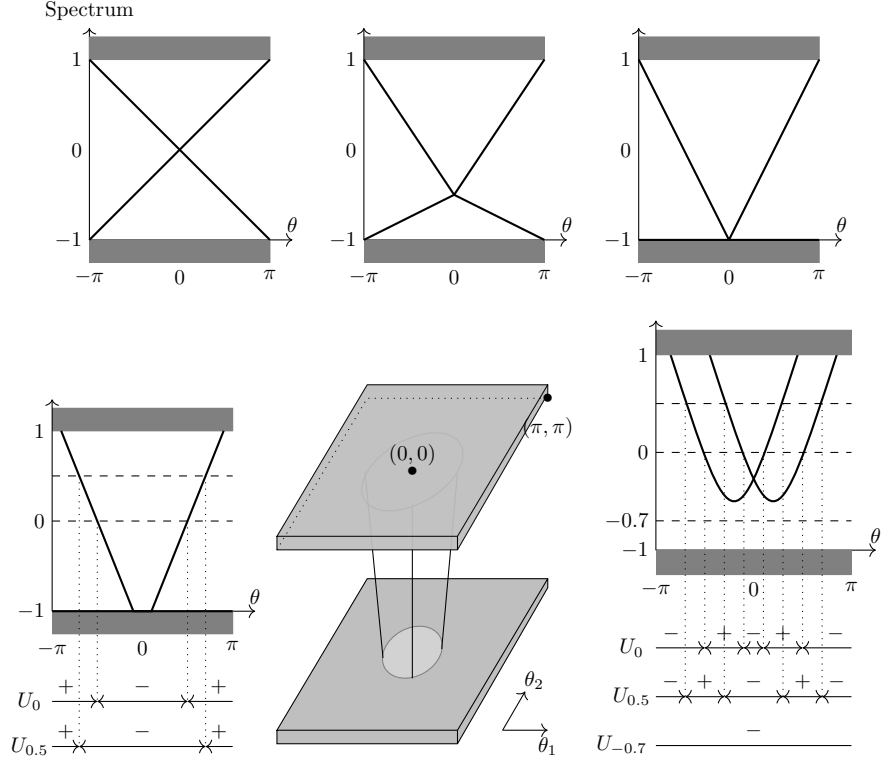

\medskip

Suppose $H_{\rm strong}$ is the Hamiltonian for a strong topological insulator in 3D, that is, its Fermi projection represents the strong generator in $\widetilde{KR}^{-4}(\TT^3)$. Then we saw from Prop.\ \ref{prop:commute.3torus} that for the half space Hamiltonian $\widetilde{H}_{\rm strong}$, its `Real' Fermi gerbe over $\TT^2$ has sign invariants given (up to a global sign) by the map $\tilde{\mathfrak{d}}_{\rm strong}$ which is $-1$ at $(+1,+1)$ and $+1$ elsewhere. Schematically, if we rotate the first diagram of Fig.\ \ref{fig:phase.function3} about the vertical axis through $\theta=0$, we obtain the classic Dirac cone edge spectrum, and we see that the Fermi gerbe has sign invariants given by $\tilde{\mathfrak{d}}_{\rm strong}$. The true spectrum of $\widetilde{H}_{\rm strong}$ will not have such a nice dispersion: the Dirac cone may be deformed, there could be three Dirac cones over the three fixed points other than $(+1,+1)$, or there could be some merging of the Dirac points with the bulk spectrum. The allowed dispersions of $\widetilde{H}_{\rm strong}$ are constrained to have equivalent `Real' Fermi gerbes, and in particular, the sign invariants (up to a global sign) must be preserved.

More generally, $H$ could have some non-trivial combination of weak invariants $\nu_1,\nu_2$ and the strong invariant $\nu_0$; or equivalently, sign maps $\mathfrak{d}_0,\mathfrak{d}_1,\mathfrak{d}_2$. 
Whatever the case may be, the `Real' Fermi gerbe for $\widetilde{H}$ will have non-trivial sign invariants indicating how the bulk spectral gap must be fully bridged by the edge states.

%%%%%%%%%%%%%%%%%%%%%%%%%%%%%%%%%%%%%%%%%%%%%%%%
%%%%%%%%%%%%%%%%%%%%%%%%%%%%%%%%%%%%%%%%%%%%%%%%

\appendix

\section{Equivariant cohomology}\label{appendix:Borel.Cech}

\subsection{\v{C}ech formulation}\label{appendix:Borel.Cech.untwisted}
Let $G$ be a compact Lie group and $X$ a space with a left continuous action of $G$. 
By definition, the Borel equivariant cohomology $H^n_G(X)$ is defined to be the singular integral cohomology of the Borel construction 
$$
H^n_G(X) = H^n_G(X; \Z) = H^n(EG \times_G X; \Z).
$$
Let us assume that
\begin{quote}
$X$ is a $G$-CW complex.
\end{quote}
This assumption will be good enough for a \v{C}ech cohomology description of the Borel equivariant cohomology, as outlined below:

\begin{enumerate}

\item
In general, if $X_\bullet$ is a simplicial space and $A$ is a finitely generated abelian group, then the singular cohomology $H^*(\lVert X_\bullet \rVert; A)$ of the geometric realisation $\lVert X_\bullet \rVert$ is isomorphic to the singular cohomology $H^*(X_\bullet; A)$ of the simplicial space \cite{Dup}. The Borel construction $EG \times_G X$ is homotopy equivalent to the geometric realisation $\lVert G^\bullet \times X \rVert$ of the simplicial space $G^\bullet \times X$. In more detail, the latter comprises the sequence of spaces $\{ G^p \times X \}_{p \ge 0}$ together with the face maps
$\partial_i : G^{p+1} \times X \to G^{p} \times X$, ($i = 0, \ldots, p+1$) given by
\begin{equation*}
\partial_i(g_1, \ldots, g_{p+1}, x)
=
\left\{
\begin{array}{ll}
(g_2, \ldots, g_{p+1}, x), & (i = 0) \\
(g_1, \ldots, g_{i-1}, g_ig_{i+1}, g_{i+1}, \ldots, g_{p+1}, x), & 
(i = 1, \ldots, p) \\
(g_1, \ldots, g_p, g_{p+1}x), & (i = p+1)
\end{array}
\right.
\end{equation*}
and the degeneracy maps $s_i : G^p \times X \to G^{p+1} \times X$, ($i = 0, \ldots, p$) given by
\begin{equation*}
s_i(g_1, \ldots, g_p, x) = 
(g_1, \ldots, g_i, 1, g_{i+1}, \ldots, g_p, x).
\end{equation*}
These facts lead to isomorphisms:
\begin{equation*}
H^n_G(X) \cong H^n(\lVert G^\bullet \times X \rVert; \Z)
\cong H^n(G^\bullet \times X; \Z),
\end{equation*}
where $H^n(G^\bullet \times X; \Z)$ is the singular integral cohomology of the simplicial space $G^\bullet \times X$. This is the cohomology associated to the double complex $(C^q(G^p \times X; \Z), \delta, \partial)$, where for each $p$, $(C^q(G^p \times X; \Z), \delta)$ is the singular cochain complex of the topological space $G^p \times X$, while \mbox{$\partial : C^q(G^p \times X; \Z) \to C^q(G^{p+1} \times X; \Z)$} is given by $\partial = \sum_{i=0}^{p+1} (-1)^i \partial_i^*$, with $\partial_i : G^{p+1} \times X \to G^p \times X$ the face maps in the simplicial space $G^\bullet \times X$.

\item\label{item:cover}
There is an isomorphism
$$
H^n(G^\bullet \times X; \Z)
\cong \check{H}^n(G^\bullet \times X; \Z).
$$
The right hand side is the \v{C}ech cohomology of the simplicial space $G^\bullet \times X$ with coefficients in the constant sheaf $\Z$, defined as the colimit
$$
\check{H}^n(G^\bullet \times X; \Z)
= \varinjlim_{\U^\bullet} \check{H}^n(\U^\bullet; \Z),
$$
where $\check{H}^n(\U^\bullet; \Z)$ is the \v{C}ech cohomology associated to an open cover $\U^\bullet$ of the simplicial space $G^\bullet \times X$. Such an open cover consists of a sequence of open covers $\U^0, \U^1, \ldots$, where $\U^p = \{ U^p_i \}_{i \in I^p }$ is an open cover of $G^p \times X$, the sequence of sets $\{ I^p \}$ forms a simplicial set, and a compatibility with the face maps is supposed. 

For example, suppose $G$ is finite, and let $\{ U_i \}_{i \in I}$ be an invariant open cover of $X$. Then there is an associated open cover $\U^\bullet$ of $G^\bullet \times X$, in which $\U^p$ of $G^p \times X$ consists of open sets 
$$
U_{(g_1, \ldots, g_p, i)} = \{ (g_1, \ldots, g_p) \} \times U_i
$$
indexed by $(g_1, \ldots, g_p, i) \in G^p \times I$. (Details can be found in \cite{GDeligne}.) 

For each $p = 0, 1, \ldots$, we have the usual \v{C}ech complex $(\check{C}^q(\U^p; \Z), \delta)$. From the compatibility, we can define a homomorphism $\partial : \check{C}^q(\U^p; \Z) \to \check{C}^q(\U^{p+1}; \Z)$. This leads to a double complex $(\check{C}^q(\U^p; \Z), \delta, \partial)$, and its associated cohomology is $\check{H}^*(\U^\bullet; \Z)$.

Mimicking the idea in \cite{B-T}, we may construct a homomorphism
$$
H^n(G^\bullet \times X; \Z)
\to \check{H}^n(G^\bullet \times X; \Z).
$$
Each of the singular cohomology and the \v{C}ech cohomology is defined through a double complex. Hence the above homomorphism becomes an isomorphism when 
$$
H^n(G^p \times X; \Z) \to \check{H}^n(G^p \times X; \Z)
$$
is an isomorphism for all $n, p$, as is the case in our setup.

\item
Henceforth, suppose $G$ is a finite group. Then, we may show that any open cover $\U^\bullet$ of $G^\bullet \times X$ admits a refinement associated to an \emph{invariant} open cover:  Since $G^p \times X$ is paracompact, we may find a refinement $\V^\bullet$ such that the open cover $\V^p = \{ V^p_j \}_{j \in J^p}$ of $G^p \times X$ is locally finite. Being a CW complex, $X$ is (completely) regular. So, for each $x \in X$, we can find an open set $W_x$ such that $x \in W_x \subset V^0_j$ for all $j \in J^0$ such that $x \in V^0_j$. Considering the action of the finite group $G$, we may take $W_x$ to be $G$-invariant, and we would eventually get an invariant open cover $\W$ of $X$ whose associated open cover of $G^\bullet \times X$ refines $\V^\bullet$.

As a result, in the case that $G$ is finite, it is enough to consider \v{C}ech cohomology groups $\check{H}^*(\U^\bullet; \Z)$ of the open covers $\U^\bullet$ associated to invariant open covers $\U$ of $X$.

\item
Suppose further that $G = \Z_2$. Let $\U = \{ U_i \}_{i \in I}$ be an invariant open cover of $X$. We have another double complex $(\check{C}^{p, q}(\U), \delta, \partial)$ comprising \v{C}ech $q$-cochains, with $\delta$ the usual \v{C}ech coboundary, and $\partial:\check{C}^{p,q}(\U)\rightarrow\check{C}^{p+1,q}(\U)$ given by $\partial(\omega)=\omega-(-1)^p\iota^*\omega$. It can be shown that $(\check{C}^{p, q}(\U), \delta, \partial)$ and $(\check{C}^q(\U^p; \Z), \delta, \partial)$ are quasi-isomorphic: Taking a refinement if necessary, we can
assume that $\U^\bullet$ is an open cover of the simplicial space $\Z_2^\bullet \times
X$ of the form described in item \ref{item:cover} above. Then the complex $\partial : \check{C}^q(\U^p; \Z) \to
\check{C}^q(\U^{p+1}; \Z)$ is identified with the cochain complex of
the group $\Z_2$ with coefficients in a $\Z_2$-module. It is
well-known that the cochain complex of a group is quasi-isomorphic to
the \textit{normalized} cochain complex. In the present case, the
normalized cochain complex agrees with $\partial : \check{C}^{p,
q}(\U) \to \check{C}^{p+q, q}(\U)$. Thus, by a spectral sequence
argument, the total complexes for the two double complexes are
quasi-isomorphic.

\item
To summarize, if $X$ is a $\ZZ_2$-CW complex, then
\begin{equation}
H^n_{\Z_2}(X) \cong
\varinjlim_{\U} H^n(\check{C}^{*, *}(\U)),\label{eqn:Borel.Cech}
\end{equation}
where $\U$ runs over invariant open covers of $X$.

\end{enumerate}

If $X$ is a $G$-CW complex with $G$ a compact Lie group, then $G^p \times X$ is a CW complex, thus locally contractible. As a result, we get the exponential exact sequence of sheaves,
\begin{equation}
0 \longrightarrow \Z \longrightarrow \underline{\R} \overset{{\rm exp}\,2\pi i(\cdot)}{\longrightarrow} \underline{\mathbb{T}} \longrightarrow 0,\label{eqn:sheaf.sequence}
\end{equation}
where $\Z$ is the constant sheaf, and $\underline{A}$ means the sheaf of germs of $A$-valued continuous functions.
From this, we get the associated exact sequence
$$
\to
\check{H}^n(G^\bullet \times X; \Z) \to
\check{H}^n(G^\bullet \times X; \underline{\R}) \to
\check{H}^n(G^\bullet \times X; \underline{\mathbb{T}}) \to
\check{H}^{n+1}(G^\bullet \times X; \Z) \to.
$$
Because $G^p \times X$ is paracompact for $p \ge 0$, a spectral sequence argument identifies $\check{H}^n(G^\bullet \times X; \underline{\R})$ with the group cohomology with coefficients in a $G$-module. Then a standard ``averaging argument'' leads to $\check{H}^n(G^\bullet \times X; \underline{\R}) = 0$ for $n \neq 0$, and
\begin{equation}
\check{H}^n(G^\bullet \times X; \underline{\mathbb{T}}) \cong
\check{H}^{n+1}(G^\bullet \times X; \Z)
\cong H^{n+1}_G(X)\label{eqn:Borel.Cech.simplicial}
\end{equation}
for $n \ge 1$. The details of this argument can be found in \cite{GDeligne} (Lemma 4.4). This gives another way to a \v{C}ech formulation of the Borel equivariant cohomology.

\subsection{Local coefficients}\label{appendix:Borel.Cech.twisted}
On a connected CW complex $X$, a choice of a homomorphism $\ell \in \mathrm{Hom}(\pi_1(X), \Z_2) \cong H^1(X; \Z_2)$ makes the group $\Z$ into a module over the fundamental group $\pi_1(X)$, and allows us to ``twist'' the integral cohomology $H^*(X)$ to produce the cohomology $H^{\ell + *}(X)$ with local coefficients \cite{Hatcher}. As an application of this construction, a choice of a homomorphism $\phi : G \to \Z_2$ leads to a $\phi$-twisted version $H^{\phi + *}_G(X)$ of the  Borel equivariant integral cohomology $H^*_G(X)$ of a $G$-CW complex $X$, since $\phi \in \mathrm{Hom}(\pi_1(BG)), \Z_2) \cong H^1(BG; \Z_2)$ leads to an element in $H^1(EG \times_G X; \Z_2)$ by pull-back. In particular, when $G = \Z_2$ and $\phi = \mathrm{id}$ is the identity map, the $\phi$-twisted equivariant cohomology is denoted $H^*_\pm(X) = H^*_{\Z_2}(X; \Z(1)) = H^{\phi + *}_{\Z_2}(X)$.

The exact sequence of sheaves, Eq.\ \eqref{eqn:sheaf.sequence}, can be made $\ZZ_2$-equivariant, with the action on $\ZZ$ being $m\mapsto -m$ (whence the equivariant sheaf is denoted $\ZZ(1)$). To the equivariant sheaf $\Z(1)$ we can associate a sheaf $\tilde{\Z}$ on the simplicial space $\Z_2^\bullet \times X$ as a twisted version of the constant sheaf $\Z$, and the same argument leading to Eq. \eqref{eqn:Borel.Cech.simplicial} gives 
$$
H^n_\pm(X) \cong
\check{H}^n(\Z_2^\bullet \times X; \tilde{\Z}) \cong
\check{H}^n(\Z_2^\bullet \times X; \underline{\tilde{\mathbb{T}}}),
$$
where $\underline{\tilde{\mathbb{T}}}$ is likewise a twisted version of the sheaf $\underline{\mathbb{T}}$ on $\Z_2^\bullet \times X$. See \cite{Gomi} Appendix A, which is based on \cite{Kahn,Gro}, for details.

As in Eq.\ \eqref{eqn:Borel.Cech} for the untwisted case, we may identify $H^n_\pm(X)$ with the (colimit of the) cohomology of the double complex $(\check{C}^{p, q}(\U), \delta, \partial)$ comprising \v{C}ech $q$-cochains, but with \mbox{$\partial:\check{C}^{p,q}(\U)\rightarrow\check{C}^{p+1,q}(\U)$} given by $\partial(\omega)=\omega+(-1)^p\iota^*\omega$.

\subsection{First equivariant cohomology}

Let $X$ be a $\Z_2$-CW complex. In the \v{C}ech formulation of $H^n_{\Z_2}(X)$, Eq.\ \eqref{eqn:Borel.Cech}, a $1$-cochain 
$$
(\kappa_{ij}, \lambda_i) \in
\check{C}^{0, 1}(\U) \oplus \check{C}^{1, 0}(\U)
=
\check{C}^1(\U; \Z) \oplus \check{C}^0(\U; \Z)
$$
consists of locally constant functions $\kappa_{ij} : U_i \cap U_j \to \Z$ and $\lambda_i : U_i \to \Z$. The cocycle condition for $(\kappa_{ij}, \lambda_i)$ reads
\begin{align*}
(\delta \kappa)_{ijk} &= 0, &
(\delta \lambda)_{ij} &= \partial \kappa_{ij}, &
0 &= \partial \lambda_i.
\end{align*}
These conditions are equivalent to:
\begin{align}
\kappa_{jk}(x) - \kappa_{ik}(x) + \kappa_{ij}(x) &= 0, & 
& (x \in U_i \cap U_j \cap U_k)\nonumber \\
\lambda_j(x) - \lambda_i(x) &= \kappa_{ij}(x) - \kappa_{ij}(\iota(x)), &
& (x \in U_i \cap U_j) \nonumber \\
0 &= \lambda_i(x) + \lambda_i(\iota(x)). &
& (x \in U_i)\label{eqn:1cocycle.conditions}
\end{align}

\begin{lem}\label{lem:vanishing.equivariant.H1}
Let $X$ be a $\ZZ_2$-CW complex such that $H^1_{\Z_2}(X) = 0$. Suppose that we have:
\begin{enumerate}
\item
an invariant open cover $\U = \{ U_\mu \}_{\mu \in A}$ of $X$,

\item
locally constant functions $r_{\mu \nu} : U_\mu \cap U_\nu \to \Z$ for $\mu, \nu \in A$ such that 
\begin{itemize}
\item
$r_{\mu \xi}(x) = r_{\mu \nu}(x) + r_{\nu \xi}(x)$ for all $x \in U_{\mu} \cap U_{\nu} \cap U_{\xi}$ and $\mu, \nu, \xi \in A$;

\item
$r_{\mu \nu}(x) = r_{\mu \nu}(\iota(x))$ for all $x \in U_{\mu} \cap U_{\nu}$ and $\mu, \nu \in A$.

\end{itemize}
\end{enumerate}
Then there is an invariant open cover $\V = \{ V_i \}_{i \in I}$ which refines $\U$ and, we can find locally constant functions $\tau_i : V_i \to \mathbb{T}$ satisfying
\begin{align*}
\exp \pi i r_{ij}(x) &= \tau_j(x) \tau_i(x)^{-1}, &
& (x \in V_i \cap V_j) \\
\overline{\tau_i(x)} &= \tau_i (\iota(x)), &
& (x \in V_i)
\end{align*}
where $r_{ij} : V_i \cap V_j \to \Z$ are induced from $r_{\mu \nu}$ through the refinement.
\end{lem}

\begin{proof}
The locally constant functions $r_{\mu \nu}$ define a $1$-cocycle 
$$
(\kappa_{\mu \nu}, \lambda_\mu) = (r_{\mu \nu}, 0) 
\in \check{C}^{0, 1}(\U) \oplus \check{C}^{1, 0}(\U). 
$$
By the assumption $0 = H^1_{\Z_2}(X) \cong \varinjlim H^1(\check{C}^{*, *}(\U))$, there exists an invariant open cover $\V = \{ V_i \}_{i \in I}$ which refines $\U$ and $H^1(\check{C}^{*, *}(\V)) = 0$. Let $(r_{ij}, 0) \in \check{C}^{0, 1}(\V) \oplus \check{C}^{1, 0}(\V)$ be the induced $1$-cocycle. Because of the vanishing $H^1(\check{C}^{*, *}(\V)) = 0$, the $1$-cocyle is a coboundary, so that there exists 
$$
(m_i) \in \check{C}^{0, 0}(\V) = \check{C}^0(\V; \Z)
$$
such that 
\begin{align*}
(\delta m)_{ij} &= r_{ij}, &
\partial m_i &= 0.
\end{align*}
Thus, we have locally constant functions $m_i : V_i \to \Z$ for $i \in I$ such that
\begin{align*}
m_j(x) - m_i(x) &= r_{ij}(x), &
&(x \in V_i \cap V_j) \\
m_i(x) - m_i(\iota(x)) &= 0. &
&(x \in V_i)
\end{align*}
Now, $\tau_i : V_i \to \mathbb{T}$ is given by $\tau_i(x) = \exp \pi i m_i(x)$. \end{proof}

\begin{rem}\label{rem:vanishing.equivariant.H1}
If we assume that $X$ is connected, then the choice of $(m_i)$ can be constrained: Let $(m'_i)$ be another choice which cobounds the $1$-cocycle $(r_{ij}, 0)$. The difference cocycle $d_i = m'_i - m_i$ gives rise to global constant $d \in \Z$. Hence the difference $\tau'(x) / \tau(x) = \exp \pi i d = \pm 1$ is a global sign.
\end{rem}

\subsubsection{Other formulations of first equivariant cohomology}

Let $\TT_{\rm triv}$ be the unit circle in the complex plane equipped with the trivial involution, and $\TT_{\rm conj}$ be the one equipped with the complex conjugation involution.

\begin{lem}\label{lem:H1.equivariant}
Let $X$ be a $\ZZ_2$-CW complex. Then there are isomorphisms of groups
\begin{align*}
H^1_{\Z_2}(X) &\cong [X, \mathbb{T}_{\mathrm{triv}}]_{\Z_2},\\
H^1_{\pm}(X) &\cong [X, \mathbb{T}_{\mathrm{conj}}]_{\Z_2}.
\end{align*}
\end{lem}

\begin{proof}
A proof of $H^1_{\pm}(X) \cong [X, \mathbb{T}_{\mathrm{conj}}]_{\Z_2}$ is given in Prop.\ A.2 of \cite{Gomi}. The proof of the first statement is similar, and detailed below. Write $G=\ZZ_2$ for brevity. We have $H^1_G(X) = H^1_G(X; \Z)\cong\check{H}^1(G^\bullet \times X; \Z)$, and the long exact sequence
$$
\to
\check{H}^n(G^\bullet \times X; \Z) \to
\check{H}^n(G^\bullet \times X; \underline{\R}) \to
\check{H}^n(G^\bullet \times X; \underline{\mathbb{T}}) \overset{\delta}{\to}
\check{H}^{n+1}(G^\bullet \times X; \Z) \to
$$
induced by the short exact sequence of sheaves, Eq.\ \eqref{eqn:sheaf.sequence}.

For $n=0$, the zeroth \v{C}ech cohomology is the group of global sections of the coefficient sheaf. Since $G^p \times X$ is paracompact and $G$ is finite, $\check{H}^1(G^\bullet \times X; \underline{\R}) = 0$, thus we have
$$
C(X, \Z)^G \to
C(X, \R)^G \to
C(X, \mathbb{T})^G \overset{\delta}{\to}
H^1_G(X) \to
0,
$$
where $C(X, A)$ stands for the group of continuous functions $f : X \to A$, and $C(X, A)^G$ is the subgroup consisting of invariant functions: $f(gx) = f(x)$ for all $x \in X$ and $g \in G$. Note that the group of integers $A = \Z$ is given the discrete topology, while $A = \R, \mathbb{T}$ are given the standard topologies.

The homomorphism $\delta : C(X, \mathbb{T})^G \to H^1_G(X)$ descends to give a homomorphism $\delta : [X, \mathbb{T}_{\mathrm{triv}}]_G \to H^1_G(X)$. To see this, let $\pi : X \times [0, 1] \to X$ be the projection, and $i_t : X \to X \times [0, 1]$ the map $i_t(x) = (x, t)$ for each $t \in [0, 1]$. By the homotopy axiom, $i_t^* : H^1_G(X \times [0, 1]) \to H^1_G(X)$ is an isomorphism for any $t \in [0, 1]$ and its inverse is $\pi^* : H^1_G(X) \to H^1_G(X \times [0, 1])$. From the commutative diagram
$$
\begin{CD}
C(X \times [0, 1], \mathbb{T})^G @>{i_t^*}>> C(X, \mathbb{T})^G \\
@V{\delta}VV @VV{\delta}V \\
H^1_G(X \times [0, 1]) @>{i_t^*}>> H^1_G(X),
\end{CD}
$$
it follows that $\delta(i_0^*\tilde{f}) = i_0^* \delta(\tilde{f}) = i_1^*\delta(\tilde{f}) = \delta(i_1^*\tilde{f})$ for all $\tilde{f} \in C(X, \mathbb{T})^G$. This means that homotopic maps in $C(X, \mathbb{T})^G$ have the same image under $\delta$, and we get a well-defined $\delta : [X, \mathbb{T}_{\mathrm{triv}}]_G \to H^1_G(X)$.

Now, let $[f] \in [X, \mathbb{T}_{\mathrm{triv}}]_G$ be the element represented by $f \in C(X, \mathbb{T})^G$. Suppose that $\delta([f]) = \delta(f) = 0$. By the exact sequence, there is a continuous map $h : X \to \R$ such that $f(x) = \exp 2\pi i h(x)$ and $h(gx) = h(x)$ for all $x \in X$ and $g \in G$. If we define $\tilde{f} : X \times [0, 1] \to \mathbb{T}$ by $\tilde{f}(x, t) = \exp 2 \pi i t h(x)$, then $\tilde{f}$ is an equivariant homotopy between $f$ and the constant map $1$. Therefore the epimorphism 
\begin{equation*}
\delta : [X, \mathbb{T}_{\mathrm{triv}}]_G \to H^1_G(X)
\end{equation*}
has trivial kernel, thus it is an isomorphism.
\end{proof}

\begin{rem}
We can describe the homomorphism $\delta : [X, \mathbb{T}_{\mathrm{triv}}]_{\Z_2} \to H^1_{\Z_2}(X)$ explicitly. Let $\U = \{ U_i \}_{i \in I}$ be an invariant open cover. Taking a refinement if necessary, we can assume that $\U$ is such that $H^1(U_i) = 0$. Given $\varphi \in C(X, \mathbb{T})^{\Z_2}$, there are continuous maps $\tilde{\varphi}_i : U_i \to \R$ such that $\varphi(x) = \exp 2\pi \sqrt{-1}\tilde{\varphi}_i(x)$ for $x \in U_i$. Now, by setting 
\begin{align*}
\kappa_{ij}(x) &= \tilde{\varphi}_j(x) - \tilde{\varphi}_i(x), &
\lambda_i(x) &= \tilde{\varphi}_i(x) - \tilde{\varphi}_i(\iota(x)),
\end{align*}
we get locally constant functions $\kappa_{ij} : U_i \cap U_j \to \Z$ and $\lambda_i : U_i \to \Z$ satisfying the 1-cocycle conditions, Eq.\ \eqref{eqn:1cocycle.conditions}, thus giving a \v{C}ech $1$-cohomology class in $H^1_{\Z_2}(X)$. A similar construction applies to $[X, \mathbb{T}_{\mathrm{conj}}]_{\Z_2} \to H^1_\pm(X)$.
\end{rem}

\section{Equivariant cohomology of the fixed point set}\label{appendix:equiv.cohom.fixed}

\begin{lem} \label{lem:equivariant_cohomology_of_fixed_point_set}
Let $X$ be a space with involution $\iota : X \to X$. If the involution is trivial, $X^\iota = X$, then there is a natural isomorphism of groups
\begin{align*}
H^n_\pm(X) 
&\cong \bigoplus_{k \ge 0} H^{n - 2k - 1}(X; \Z_2) \\
&= H^{n-1}(X; \Z_2) \oplus H^{n-3}(X; \Z_2) \oplus H^{n-5}(X; \Z_2) \oplus 
\cdots.
\end{align*}
\end{lem}

\begin{proof}
The $\Z_2$-equivariant cohomology $H^n_\pm(X)$ with local coefficients can be identified with the $n$th cohomology of the total complex associated to the double complex $(C^{p, q}, \delta, \partial)$, where $C^{p, q} = C^q(X; \Z)$ is the (\v{C}ech) cochain complex of $X$ computing the integral cohomology of $X$, $\delta : C^{p, q} \to C^{p, q+1}$ its coboundary operator, and $\partial : C^{p, q} \to C^{p+1, q}$ the map defined by $\partial c = c + (-1)^p \iota^*c$. 

Under the assumption $X^\iota = X$, we have $\partial c = c + (-1)^p c$. This leads us to decompose the total complex $C^*$ of $C^{p, q}$ as
$$
C^* =
\bigoplus_{k \ge 0}
\mathrm{Cone}( 2 : C^*(X; \Z) \to C^*(X; \Z))[-2k-1],
$$
where $\mathrm{Cone}( f : D^* \to E^* )$ stands for the cone complex associated to a cochain map $f : D^* \to E^*$, and $[m]$ means the degree shift of a cochain complex by $m$. The cohomology of $\mathrm{Cone}(2 : C^*(X; \Z) \to C^*(X; \Z))$ is naturally isomorphic to the mod $2$ cohomology $H^*(X; \Z_2)$, and the proof is completed. 
\end{proof}

\begin{rem}
A similar proof also shows that
\begin{align*}
H^n_{\Z_2}(X) 
&\cong H^n(X; \Z) \oplus \bigoplus_{k \ge 1} H^{n - 2k}(X; \Z_2) \\
&= H^n(X; \Z) \oplus H^{n-2}(X; \Z_2) \oplus H^{n-4}(X; \Z_2) \oplus 
\cdots.
\end{align*}
\end{rem}

\subsection*{Acknowledgments}
KG is supprted by Japan JSPS KAKENHI Grant Numbers 20K03606 and JP17H06461. GCT is supported by Australian Research Council Discovery Projects Grant DP200100729, and thanks J. Kellendonk for helpful correspondence.

\end{document}